\documentclass[11pt,a4paper]{article}
\usepackage{amsfonts}
\usepackage[dvips]{graphicx}
\usepackage{amsmath}
\usepackage{mathrsfs}
\usepackage{makeidx}
\usepackage{xypic}
\usepackage[nottoc]{tocbibind}
\usepackage{pstricks}
\usepackage{bbm}
\usepackage[arrow, matrix, curve]{xy}
\usepackage{amsthm}
\usepackage{amssymb}
\usepackage{epic}
\usepackage{eepic}
\usepackage{times}
\usepackage{epsfig}

\xymatrixcolsep{0.5cm}
\xymatrixrowsep{0.5cm}
\newtheorem{theorem}{Theorem}[section]
\newtheorem{proposition}{Proposition}[section]
\newtheorem{lemma}{Lemma}[section]

\newcommand{\beqa}{\begin{eqnarray}}
\newcommand{\eeqa}{\end{eqnarray}}
\newcommand{\rf}[1]{(\ref{#1})}

\newcommand{\la}{\lambda}

\makeindex
\numberwithin{equation}{section}

\begin{document}

\bigskip \begin{flushright}
YITP-SB-12-25
\end{flushright}

\bigskip

\begin{center}
\textbf{\Large Antiperiodic dynamical 6-vertex model I:\\ Complete spectrum by
SOV, matrix elements of the identity\\ on separate states and connections to
the periodic 8-vertex model}

\vspace{50pt}

{\large G.~Niccoli}\footnote{%
YITP, Stony Brook University, New York 11794-3840, USA,
niccoli@max2.physics.sunysb.edu}

\vspace{50pt}

\vspace{50pt}
\end{center}

\begin{itemize}
\item[] {\small \textbf{Abstract}\thinspace \thinspace \thinspace The spin-1/2
highest weight representations of the dynamical 6-vertex and the standard
8-vertex Yang-Baxter algebra on a finite chain are considered in this paper.
In particular, the integrable quantum models associated to the corresponding transfer matrices under antiperiodic boundary conditions for the dynamical 6-vertex and periodic boundary conditions for the 8-vertex case are here analyzed.

For the antiperiodic dynamical 6-vertex transfer matrix defined on chains with an odd number of sites, we
adapt the Sklyanin's quantum separation of variable (SOV) method and explicitly construct the SOV representations from the original space of the
representations. In this way we provide the complete characterization of the eigenvalues
and the eigenstates proving also the simplicity of its spectrum. Moreover,
we characterize the matrix elements of the identity on separated states of
this model by determinant formulae. The matrices entering in these
determinants have elements given by sums over the SOV spectrum of the
product of the coefficients of the separate states. This SOV analysis is
done without any need to be reduced to the case of the so-called elliptic
roots of unit and the results here derived define the required setup to
extend to the dynamical 6-vertex model the approach recently developed in 
\cite{FarXYZGMN12-SG}-\cite{FarXYZN12-2} to compute the form factors of the
local operators in the SOV framework, these results will be presented in a
future publication.

For the periodic 8-vertex transfer matrix, we prove that its
eigenvalues have to satisfy a fixed system of equations. In the case of a
chain with an odd number of sites, this system of equations is the same
entering in the SOV characterization of the antiperiodic dynamical 6-vertex
transfer matrix spectrum. This implies that the set of the periodic 8-vertex
eigenvalues is contained in the set of the antiperiodic dynamical 6-vertex
eigenvalues. A criterion is introduced to find simultaneous
eigenvalues of these two transfer matrix and associate to any of such
eigenvalues one nonzero eigenstate of the periodic 8-vertex transfer matrix
by using the SOV results. Moreover, a preliminary discussion on the
degeneracy occurring for odd chains in the periodic 8-vertex transfer matrix
spectrum is also presented.}
\end{itemize}

\newpage
\tableofcontents \newpage

\section{Introduction}

In this paper we analyze two classes of lattice integrable quantum models
characterized in the quantum inverse scattering method (QISM) \cite%
{FarXYZSF78}-\cite{FarXYZIK82} by monodromy matrices which are solutions of
the dynamical Yang-Baxter equation w.r.t. the 6-vertex dynamical R-matrix
and of the (standard) Yang-Baxter equation w.r.t. the 8-vertex R-matrix,
respectively. The representation theory of the dynamical 6-vertex
Yang-Baxter algebra was introduced by Felder in \cite{FarXYZFelder94} by the
so-called theory of the elliptic quantum groups, see also \cite%
{FarXYZFelderV96}. There, it was recognized that the known Boltzmann weights
defining the SOS (solid on solid) statistical models \cite{FarXYZBa72-2}
when opportunely reorganized in a $4\times 4$ matrix define the R-matrix
solution of the dynamical 6-vertex Yang-Baxter algebra. The prototypical
elements in these classes of integrable quantum models are constructed by
defining representations of the corresponding monodromy matrices on chains
of 2-dimensional representations (spin-1/2 quantum chains). Under
homogeneous limits these representations define in the dynamical 6-vertex
case the SOS model \cite{FarXYZBa72-2}, \cite{FarXYZBa72-1}-\cite{FarXYZJMO}
while in the 8-vertex case the spin-1/2 XYZ quantum chain \cite{FarXYZH28}-%
\cite{FarXYZLM66}. It is worth recalling that the monodromy matrices of these models are related by the Baxter's {\it intertwining vectors} and the spectral problems of the transfer matrices under periodic boundary conditions have been analyzed by Bethe ansatz and Q-operator techniques\footnote{For the Q-operator construction see \cite{FarXYZBa72-2,FarXYZBa72} and also
the series of papers \cite{FarXYZFMcCoy03}-\cite{FarXYZFMcCoy08}.} in the
case of chains with an even number of sites.  In \cite{FarXYZBa72-1, FarXYZBa72-2, FarXYZBa72-3}, R. Baxter has introduced the intertwining vectors, also called gauge transformations\footnote{Let us comment that historically R. Baxter has used a vectorial representation for these transformations, e.g. see equation (3.3) of \cite{FarXYZBa72-2}, which explains the original use of the terminology intertwining vectors. Here, we use the terminology gauge transformations, also used in \cite{FarXYZFT79}, to refer to the matrix representation of the same transformations presented in \rf{8V-6VD-GT0}.}, in order to be able to use Bethe ansatz techniques to
analyzed the spectral problem (eigenvalues \& eigenstates) of the 8-vertex
transfer matrix reducing it to one of 6-vertex type. The use of gauge
transformations allows in particular to overcome the problem of the absence
of reference states opening the possibility to analyze the 8-vertex spectral
problem by using the algebraic Bethe ansatz (ABA) \cite{FarXYZSF78}-\cite%
{FarXYZFST80}, as pioneered by Faddeev and Takhtajan in \cite{FarXYZFT79}
while ABA analysis for the SOS model with periodic boundary conditions has
been developed in \cite{FarXYZFelderVa96}. However, it is worth remarking
that, a part the general problem
related to the proof of the completeness of the spectrum description\footnote{%
In fact, for the periodic 8-vertex transfer matrices the completeness of the
spectrum description is verified only by some numerical analysis \cite%
{FarXYZBa02}.}, the analysis by Bethe ansatz methods of the spectrum of these models
leads to the introduction of two constrains. The first constrain is on the number of sites of the quantum
chains which has to be even. This is required to obtain the commutativity of
the dynamical 6-vertex transfer matrix which holds only for the reduction to
the total spin zero-eigenspace under periodic boundary conditions. The
second constrain is on the allowed values of the coupling constant $\eta $
of the 8-vertex transfer matrix which has to be restricted to the so-called 
\textit{cyclic values or elliptic roots of unit} (i.e. when $\eta $ belongs
to an integer square lattice with steps the periods of the theta functions).
This is required in order to construct 8-vertex transfer matrix eigenstates
by finite sums of the dynamical 6-vertex ones. In addition to the previously
described constrains in the algebraic Bethe ansatz framework the lack of a
scalar product analogue to the Slavnov's formula \cite%
{FarXYZSlav89,FarXYZSlav97,FarXYZKitMT99} is the first fundamental missing
step toward the computation of matrix elements of local operators. It is
then clear the need to overcome these problems in order to compute
correlation functions.

In the present paper we implement a modified version of Sklyanin's quantum
separation of variables (SOV) \cite{FarXYZSk1}-\cite{FarXYZSk3}. In
particular, we derive the complete characterizations of the spectrum of the
antiperiodic\footnote{This quantum integrable model has been introduced in \cite{FarXYZDJMO}.}
dynamical 6-vertex transfer matrix defined on chains with an odd number of sites. Moreover, we compute the matrix elements of the
identity for general \textit{separate states}\footnote{See Section \ref{FarXYZseparate-states} for the definition.} which apply in
particular for the eigenstates of the antiperiodic dynamical 6-vertex
transfer matrix. Let us comment that the existing results \cite%
{FarXYZFelder-Schorr-99,FarXYZSchorr} for the antiperiodic dynamical
6-vertex model are mainly restricted to the construction of the functional
separation of variables of Sklyanin. In this functional version an SOV
representation of the dynamical 6-vertex Yang-Baxter algebra is defined on a
space of symmetric functions leading only to the description of the wave
functions of the transfer matrix eigenstates. In fact, the explicit
construction of the SOV representation as well as of the transfer matrix
eigenstates in the original representation space of the quantum chain were so far missing.

For the periodic 8-vertex transfer matrix, we prove that the set of
all its eigenvalues is contained in the set of the solutions to an
inhomogeneous system of $\mathsf{N}$ quadratic equations in $\mathsf{N}$
unknown, where $\mathsf{N}$ is the number of sites of the chain. In the case $%
\mathsf{N}$ odd this system coincides with the one entering in the SOV
characterization of the antiperiodic dynamical 6-vertex transfer matrix
spectrum and so the set of the periodic 8-vertex eigenvalues is proven to be
contained in the antiperiodic dynamical 6-vertex one. Let us recall that the
analysis of the odd $\mathsf{N}$ case is of particular interest for the
periodic 8-vertex transfer matrix as in this case the Bethe ansatz analysis
of \cite{FarXYZBa72-1, FarXYZBa72-2, FarXYZBa72-3,FarXYZFT79} does not
apply. We use the Baxter's gauge transformations to further relate the
periodic 8-vertex and the antiperiodic dynamical 6-vertex transfer matrices,
allowing to analyze the 8-vertex spectral problem in terms of the SOV
characterization here derived. It is worth to stress that in the \textit{%
dynamical quantum space }$\mathbb{\bar{D}}_{\mathsf{(6VD)},\mathsf{N}}$,
characterized by the condition that the antiperiodic dynamical 6-vertex
transfer matrix is a one parameter family of commuting operators, these
gauge transformations are not invertible operators. Nevertheless, we are able
to use them to get a sufficient criterion which allows us to select
simultaneous eigenvalues of the antiperiodic dynamical 6-vertex and the
periodic 8-vertex transfer matrix and to associate to any one of these
eigenvalues one corresponding nonzero 8-vertex eigenstate. In the paper we
will explain as the non invertible character of the gauge transformations on
the dynamical quantum space $\mathbb{\bar{D}}_{\mathsf{(6VD)},%
\mathsf{N}}$ is a natural requirement as these transformations link transfer
matrices with different degeneracy properties. Indeed, while the
antiperiodic dynamical 6-vertex transfer matrix is proven here to have
simple spectrum the periodic 8-vertex one has degeneracy even for completely
general inhomogeneities. A preliminary analysis of this degeneracy issue is
here presented by analyzing explicitly the periodic 8-vertex transfer matrix
spectrum for chains with one and three sites.

The results derived in the present paper represent the first fundamental
step in an approach to solve integrable quantum models which can be
considered as the generalization to the SOV framework of the Lyon group
method \cite{FarXYZKitMT99}, \cite{FarXYZMaiT00}-\cite{FarXYZKKMNST08}. The
use of SOV is a strength point of our approach as it works for a large class
of integrable quantum models, under simple conditions it leads to the
complete construction of both the eigenvalues and the eigenstates of the
transfer matrix and the simplicity of the spectrum can be easily shown in
this framework. Moreover, the analysis developed in the present paper and
that implemented previously in \cite{FarXYZGMN12-SG}-\cite{FarXYZN12-2}
suggest that this approach can lead to an universal representation of both
the spectrum and the dynamics of a class of integrable quantum models which
were not entirely solvable with other methods\footnote{%
Like the algebraic Bethe ansatz, the coordinate Bethe ansatz \cite%
{FarXYZBe31}, \cite{FarXYZBaxBook} and \cite{FarXYZABBBQ87}, the Baxter
Q-operator method \cite{FarXYZBaxBook} and the analytic Bethe ansatz \cite%
{FarXYZRe83-1}-\cite{FarXYZRe83-2}.}. Indeed, this is the case for all the
key integrable quantum models analyzed so far in \cite{FarXYZGMN12-SG}-\cite%
{FarXYZN12-2}. More in detail, in \cite{FarXYZN12-0} and \cite{FarXYZN12-1},
the XXZ spin-1/2 quantum chain\footnote{%
Let us comment that previous results on this model with antiperiodic
boundary conditions were mainly given by the Q-operator construction \cite%
{FarXYZBBOY95} and the functional separation of variables \cite%
{FarXYZSk2,FarXYZNWF09}. Moreover see \cite{FarXYZG08} for the eigenvalue
analysis of the XXZ spin-1/2 chain with general antiperiodic boundary
conditions by a functional method based on the Yang-Baxter algebra; method
also applies to open chain for general integrable boundary conditions.} \cite%
{FarXYZH28}-\cite{FarXYZLM66} and the higher spin-s XXX quantum chain%
\footnote{%
Instead in the periodic chain matrix elements of local operators were
compute in the ABA framework in \cite{FarXYZK01,FarXYZCM07}.}, both under
antiperiodic boundary conditions, have been characterized and the form
factors of the local spin operators have been represented in a determinant
form. Similar results for the form factors have been obtained previously by
this approach in \cite{FarXYZGMN12-SG,FarXYZGMN12-T2} for the lattice
quantum sine-Gordon model \cite{FarXYZFST80,FarXYZIK82}, the chiral Potts
model \cite{FarXYZBS90}-\cite{FarXYZTarasovSChP} and the $\tau _{2}$-model 
\cite{FarXYZBa04}. These results are obtained by using as background the
complete SOV spectrum characterization constructed in \cite{FarXYZNT-10}-%
\cite{FarXYZN-11} for the lattice quantum sine-Gordon model and in\footnote{%
See \cite{FarXYZGIPS06}-\cite{FarXYZGIPS09} for a first analysis by SOV
method of the $\tau _{2}$-model and some results on form factors in the
restricted case of the generalized Ising model.} \cite{FarXYZGN12} for the
general cyclic representations of the 6-vertex Yang-Baxter algebra
corresponding to the $\tau _{2}$-model and the chiral Potts model. Moreover,
in \cite{FarXYZN12-2} the SOV setup has been implemented and the matrix
elements of some interesting quasi-local string of local operators have been
computed for the integrable quantum model associated to the spin-1/2
representations of the reflection algebra \cite{FarXYZGau71}-\cite%
{FarXYZGZ94}, under quite general non-diagonal boundary conditions. In all
these models the matrix elements of local operators on \textit{separate
states} are characterized by determinant formulae written as simple
modifications of those of the identity. The main differences in all these
formulae are only due to model dependent features, like the nature of the
spectrum of the quantum separate variables and the form of the SOV
reconstruction of local operators. Let us comment that in the literature
there exist previous results on matrix elements of local operators which
even if developed by different approaches made use of quantum separation of
variables. The results in Smirnov's paper \cite{FarXYZSm98} are of special
interest; there, for the quantum integrable Toda chain \cite{FarXYZSk1}, the
form factors of a conjectured\footnote{%
There this conjecture is required by the absence of a direct SOV
reconstruction of local operators. Later a reconstruction has been given in 
\cite{FarXYZOB-04} w.r.t. a new set of quantum separate variables defined by
a change of variables on the original Sklyanin's ones.} basis of local
operators are derived in Sklyanin's SOV framework by determinant formulae
which confirm the universal picture outlined. It is also worth pointing out
that the form factors\footnote{%
Form factors which have been rederived in \cite{FarXYZJMS11-03} also by
exploiting previous results established in the series of papers \cite%
{FarXYZJMS09-02}-\cite{FarXYZJMS11-02} for the infinite volume limit of the
XXZ spin-1/2 chain.} of the restricted sine-Gordon model at the
reflectionless points in the S-matrix formulation\footnote{%
See \cite{FarXYZA.Zam77}-\cite{FarXYZM92} and references therein.} \cite%
{FarXYZBBS96,FarXYZBBS97} admit once again determinant representations and
the connection with SOV is established on the basis of the semi-classical
analysis of \cite{FarXYZBBS96}, used there also as a tool to overcome the
problem\footnote{%
This is a longstanding problem in the S-matrix formulation which has been so
far addressed by exploiting the description \cite{FarXYZZam88}-\cite%
{FarXYZGM96} of massive IQFTs as (superrenormalizable) perturbations of
conformal field theories \cite{FarXYZVi70}-\cite{FarXYZDFMS97} by relevant
local fields. Several results are known which allows to classify the local
fields of massive theories (i.e. the solutions to the form factor equations 
\cite{FarXYZKW78}-\cite{FarXYZD04}) in terms of those of the ultraviolet
conformal field theories, see for example \cite{FarXYZCM90}-\cite%
{FarXYZJMT03} and the series of works \cite{FarXYZDN05-1}-\cite{FarXYZD09}.}
of the local field identification.

\section{The dynamical 6-vertex models: spectrum and elementary matrix
elements}

\subsection{The dynamical 6-vertex models}

In the following, we introduce an operator $\tau $ whose eigenvalues on the
space of the representation coincide with the dynamical parameter $t$. The
aim is to recover a separate description for the dynamical parameter which
will be particularly useful in the SOV description of the antiperiodic
dynamical 6-vertex spectral problem.

\subsubsection{Representation spaces of dynamical and spin operators\label{FarXYZRep-space}}

Let us introduce a couple of \textit{dynamical}\ operators $\tau $ and $\mathsf{T}_{\tau }^{\pm }$ which satisfy the following commutation relations:
\begin{equation}
\mathsf{T}_{\tau }^{\pm }\tau =(\tau \pm \eta )\mathsf{T}_{\tau }^{\pm },
\label{FarXYZDyn-op-comm}
\end{equation}%
and $\mathsf{N}$ copies of (local spin) $sl(2)$ generators $S_{n}^{a}$ and
let us impose the following commutation relations:%
\begin{equation}
\lbrack S_{n}^{a},S_{m\neq n}^{b}]=[S_{n}^{a},\tau ]=[S_{n}^{a},\mathsf{T}%
_{\tau }^{\pm }]=0\text{ \ \ }\forall n,m\in \{1,...,\mathsf{N}\}\text{ and }%
a,b=x,y,z.
\end{equation}%
Then the space of the representation of these dynamical and spin operators
can be chosen as it follows:%
\begin{equation}
\mathbb{D}_{\mathsf{(6VD)},\mathsf{N}}\equiv \mathbb{D}_{\mathsf{N}}\otimes 
\underset{\mathsf{N}}{\underbrace{\mathbb{C}^{2}\otimes \otimes \otimes 
\mathbb{C}^{2}}},
\end{equation}%
where $\mathbb{D}_{\mathsf{N}}$\ is the space of the representation of the
dynamical operators which is infinite dimensional in our definition\footnote{%
Note that in the root of unit case, we can define also a cyclic
representation for the operator $\tau $. This point in the present formalism
will be described elsewhere anyhow for the antiperiodic chain this cyclicity
condition does not play a fundamental role as it was in the periodic case
for the application of algebraic Bethe ansatz to the 8-vertex model.}. Here,
we have chosen for all the $sl(2)$ generators the spin-1/2 representation;
i.e. a 2-dimensional local quantum space $\mathbb{C}^{2}$ is associated to
any site of the chain and the local spin generators are represented by the $%
2\times 2$ Pauli matrices $\sigma _{n}^{a}$.

Moreover, we introduce the following definition of left (covectors) and
right (vectors) $\tau $-eigenbasis of $\mathbb{D}_{\mathsf{N}}^{\mathcal{L}}$
and $\mathbb{D}_{\mathsf{N}}^{\mathcal{R}}$, respectively:%
\begin{equation}
\langle t(a)|\equiv \langle t(0)|\mathsf{T}_{\tau }^{a},\text{ \ \ \ }%
|t(a)\rangle \equiv \mathsf{T}_{\tau }^{-a}|t(0)\rangle ,\text{ \ }\forall
a\in \mathbb{Z}\text{,}  \label{FarXYZt-dyn-sp}
\end{equation}%
with:%
\begin{equation}
\langle t(a)|\tau =t(a)\langle t(a)|,\text{ \ \ \ }\tau |t(a)\rangle
=t(a)|t(a)\rangle \text{,\ \ \ \ \ }t(a)\equiv -\frac{\eta }{2}a\text{\ \ \ }%
\forall a\in \mathbb{Z},
\end{equation}%
where we fix $\langle t(a)|t(b)\rangle =\delta _{a,b}\text{ \ },\forall a,b\in \mathbb{Z}%
\text{.}$ Moreover, defined the following left and right spin basis:%
\begin{equation}
\langle n,h_{n}|\sigma _{n}^{z}=(1-2h_{n})\langle n,h_{n}|,\text{ \ \ \ \ \ }%
\sigma _{n}^{z}|n,h_{n}\rangle =(1-2h_{n})|n,h_{n}\rangle ,\text{ \ \ }%
h_{n}\in \{0,1\},
\end{equation}%
with $\langle n,h_{n}|n,h_{n}^{\prime }\rangle =\delta _{h_{n},h_{n}^{\prime }}%
\text{,}$ in each local quantum spin chain of the representation, we can introduce the
left and right dynamical-spin basis in $\mathbb{D}_{\mathsf{(6VD)},\mathsf{N}%
}^{\mathcal{L}}$ and $\mathbb{D}_{\mathsf{(6VD)},\mathsf{N}}^{\mathcal{R}}$,
respectively:%
\begin{equation}
\otimes _{n=1}^{\mathsf{N}}\langle n,h_{n}|\otimes \langle t(a)|,\text{ \ \
\ \ \ \ \ }\otimes _{n=1}^{\mathsf{N}}|n,h_{n}\rangle \otimes |t(a)\rangle ,
\end{equation}%
composed of common eigenstates of the commuting operators $\tau $ and $%
\sigma _{n}^{z}$. A scalar product is introduced in the space $\mathbb{D}_{%
\mathsf{(6VD)},\mathsf{N}}^{\mathcal{R}}$ by defining its action on the
elements of the dynamical-spin basis: 
\begin{equation}
(\otimes _{n=1}^{\mathsf{N}}|n,h_{n}\rangle \otimes |t(a)\rangle ,\otimes
_{n=1}^{\mathsf{N}}|n,h_{n}^{\prime }\rangle \otimes |t(a^{\prime })\rangle
)=\delta _{a,a^{\prime }}\prod_{n=1}^{\mathsf{N}}\delta
_{h_{n},h_{n}^{\prime }}.
\end{equation}%
Note that we have defined the representation in a way that the spectrum
(eigenvalues) of the operator $\tau $ contains that of $-\eta \mathsf{S}$/2
where:%
\begin{equation}
\text{$\mathsf{S}$}=\sum_{n=1}^{\mathsf{N}}\sigma _{n}^{z}
\label{FarXYZDef-S}
\end{equation}%
is\ the total $z$-component of the spin; the reason for that will be clear
in the following.

\subsubsection{Representations of the dynamical 6-vertex models}

Let us define the elliptic dynamical 6-vertex R-matrix\footnote{%
The presentation of our results will be done directly in this elliptic case,
however, it is interesting to remark that also the trigonometric case
corresponding to the following choice of dynamical R-matrix: 
\begin{equation*}
a(\lambda )=\sinh (\lambda +\eta ),\quad b(\lambda |\tau )=\frac{\sinh
\lambda \sinh (\tau +\eta )}{\sinh \tau },\quad c(\lambda |\tau )=\frac{%
\sinh \eta \sinh (\tau +\lambda )}{\sinh \tau },
\end{equation*}%
can be similarly described in our approach.}:$\allowbreak $%
\begin{equation}
R_{0a}^{\mathsf{(6VD)}}(\lambda |\tau )=\left( 
\begin{array}{cccc}
a(\lambda ) & 0 & 0 & 0 \\ 
0 & b(\lambda |\tau ) & c(\lambda |\tau ) & 0 \\ 
0 & c(\lambda |-\tau ) & b(\lambda |-\tau ) & 0 \\ 
0 & 0 & 0 & a(\lambda )%
\end{array}%
\right) ,  \label{FarXYZop-L}
\end{equation}%
where $a(\lambda )$, $b(\lambda |\tau )$ and $c(\lambda |\tau )$ are defined
by:%
\begin{equation}
a(\lambda )=\theta (\lambda +\eta ),\quad b(\lambda |\tau )=\frac{\theta
(\lambda )\theta (\tau +\eta )}{\theta (\tau )},\quad c(\lambda |\tau )=%
\frac{\theta (\eta )\theta (\tau +\lambda )}{\theta (\tau )},
\label{FarXYZDef-ell-abc}
\end{equation}%
here and in the following we use the notation:%
\begin{equation}
\theta (\lambda )=\theta _{1}(\lambda |\omega ),
\end{equation}%
where\footnote{%
In this paper the $\theta _{i}(\lambda |\omega )$ for $i\in \{1,..4\}$ are
the standard theta functions as defined for example at page 877 of \cite%
{FarXYZTables of integrals}, where we are using for the argument of these
functions $(\lambda |\omega )$ instead of $(u|\tau )$.} $\theta _{1}(\lambda
|\omega )$ is the standard theta-1 elliptic function of modular parameter $%
\omega $. Then, the R-matrix $R_{1,2}^{\mathsf{(6VD)}}(\lambda |\tau )$ is
solution of the following dynamical Yang-Baxter equation:%
\begin{equation}\label{6VD-YBeq0}
R_{1,2}^{\mathsf{(6VD)}}(\lambda _{12}|\tau +\eta \sigma _{a}^{z})R_{1,a}^{%
\mathsf{(6VD)}}(\lambda _{1}|\tau )R_{2,a}^{\mathsf{(6VD)}}(\lambda
_{2}|\tau +\eta \sigma _{1}^{z})=R_{2,a}^{\mathsf{(6VD)}}(\lambda _{2}|\tau
)R_{1,a}^{\mathsf{(6VD)}}(\lambda _{1}|\tau +\eta \sigma _{2}^{z})R_{1,2}^{%
\mathsf{(6VD)}}(\lambda _{12}|\tau ),
\end{equation}%
where $\lambda _{12}\equiv \lambda _{1}-\lambda _{2}$. It is possible to
introduce the following dynamical 6-vertex monodromy matrix:%
\begin{equation}
\mathsf{M}_{0}^{\mathsf{(6VD)}}(\lambda |\tau )\equiv R_{0,\mathsf{N}}^{%
\mathsf{(6VD)}}(\lambda-\xi
_{\mathsf{N}} |\tau +\eta \sum_{a=1}^{\mathsf{N}-1}\sigma
_{a}^{z})\cdots R_{0,1}^{\mathsf{(6VD)}}(\lambda-\xi
_{1} |\tau )\equiv \left( 
\begin{array}{cc}
\mathsf{A}(\lambda |\tau ) & \mathsf{B}(\lambda |\tau ) \\ 
\mathsf{C}(\lambda |\tau ) & \mathsf{D}(\lambda |\tau )%
\end{array}%
\right) _{0}\text{,}  \label{FarXYZ6VD-Monodromy}
\end{equation}
where the $\xi
_{n}$ for $n\in\{1,...,-{\mathsf{N}}\}$ are parameters of the model called inhomogeneities. Then this monodromy matrix is a solution of the same type of dynamical Yang-Baxter equation:%
\begin{equation}
R_{1,2}^{\mathsf{(6VD)}}(\lambda _{12}|\tau +\eta \text{$\mathsf{S}$})%
\mathsf{M}_{1}^{\mathsf{(6VD)}}(\lambda _{1}|\tau )\mathsf{M}_{2}^{\mathsf{%
(6VD)}}(\lambda _{2}|\tau +\eta \sigma _{1}^{z})=\mathsf{M}_{2}^{\mathsf{%
(6VD)}}(\lambda _{2}|\tau )\mathsf{M}_{1}^{\mathsf{(6VD)}}(\lambda _{1}|\tau
+\eta \sigma _{2}^{z})R_{1,2}^{\mathsf{(6VD)}}(\lambda _{12}|\tau ),
\label{FarXYZDyn-YB-Eq0}
\end{equation}%
where $\mathsf{S}$ is\ the total $z$-component of the spin defined in (\ref%
{FarXYZDef-S}). Moreover, it is worth remarking that the following commutation relations
hold:%
\begin{equation}
\lbrack \mathsf{A}(\lambda |\tau ),\tau ]=[\mathsf{B}(\lambda |\tau ),\tau
]=[\mathsf{C}(\lambda |\tau ),\tau ]=[\mathsf{D}(\lambda |\tau ),\tau ]=0,
\label{FarXYZss-ABCD-tau-com}
\end{equation}%
and%
\begin{equation}
\lbrack \mathsf{A}(\lambda |\tau ),\mathsf{S}]=[\mathsf{D}(\lambda |\tau ),%
\mathsf{S}]=0,\text{ \ }[\mathsf{C}(\lambda |\tau ),\mathsf{S}]=-2\mathsf{C}%
(\lambda |\tau ),\text{ \ \ }[\mathsf{B}(\lambda |\tau ),\mathsf{S}]=2%
\mathsf{B}(\lambda |\tau ).  \label{FarXYZXsc-S-spin}
\end{equation}%
Note that defined:%
\begin{equation}
\mathsf{T}_{\tau }^{\pm \sigma _{a}^{z}}\equiv \left( 
\begin{array}{cc}
\mathsf{T}_{\tau }^{\pm } & 0 \\ 
0 & \mathsf{T}_{\tau }^{\mp }%
\end{array}%
\right) _{a},
\end{equation}%
we can rewrite the dynamical Yang-Baxter equations in the following form:%
\begin{equation}
R_{1,2}^{\mathsf{(6VD)}}(\lambda _{12}|\tau +\eta \text{$\mathsf{S}$})%
\mathsf{M}_{1}^{\mathsf{(6VD)}}(\lambda _{1}|\tau )\mathsf{T}_{\tau
}^{\sigma _{1}^{z}}\mathsf{M}_{2}^{\mathsf{(6VD)}}(\lambda _{2}|\tau )%
\mathsf{T}_{\tau }^{-\sigma _{1}^{z}}\left. =\right.\mathsf{M}_{2}^{\mathsf{%
(6VD)}}(\lambda _{2}|\tau )\mathsf{T}_{\tau }^{\sigma _{2}^{z}}\mathsf{M}%
_{1}^{\mathsf{(6VD)}}(\lambda _{1}|\tau )\mathsf{T}_{\tau }^{-\sigma
_{2}^{z}}R_{1,2}^{\mathsf{(6VD)}}(\lambda _{12}|\tau ).
\end{equation}%
Let us remark that while the dynamical 6-vertex generators $\mathsf{B}%
(\lambda |\tau )$ and $\mathsf{C}(\lambda |\tau )$ are nilpotent operators
of order $\mathsf{N}+1$, for a chain of size $\mathsf{N}$, the generators $%
\mathsf{A}(\lambda |\tau )$ and $\mathsf{D}(\lambda |\tau )$ are not
nilpotent operators. So, from the form of the dynamical 6-vertex Yang-Baxter
commutation relations, it is clear that the set spanned by the dynamical
parameter $t$ (the eigenvalues of $\tau $) is always an infinite lattice of
step $\eta $. We can also define the following monodromy matrix:%
\begin{equation}
\left( 
\begin{array}{cc}
\mathcal{A}(\lambda |\tau ) & \mathcal{B}(\lambda |\tau ) \\ 
\mathcal{C}(\lambda |\tau ) & \mathcal{D}(\lambda |\tau )%
\end{array}%
\right) _{0}=\mathcal{M}_{0}^{\mathsf{(6VD)}}(\lambda |\tau )\equiv \mathsf{M%
}_{0}^{\mathsf{(6VD)}}(\lambda |\tau )\text{ }\mathsf{T}_{\tau }^{\sigma
_{0}^{z}},
\end{equation}%
note that the $\mathcal{X}(\lambda |\tau )$ ($\mathcal{X}=\mathcal{A},$ $%
\mathcal{B},$ $\mathcal{C}$ and $\mathcal{D}$) are operator functions of $%
\tau $ and \textsf{$T$}$_{\tau }^{\pm }$ but for simplicity we omit the
explicit dependence from \textsf{$T$}$_{\tau }^{\pm }$ in their arguments.
For this monodromy matrix the dynamical Yang-Baxter equation reads:%
\begin{equation}
R_{1,2}^{\mathsf{(6VD)}}(\lambda _{12}|\tau +\eta \text{$\mathsf{S}$})%
\mathcal{M}_{1}^{\mathsf{(6VD)}}(\lambda _{1}|\tau )\text{ }\mathcal{M}_{2}^{%
\mathsf{(6VD)}}(\lambda _{2}|\tau )\text{ }\left. =\right. \mathcal{M}_{2}^{%
\mathsf{(6VD)}}(\lambda _{2}|\tau )\text{ }\mathcal{M}_{1}^{\mathsf{(6VD)}%
}(\lambda _{1}|\tau )\text{ }R_{1,2}^{\mathsf{(6VD)}}(\lambda _{12}|\tau ),
\end{equation}%
where we have used that:%
\begin{equation}
\mathsf{T}_{\tau }^{-\sigma _{1}^{z}}\mathsf{T}\text{$_{\tau }^{-\sigma
_{2}^{z}}$ }R_{1,2}^{\mathsf{(6VD)}}(\lambda _{12}|\tau )\mathsf{T}_{\tau
}^{\sigma _{1}^{z}}\mathsf{T}\text{$_{\tau }^{\sigma _{2}^{z}}=$}R_{1,2}^{%
\mathsf{(6VD)}}(\lambda _{12}|\tau ).  \label{FarXYZComm-R12}
\end{equation}
Finally, let us comment that in \cite{FarXYZFelderV96} it was shown that the dynamical
6-vertex Yang-Baxter equations \rf{6VD-YBeq0} are just the
rewriting of the Baxter's star-triangle equations for the Boltzmann weights:
\begin{small}
\begin{eqnarray}
&&W\left[ \left. 
\begin{array}{cc}
t & t+1 \\ 
t+1 & t+2%
\end{array}\right\vert \lambda \right]=a\left( \lambda \right),\text{\ }W\left[
\left. 
\begin{array}{cc}
t & t+1 \\ 
t-1 & t%
\end{array}%
\right\vert \lambda \right] =b\left( \lambda |t\right),\text{\ }W\left[
\left. 
\begin{array}{cc}
t & t+1 \\ 
t+1 & t%
\end{array}%
\right\vert \lambda \right] =c\left( \lambda |-t\right),\,\,\,\,\,\,\,\,\,\,  \label{BWSOS-1} \\
&&W\left[ \left. 
\begin{array}{cc}
t & t-1 \\ 
t-1 & t.2%
\end{array}%
\right\vert \lambda \right]=a\left( \lambda \right),\text{\ }W\left[
\left. 
\begin{array}{cc}
t & t-1 \\ 
t+1 & t%
\end{array}%
\right\vert \lambda \right] =b\left( \lambda |-t\right) ,\text{ \ }W\left[
\left. 
\begin{array}{cc}
t & t-1 \\ 
t-1 & t%
\end{array}%
\right\vert \lambda \right] =c\left( \lambda |t\right),\,\,\,\,\,\,\,\,\,\,  \label{BWSOS-2}
\end{eqnarray}
\end{small}of the solid-on-solid (SOS) model\footnote{%
This model of statistical mechanics is defined on a square lattice and to
each site $n$ a "height" $l_{n}$ is associated and the interactions are
defined round each face (composed by 4 adjacent sites) of the lattice. These
interactions are nonzero only for adjacent heights which differ by 1 and are
described by the Boltzmann weights $W\left[ \left. {}\right\vert \lambda %
\right] $ of equations $\left( \ref{BWSOS-1}\right) $ and $\left( \ref%
{BWSOS-2}\right) $.}, where on the r.h.s. of the above equations
there are the entries of the dynamical 6-vertex R-matrix.
\subsubsection{Quantum determinant}

A fundamental object to define in the dynamical 6-vertex Yang-Baxter algebra
is the so-called quantum determinant\footnote{%
See \cite{FarXYZIK81} and the historical note \cite{FarXYZIK09} for a first
proof of the centrality of the quantum determinant in the Yang-Baxter
algebra.}. In particular, it plays a fundamental role in the construction of
the quantum separation of variables for these algebra as we will explain in
the following.

\begin{proposition}
In the dynamical 6-vertex Yang-Baxter algebra, we can introduce the
following central quantum determinant:%
\begin{align}
\det{}_{q}\mathsf{M}(\lambda )& \equiv \frac{\theta (\tau +\eta 
\text{$\mathsf{S}$})}{\theta (\tau )}\left( \mathsf{A}(\lambda |\tau )%
\mathsf{D}(\lambda -\eta |\tau +\eta )-\mathsf{B}(\lambda |\tau )\mathsf{C}%
(\lambda -\eta |\tau -\eta )\right) \\
& =\frac{\theta (\tau +\eta \text{$\mathsf{S}$})}{\theta (\tau )}\left( 
\mathsf{D}(\lambda |\tau )\mathsf{A}(\lambda -\eta |\tau -\eta )-\mathsf{C}%
(\lambda |\tau )\mathsf{B}(\lambda -\eta |\tau +\eta )\right) \\
& =\text{\textsc{a}}(\lambda )\text{\textsc{d}}(\lambda -\eta ),
\end{align}%
where:%
\begin{equation}
\text{\textsc{a}}(\lambda )\equiv \prod_{n=1}^{\mathsf{N}}a(\lambda -\xi
_{n}),\quad \text{\textsc{d}}(\lambda )\equiv \text{\textsc{a}}(\lambda
-\eta ).
\end{equation}%
Moreover, the following inversion formula holds:%
\begin{equation}
\mathsf{M}_{0}^{\mathsf{(6VD)}}(\lambda|\tau )\left( 
\begin{array}{cc}
\mathsf{D}(\lambda-\eta|\tau +\eta ) & -\mathsf{B}(\lambda-\eta |\tau +\eta )
\\ 
-\mathsf{C}(\lambda-\eta |\tau -\eta ) & \mathsf{A}(\lambda-\eta |\tau -\eta
)%
\end{array}%
\right) _{0}\frac{\theta (\tau +\eta \text{$\mathsf{S}$})/\theta (\tau )}{%
\det{}_{q}\mathsf{M}(\lambda )}=\left( 
\begin{array}{cc}
1 & 0 \\ 
0 & 1%
\end{array}%
\right) _{0} .  \label{FarXYZRight-1-dyn-Mon}
\end{equation}
\end{proposition}

\begin{proof}
The proof follows by proving the statement for the generic quantum site $n$
and then showing that the product of the local quantum determinants
reproduce the complete one. Let us introduce the notation:%
\begin{align}
\det{}_{q}R_{0n}^{\mathsf{(6VD)}}(\lambda |\tau )& =\left( R_{0n}^{\mathsf{(6VD)}%
}\right) _{11}(\lambda |\tau )\left( R_{0n}^{\mathsf{(6VD)}}\right)
_{22}(\lambda -\eta |\tau +\eta )  \notag \\
& -\left( R_{0n}^{\mathsf{(6VD)}}\right) _{12}(\lambda |\tau )\left( R_{0n}^{%
\mathsf{(6VD)}}\right) _{21}(\lambda -\eta |\tau -\eta ),
\end{align}%
and%
\begin{equation}
\text{$\mathsf{S}$}_{n}\equiv \sum_{a=1}^{n}\sigma _{a}^{z}.
\end{equation}%
Then it is a simple exercise to verify the identity:%
\begin{equation}
\det{}_{q}R_{0,n}^{\mathsf{(6VD)}}(\lambda -\xi _{n}|\tau +\eta \text{$\mathsf{S}$%
}_{n-1})=a(\lambda -\xi _{n})a(\lambda -\xi _{n}-2\eta )\frac{\theta (\tau
+\eta \text{$\mathsf{S}$}_{n-1})}{\theta (\tau +\eta \text{$\mathsf{S}$}_{n})%
},
\end{equation}%
once we use the formula\footnote{%
See for example equation 7 at page 881 of \cite{FarXYZTables of integrals}.}:%
\begin{equation}
\theta _{1}(x+y)\theta _{1}(x-y)\theta _{4}^{2}(0)=\theta _{3}^{2}(x)\theta
_{2}^{2}(y)-\theta _{2}^{2}(x)\theta _{3}^{2}(y)=\theta _{1}^{2}(x)\theta
_{4}^{2}(y)-\theta _{4}^{2}(x)\theta _{1}^{2}(y).
\end{equation}%
Now, by taking the product:%
\begin{align}
\text{\textsc{a}}(\lambda )\text{\textsc{d}}(\lambda -\eta )& =\prod_{n=1}^{%
\mathsf{N}}a(\lambda -\xi _{n})a(\lambda -\xi _{n}-2\eta )  \notag \\
& =\prod_{n=1}^{\mathsf{N}}\left[ \frac{\theta (\tau +\eta \text{$\mathsf{S}$%
}_{n})}{\theta (\tau +\eta \text{$\mathsf{S}$}_{n-1})}\det{}_{q}R_{0,n}^{\mathsf{%
(6VD)}}(\lambda |\tau +\eta \mathsf{S}_{n-1})\right]  \notag \\
& =\frac{\theta (\tau +\eta \text{$\mathsf{S}$})}{\theta (\tau )}\det{}_{q}R_{0,%
\mathsf{N}}^{\mathsf{(6VD)}}(\lambda |\tau +\eta \mathsf{S}_{\mathsf{N}%
-1})\cdots \det{}_{q}R_{0,2}^{\mathsf{(6VD)}}(\lambda |\tau +\eta \mathsf{S}%
_{1})\det{}_{q}R_{0,1}^{\mathsf{(6VD)}}(\lambda |\tau )  \notag \\
& =\det{}_{q}\mathsf{M}(\lambda ).
\end{align}%
Finally, the inversion formula $\left( \ref{FarXYZRight-1-dyn-Mon}\right) $\
follows from the quantum determinant formulae and from the identities:%
\begin{eqnarray}
\mathsf{A}(\lambda |\tau )\mathsf{B}(\lambda -\eta |\tau +\eta )-\mathsf{B}%
(\lambda |\tau )\mathsf{A}(\lambda -\eta |\tau -\eta ) &=&0, \\
\mathsf{D}(\lambda |\tau )\mathsf{C}(\lambda -\eta |\tau -\eta )-\mathsf{C}%
(\lambda |\tau )\mathsf{D}(\lambda -\eta |\tau +\eta ) &=&0,
\end{eqnarray}%
which directly follows from the dynamical Yang-Baxter equations $\left( \ref%
{FarXYZDyn-YB-Eq0}\right) $.
\end{proof}

\subsubsection{Antiperiodic dynamical 6-vertex representations}

As we will show in the paper it is of particular interest to introduce an
antiperiodic version of the dynamical 6-vertex model by introducing the
following monodromy matrix:%
\begin{equation}
\mathsf{\bar{M}}_{0}^{\mathsf{(6VD)}}(\lambda _{1}|\tau )\equiv \sigma
_{0}^{x}\mathsf{M}_{0}^{\mathsf{(6VD)}}(\lambda _{1}|\tau )
\label{FarXYZanti-p-6vD-M}
\end{equation}%
then the dynamical Yang-Baxter equation reads:%
\begin{equation}
R_{1,2}^{\mathsf{(6VD)}}(\lambda _{12}|-\tau -\eta \text{$\mathsf{S}$})%
\mathsf{\bar{M}}_{1}^{\mathsf{(6VD)}}(\lambda _{1}|\tau )\mathsf{T}_{\tau
}^{\sigma _{1}^{z}}\mathsf{\bar{M}}_{2}^{\mathsf{(6VD)}}(\lambda _{2}|\tau )%
\mathsf{T}_{\tau }^{-\sigma _{1}^{z}}\left. =\right. \mathsf{\bar{M}}_{2}^{%
\mathsf{(6VD)}}(\lambda _{2}|\tau )\mathsf{T}_{\tau }^{\sigma _{2}^{z}}%
\mathsf{\bar{M}}_{1}^{\mathsf{(6VD)}}(\lambda _{1}|\tau )\mathsf{T}_{\tau
}^{-\sigma _{2}^{z}}R_{1,2}^{\mathsf{(6VD)}}(\lambda _{12}|\tau ),
\label{FarXYZD-YB-op-0}
\end{equation}%
being:%
\begin{equation}
\sigma _{1}^{x}\otimes \sigma _{2}^{x}R_{1,2}^{\mathsf{(6VD)}}(\lambda
|y)=R_{1,2}^{\mathsf{(6VD)}}(\lambda |-y)\sigma _{1}^{x}\otimes \sigma
_{2}^{x}.
\end{equation}%
It is worth also to define the following monodromy matrix:%
\begin{equation}
\mathcal{\bar{M}}_{0}^{\mathsf{(6VD)}}(\lambda |\tau )\equiv \mathsf{\bar{M}}%
_{0}^{\mathsf{(6VD)}}(\lambda |\tau )\text{ }\mathsf{T}_{\tau }^{\sigma
_{0}^{z}},
\end{equation}%
for it the dynamical Yang-Baxter equation reads:%
\begin{equation}
R_{1,2}^{\mathsf{(6VD)}}(\lambda _{12}|-\tau -\eta \text{$\mathsf{S}$})%
\overline{\mathcal{M}}_{1}^{\mathsf{(6VD)}}(\lambda _{1}|\tau )\text{ }%
\overline{\mathcal{M}}_{1}^{\mathsf{(6VD)}}(\lambda _{2}|\tau )\text{ }%
\left. =\right. \overline{\mathcal{M}}_{1}^{\mathsf{(6VD)}}(\lambda
_{2}|\tau )\text{ }\overline{\mathcal{M}}_{1}^{\mathsf{(6VD)}}(\lambda
_{1}|\tau )\text{ }R_{1,2}^{\mathsf{(6VD)}}(\lambda _{12}|\tau ).
\label{FarXYZYBECalDyn}
\end{equation}

\subsubsection{Invariant subspace under antiperiodic 6VD-generators}

Let us define the operator:%
\begin{equation}
\mathsf{S}_{\tau }\equiv \eta \mathsf{S}+2\tau ,
\end{equation}%
in $\mathbb{D}_{\mathsf{(6VD)},\mathsf{N}}$ then we denote with $\mathbb{%
\bar{D}}_{\mathsf{(6VD)},\mathsf{N}}$ the $2^{\mathsf{N}}$-dimensional
linear eigenspace corresponding to the eigenvalue zero $\mathsf{S}_{\tau }$;
i.e. $\mathbb{\bar{D}}_{\mathsf{(6VD)},\mathsf{N}}$ is the linear space
defined by the condition that the eigenvalues of the commuting operators $%
-2\tau $ and $\eta \mathsf{S}$ are coinciding. In terms of the
dynamical-spin basis the linear (covector) space $\mathbb{\bar{D}}_{\mathsf{%
(6VD)},\mathsf{N}}^{\mathcal{L}}$ is generated by the elements\footnote{%
Note that we are using the simplified notation $t_{\text{\textbf{h}}}$\
instead of $t(\mathsf{s}_{\text{\textbf{h}}})$.}:%
\begin{equation}
\otimes _{n=1}^{\mathsf{N}}\langle n,h_{n}|\otimes \langle t_{\text{\textbf{h%
}}}|\text{, \ \ \ where \ }t_{\text{\textbf{h}}}\equiv -\frac{\eta }{2}%
\mathsf{s}_{\text{\textbf{h}}},\text{ \ }\mathsf{s}_{\text{\textbf{h}}%
}\equiv \sum_{k=1}^{\mathsf{N}}(1-2h_{k})\text{ and \textbf{h}}\equiv
(h_{1},...,h_{\mathsf{N}}),  \label{FarXYZDyS-basis-L}
\end{equation}%
and the linear (vector) space $\mathbb{\bar{D}}_{\mathsf{(6VD)},\mathsf{N}}^{%
\mathcal{R}}$ is generated by the elements:%
\begin{equation}
\otimes _{n=1}^{\mathsf{N}}|n,h_{n}\rangle \otimes |t_{\text{\textbf{h}}%
}\rangle .  \label{FarXYZDyS-basis-R}
\end{equation}%
In the $2^{\mathsf{N}}$-dimensional linear space $\mathbb{\bar{D}}_{\mathsf{%
(6VD)},\mathsf{N}}^{\mathcal{L}/\mathcal{R}}$ it is central to remark that
zero is a $\tau $-eigenvalue for a chain with $\mathsf{N}$ even while it is not for a chain with $\mathsf{N}$ odd. This simple observation implies that $\mathbb{\bar{D}}_{\mathsf{(6VD)}%
,\mathsf{N}}^{\mathcal{L}/\mathcal{R}}$ is not a well define representation
space of the dynamical Yang-Baxter algebra for the presence of divergencies
in $\left( \ref{FarXYZop-L}\right) $ for the zero $\tau $-eigenvalue. On the
contrary, in the case of an odd chain it holds:

\begin{theorem}
On the linear spaces $\mathbb{\bar{D}}_{\mathsf{(6VD)},\mathsf{N}}^{\mathcal{%
L}/\mathcal{R}}$ are well defined left/right finite dimensional
representations of the operators:%
\begin{equation}
\mathsf{A}(\lambda |\tau ),\text{ \ }\mathsf{D}(\lambda |\tau ),\text{ \ }%
\mathcal{B}(\lambda |\tau ),\text{ \ }\mathcal{C}(\lambda |\tau ).
\end{equation}%
Moreover, the antiperiodic dynamical 6-vertex transfer matrix:%
\begin{equation}
\overline{\mathcal{T}}^{\mathsf{(6VD)}}(\lambda |\tau )\equiv tr_{0}%
\overline{\mathcal{M}}_{0}^{\mathsf{(6VD)}}(\lambda |t)=\mathcal{B}(\lambda
|\tau )+\mathcal{C}(\lambda |\tau ),
\end{equation}%
defines a one parameter family of commuting operators on $\mathbb{\bar{D}}_{%
\mathsf{(6VD)},\mathsf{N}}^{\mathcal{L}/\mathcal{R}}$.
\end{theorem}

\begin{proof}
To prove the first statement in the theorem we have to prove that the linear
spaces $\mathbb{\bar{D}}_{\mathsf{(6VD)},\mathsf{N}}^{\mathcal{L}/\mathcal{R}%
}$ are invariant under the action of the operators:%
\begin{equation}
\mathsf{A}(\lambda |\tau ),\text{ \ }\mathsf{D}(\lambda |\tau ),\text{ \ }%
\mathcal{B}(\lambda |\tau ),\text{ \ }\mathcal{C}(\lambda |\tau ).
\end{equation}%
All what we need are the following commutation relations:%
\begin{eqnarray}
\lbrack \mathsf{A}(\lambda |\tau ),\mathsf{S}] &=&0,\text{ \ \ \ \ \ \ \ \ \
\ \ \ }[\mathsf{A}(\lambda |\tau ),\tau ]=0, \\
\lbrack \mathsf{A}(\lambda |\tau ),\mathsf{S}] &=&0,\text{ \ \ \ \ \ \ \ \ \
\ \ \ }[\mathsf{A}(\lambda |\tau ),\tau ]=0, \\
\lbrack \mathcal{B}(\lambda |\tau ),\mathsf{S}] &=&2\mathcal{B}(\lambda
|\tau )\text{ \ \ \ }[\mathcal{B}(\lambda |\tau ),\tau ]=-\eta \mathcal{B}%
(\lambda |\tau ), \\
\lbrack \mathcal{C}(\lambda |\tau ),\mathsf{S}] &=&-2\mathcal{C}(\lambda
|\tau ),\text{ }[\mathcal{C}(\lambda |\tau ),\tau ]=\eta \mathcal{C}(\lambda
|\tau ),
\end{eqnarray}%
from which it follows that:%
\begin{equation}
\lbrack \mathsf{A}(\lambda |\tau ),\mathsf{S}_{\tau }]=[\mathsf{D}(\lambda
|\tau ),\mathsf{S}_{\tau }]=[\mathcal{B}(\lambda |\tau ),\mathsf{S}_{\tau
}]=[\mathcal{C}(\lambda |\tau ),\mathsf{S}_{\tau }]=0,
\end{equation}%
and then $\mathbb{\bar{D}}_{\mathsf{(6VD)},\mathsf{N}}^{\mathcal{L}/\mathcal{%
R}}$ are invariant under the action of these operators; then $\mathbb{\bar{D}%
}_{\mathsf{(6VD)},\mathsf{N}}^{\mathcal{L}/\mathcal{R}}$ are invariant also
w.r.t. the action of the transfer matrix $\overline{\mathcal{T}}^{\mathsf{%
(6VD)}}(\lambda |\tau )$. Let us now take the trace of (\ref{FarXYZYBECalDyn}%
):%
\begin{align}
& \overline{\mathcal{T}}^{\mathsf{(6VD)}}(\lambda _{1}|\tau )\overline{%
\mathcal{T}}^{\mathsf{(6VD)}}(\lambda _{2}|\tau )\left. =\right.  \notag \\
& \text{ \ \ \ \ \ \ }\left. =\right. tr_{12}\left[ \overline{\mathcal{M}}%
_{1}^{\mathsf{(6VD)}}(\lambda _{1}|\tau )\text{ }\overline{\mathcal{M}}_{2}^{%
\mathsf{(6VD)}}(\lambda _{2}|\tau )\right]  \notag \\
& \text{\ \ \ \ \ \ \ }\left. =\right. tr_{12}\left[ \left( R_{1,2}^{\mathsf{%
(6VD)}}(\lambda _{12}|-\tau -\eta \text{$\mathsf{S}$})\right) ^{-1}\overline{%
\mathcal{M}}_{2}^{\mathsf{(6VD)}}(\lambda _{2}|\tau )\text{ }\overline{%
\mathcal{M}}_{1}^{\mathsf{(6VD)}}(\lambda _{1}|\tau )R_{1,2}^{\mathsf{(6VD)}%
}(\lambda _{12}|\tau )\right]  \notag \\
& \text{\ \ \ \ \ \ \ }\left. =\right. tr_{12}\left[ \left( R_{1,2}^{\mathsf{%
(6VD)}}(\lambda _{12}|-\tau -\eta \text{$\mathsf{S}$})\right) ^{-1}\overline{%
\mathsf{M}}_{2}^{\mathsf{(6VD)}}(\lambda _{2}|\tau )\text{ }\overline{%
\mathsf{M}}_{1}^{\mathsf{(6VD)}}(\lambda _{1}|\tau +\eta \sigma _{2}^{z})%
\mathsf{T}_{\tau }^{\sigma _{1}^{z}}\mathsf{T}\text{$_{\tau }^{\sigma
_{2}^{z}}$}R_{1,2}^{\mathsf{(6VD)}}(\lambda _{12}|\tau )\right]  \notag \\
& \text{\ \ }\underset{(\ref{FarXYZComm-R12})}{\left. =\right. }tr_{12}\left[
R_{1,2}^{\mathsf{(6VD)}}(\lambda _{12}|\tau )\left( R_{1,2}^{\mathsf{(6VD)}%
}(\lambda _{12}|-\tau -\eta \text{$\mathsf{S}$})\right) ^{-1}\overline{%
\mathsf{M}}_{2}^{\mathsf{(6VD)}}(\lambda _{2}|\tau )\text{ }\overline{%
\mathsf{M}}_{1}^{\mathsf{(6VD)}}(\lambda _{1}|\tau +\eta \sigma _{2}^{z})%
\mathsf{T}_{\tau }^{\sigma _{1}^{z}}\mathsf{T}\text{$_{\tau }^{\sigma
_{2}^{z}}$}\right]
\end{align}%
which in $\mathbb{\bar{D}}_{\mathsf{(6VD)},\mathsf{N}}^{\mathcal{L}}$
coincides with%
\begin{equation}
tr_{12}\left[ \overline{\mathcal{M}}_{2}^{\mathsf{(6VD)}}(\lambda _{2}|\tau )%
\text{ }\overline{\mathcal{M}}_{1}^{\mathsf{(6VD)}}(\lambda _{1}|\tau )%
\right] =\overline{\mathcal{T}}^{\mathsf{(6VD)}}(\lambda _{1}|\tau )%
\overline{\mathcal{T}}^{\mathsf{(6VD)}}(\lambda _{2}|\tau ),
\end{equation}%
i.e. the commutativity.
\end{proof}

From now on we will implicitly assume that $\mathsf{N}$ is odd when
representations in the linear spaces $\mathbb{\bar{D}}_{\mathsf{(6VD)},%
\mathsf{N}}^{\mathcal{L}/\mathcal{R}}$ will be considered.

\subsection{SOV-representations for $\overline{\mathcal{T}}^{\mathsf{(6VD)}}$%
-spectral problem}

\subsubsection{SOV-representations}

Here, we will show as the standard method to define quantum separation of
variable (SOV) representations introduced by Sklyanin \cite%
{FarXYZSk1,FarXYZSk2,FarXYZSk3} for the transfer matrix of 6-vertex
Yang-Baxter algebra can be adapted for the dynamical case. In particular,
SOV representations for the spectral problem of the antiperiodic $\overline{%
\mathcal{T}}^{\mathsf{(6VD)}}(\lambda |\tau )$ can be defined as
the representations where the commutative family of operators $\mathsf{D}%
(\lambda |\tau )$ (or $\mathsf{A}(\lambda |\tau )$) is pseudo-diagonal and
with simple spectrum in the (left/right) $2^{\mathsf{N}}$-dimensional spaces $\mathbb{\bar{D}}_{\mathsf{(6VD)},\mathsf{N}}^{\mathcal{L}/%
\mathcal{R}}$. Here, we mean that we can construct explicitly left/right
basis of the spaces $\mathbb{\bar{D}}_{\mathsf{(6VD)},\mathsf{N}}^{\mathcal{L%
}/\mathcal{R}}$ in terms of pseudo $\mathsf{D}(\lambda |\tau )$-eigenstates%
\footnote{What we mean for pseudo-eigenstates will be clarified in the following.}.

In order to make the construction of these basis some preparation is need;
let us define the left and right \textit{references states\footnote{%
The left and right states of $\mathbb{\bar{D}}_{\mathsf{(6VD)},\mathsf{N}}^{%
\mathcal{L}/\mathcal{R}}$ with all spin up and down, respectively.}}:%
\begin{equation}
\langle \text{\textbf{0}}|\equiv \otimes _{n=1}^{\mathsf{N}}\langle
n,h_{n}=0|\otimes \langle t_{\text{\textbf{0}}}|,\text{ \ \ \ \ }|\text{%
\textbf{1}}\rangle \equiv \otimes _{n=1}^{\mathsf{N}}|n,h_{n}=1\rangle
\otimes |t_{\text{\textbf{1}}}\rangle ,
\end{equation}%
where we are using the notations \textbf{0}$\equiv (h_{1}=0,...,h_{\mathsf{N}%
}=0)$ and \textbf{1}$\equiv (h_{1}=1,...,h_{\mathsf{N}}=1)$, so that $t_{%
\text{\textbf{0}}}\equiv -(\eta \mathsf{N})/2$ and $t_{\text{\textbf{1}}%
}\equiv (\eta \mathsf{N})/2$ by the definition $\left( \ref{FarXYZDyS-basis-L}%
\right) $ of $t_{\text{\textbf{h}}}$. Then, we can define the following sets
of left and right states:%
\begin{equation}
\langle h_{1},...,h_{\mathsf{N}}|\equiv \frac{1}{\text{\textsc{n}}}\langle 
\text{\textbf{0}}|\prod_{n=1}^{\mathsf{N}}\left( \frac{\mathcal{C}(\xi
_{n}|\tau )}{\text{\textsc{d}}(\xi _{n}-\eta )}\right) ^{h_{n}},
\label{FarXYZD-left-eigenstates}
\end{equation}%
and%
\begin{equation}
|h_{1},...,h_{\mathsf{N}}\rangle \equiv \frac{1}{\text{\textsc{n}}}%
\prod_{n=1}^{\mathsf{N}}\left( \frac{\mathcal{C}(\xi _{n}-\eta |\tau )}{%
\text{\textsc{d}}(\xi _{n}-\eta )}\right) ^{(1-h_{n})}|\text{\textbf{1}}%
\rangle ,  \label{FarXYZD-right-eigenstates}
\end{equation}%
where\ $h_{n}\in \{0,1\},$ $n\in \{1,...,\mathsf{N}\}$. Note that the
normalization \textsc{n} has been introduced to simplify the form of the
coupling of the above left and right states. Let us define first the
following theta functions with characteristic:%
\begin{equation}
\vartheta _{j}\left( \lambda \right) =\sum_{n\in \mathbb{Z}}\exp \left[
2i\pi w\mathsf{N}\left( n+\frac{1}{2}-\frac{j}{\mathsf{N}}\right) +2i\pi 
\mathsf{N}\left( n+\frac{1}{2}-\frac{j}{\mathsf{N}}\right) \left( \lambda +%
\frac{1}{2\mathsf{N}}\right) \right] ,
\end{equation}%
$\mathsf{N}\in \mathbb{N}$ and $j\in \{0,...,\mathsf{N}-1\}$ which satisfies
the periodicity conditions:%
\begin{equation}
\vartheta _{j}\left( \lambda +1/\mathsf{N}\right) =-e^{2\pi i(j/\mathsf{N}%
)}\vartheta _{j}\left( \lambda \right) ,\text{ \ }\vartheta _{j}\left(
\lambda +2w\right) =-e^{-2\pi i\mathsf{N}(w+\lambda )}\vartheta _{j}\left(
\lambda \right) .
\end{equation}%
Then we can fix:%
\begin{equation}
\text{\textsc{n}}^{2}\equiv c_{\mathsf{N}}\frac{\langle \text{\textbf{0}}%
|\prod_{n=1}^{\mathsf{N}}\left( \frac{\mathcal{C}(\xi _{n}-\eta |\tau )}{%
\text{\textsc{d}}(\xi _{n}-\eta )}\right) ^{(1-h_{n})}|\text{\textbf{1}}%
\rangle }{\det_{\mathsf{N}}\Theta _{ij}^{\left( \text{\textbf{h}}\equiv 
\text{\textbf{0}}\right) }},
\end{equation}%
where $c_{\mathsf{N}}$ is a normalization constant defined in (\ref%
{FarXYZR-Pseudo-Vendermonde}) and $\left\Vert \Theta _{ij}^{\left( \text{%
\textbf{h}}\right) }\right\Vert $ is the $\mathsf{N}\times \mathsf{N}$
matrix of elements:%
\begin{equation}
\Theta _{ij}^{\left( \text{\textbf{h}}\right) }\equiv \vartheta _{i-1}\left( 
\bar{\xi}_{j}^{(h_{j})}\right) ,\text{ \ \ \ \ }\bar{\xi}_{a}^{(h_{a})}=\xi
_{a}^{(h_{a})}+\frac{\eta }{2}+\frac{\mathsf{N}-1}{2\mathsf{N}}-\frac{1}{%
\mathsf{N}}\sum_{a=1}^{\mathsf{N}}\xi _{a},\text{ \ \ \ \ }\xi
_{a}^{(h_{a})}=\xi _{a}-\eta h_{a}\,.  \label{FarXYZpseudo-Vander}
\end{equation}%
It is simple to verify that the states (\ref{FarXYZD-left-eigenstates}) and (%
\ref{FarXYZD-right-eigenstates}) are simultaneous (left/right) eigenstates
of $\tau $ and $\mathsf{S}\equiv \sum_{a=1}^{\mathsf{N}}\sigma _{a}^{z}$:%
\begin{eqnarray}
\langle h_{1},...,h_{\mathsf{N}}|\tau &=&t_{\text{\textbf{h}}}\text{ }%
\langle h_{1},...,h_{\mathsf{N}}|, \\
\langle h_{1},...,h_{\mathsf{N}}|\mathsf{S} &=&\mathsf{s}_{\text{\textbf{h}}%
}\langle h_{1},...,h_{\mathsf{N}}|,
\end{eqnarray}%
and:%
\begin{eqnarray}
\tau |h_{1},...,h_{\mathsf{N}}\rangle &=&|h_{1},...,h_{\mathsf{N}}\rangle 
\text{ }t_{\text{\textbf{h}}}, \\
\mathsf{S}|h_{1},...,h_{\mathsf{N}}\rangle &=&|h_{1},...,h_{\mathsf{N}%
}\rangle \mathsf{s}_{\text{\textbf{h}}}.
\end{eqnarray}%
Moreover, under the following condition on the $\mathsf{N}$-tuple of
inhomogeneities $\{\xi _{1},...,\xi _{\mathsf{N}}\}\in \mathbb{C}$ $^{%
\mathsf{N}}$:%
\begin{equation}
\xi _{a}\neq \xi _{b}^{(h_{b})}\text{ \thinspace \thinspace mod}%
(2w)\,\,\forall h_{b}\in \{0,1\}\,\,\,\text{and}\,\,\,a<b\in \{1,...,\mathsf{%
N}\},  \label{FarXYZE-SOV}
\end{equation}%
the following theorem holds:

\begin{theorem}
\textsf{I)} \underline{Left SOV-representations} \ Under the condition $%
\left( \ref{FarXYZE-SOV}\right) $, the states $\left( \ref%
{FarXYZD-left-eigenstates}\right) $ define a basis of pseudo $\mathsf{D}%
(\lambda |\tau )$-eigenstates in $\mathbb{\bar{D}}_{\mathsf{(6VD)},\mathsf{N}%
}^{\mathcal{L}}$; indeed it holds:%
\begin{equation}
\langle h_{1},...,h_{\mathsf{N}}|\mathsf{D}(\lambda |\tau )=\mathsf{d}_{%
\text{\textbf{h}}}^{\mathcal{L}}(\lambda )\,\left( \frac{1}{\text{\textsc{n}}%
}\langle \text{\textbf{0}}|\prod_{n=1}^{\mathsf{N}}\left( \frac{\mathcal{C}%
(\xi _{n}|\tau -\eta )}{\text{\textsc{d}}(\xi _{n}-\eta )}\right)
^{h_{n}}\right) \text{ },  \label{FarXYZD-L-EigenV}
\end{equation}%
where:%
\begin{equation}
\mathsf{d}_{\text{\textbf{h}}}^{\mathcal{L}}(\lambda )\equiv \frac{\theta
(t_{\text{\textbf{0}}}-\eta )}{\theta (t_{\text{\textbf{0}}}+\eta )}\frac{%
\theta (t_{\text{\textbf{1}}})}{\theta (t_{\text{\textbf{h}}})}\mathsf{d}_{%
\text{\textbf{h}}}(\lambda ),\,\,\,\mathsf{d}_{\text{\textbf{h}}}(\lambda
)\equiv \prod_{n=1}^{\mathsf{N}}\theta (\lambda -\xi _{n}^{(h_{n})}).
\label{FarXYZEigenValue-D}
\end{equation}%
The action of the remaining generators on the generic state $\langle
h_{1},...,h_{\mathsf{N}}|$ reads:%
\begin{align}
\langle h_{1},...,h_{\mathsf{N}}|\mathcal{C}(\lambda |\tau )& =\sum_{a=1}^{%
\mathsf{N}}\frac{\theta (\tau -\lambda +\xi _{a}^{(h_{a})})}{\theta (\tau )}%
\prod_{b\neq a}\frac{\theta (\lambda -\xi _{b}^{(h_{b})})}{\theta (\xi
_{a}^{(h_{a})}-\xi _{b}^{(h_{b})})}\text{\textsc{d}}(\xi
_{a}^{(1-h_{a})})\langle h_{1},...,h_{\mathsf{N}}|\mathsf{T}_{a}^{+},
\label{FarXYZC-SOV_D-left} \\
&  \notag \\
\langle h_{1},...,h_{\mathsf{N}}|\mathcal{B}(\lambda |\tau )& =\sum_{a=1}^{%
\mathsf{N}}\frac{\theta (\tau -\lambda +\xi _{a}^{(h_{a})})}{\theta (\tau )}%
\prod_{b\neq a}\frac{\theta (\lambda -\xi _{b}^{(h_{b})})}{\theta (\xi
_{a}^{(h_{a})}-\xi _{b}^{(h_{b})})}\text{\textsc{a}}(\xi
_{a}^{(1-h_{a})})\langle h_{1},...,h_{\mathsf{N}}|\mathsf{T}_{a}^{-},
\label{FarXYZB-SOV_D-left}
\end{align}%
where:%
\begin{equation}
\langle h_{1},...,h_{\mathsf{N}}|\mathsf{T}_{a}^{\pm }=\langle
h_{1},...,h_{a}\pm 1,...,h_{\mathsf{N}}|,
\end{equation}%
and $\mathsf{A}(\lambda |\tau )$ is uniquely defined by the quantum
determinant relation.\smallskip

\textsf{II)} \underline{Right SOV-representations} \ Under the condition $%
\left( \ref{FarXYZE-SOV}\right) $, the states $(\ref%
{FarXYZD-right-eigenstates})$ define a basis of pseudo $\mathsf{D}(\lambda
|\tau )$-eigenstates in $\mathbb{\bar{D}}_{\mathsf{(6VD)},\mathsf{N}}^{%
\mathcal{R}}$; indeed it holds:%
\begin{equation}
\mathsf{D}(\lambda |\tau +\eta )|h_{1},...,h_{\mathsf{N}}\rangle =\text{ }%
\left( \frac{1}{\text{\textsc{n}}}\prod_{n=1}^{\mathsf{N}}\left( \frac{%
\mathcal{C}(\xi _{n}-\eta |\tau +\eta )}{\text{\textsc{d}}(\xi _{n}-\eta )}%
\right) ^{(1-h_{n})}|\text{\textbf{1}}\rangle \right) \mathsf{d}_{\text{%
\textbf{h}}}^{\mathcal{R}}(\lambda )\text{ },  \label{FarXYZD-R-EigenV}
\end{equation}%
where:%
\begin{equation}
\mathsf{d}_{\text{\textbf{h}}}^{\mathcal{R}}(\lambda )\equiv \frac{\theta
(t_{\text{\textbf{h}}}+\eta )}{\theta (t_{\text{\textbf{1}}}+\eta )}\mathsf{d%
}_{\text{\textbf{h}}}(\lambda ).
\end{equation}%
The action of the remaining generators on the generic state $|h_{1},...,h_{%
\mathsf{N}}\rangle $ reads:%
\begin{align}
\mathcal{C}(\lambda |\tau )|h_{1},...,h_{\mathsf{N}}\rangle & =\sum_{a=1}^{%
\mathsf{N}}\mathsf{T}_{a}^{-}\text{ }|h_{1},...,h_{\mathsf{N}}\rangle \frac{%
\theta (\tau -\lambda +\xi _{a}^{(h_{a})})}{\theta (\tau )}\prod_{b\neq a}%
\frac{\theta (\lambda -\xi _{b}^{(h_{b})})}{\theta (\xi _{a}^{(h_{a})}-\xi
_{b}^{(h_{b})})}\text{\textsc{d}}(\xi _{a}^{(h_{a})}),
\label{FarXYZC-SOV_D-right} \\
&  \notag \\
\mathcal{B}(\lambda |\tau )|h_{1},...,h_{\mathsf{N}}\rangle & =\sum_{a=1}^{%
\mathsf{N}}\mathsf{T}_{a}^{+}\text{ }|h_{1},...,h_{\mathsf{N}}\rangle \frac{%
\theta (\tau -\lambda +\xi _{a}^{(h_{a})})}{\theta (\tau )}\prod_{b\neq a}%
\frac{\theta (\lambda -\xi _{b}^{(h_{b})})}{\theta (\xi _{a}^{(h_{a})}-\xi
_{b}^{(h_{b})})}\text{\textsc{a}}(\xi _{a}^{(h_{a})}),
\label{FarXYZB-SOV_D-right}
\end{align}%
where:%
\begin{equation}
\mathsf{T}_{a}^{\pm }|h_{1},...,h_{\mathsf{N}}\rangle =|h_{1},...,h_{a}\pm
1,...,h_{\mathsf{N}}\rangle,
\end{equation}%
and $\mathsf{A}(\lambda |\tau )$ is uniquely defined by the quantum
determinant relation.
\end{theorem}

\begin{proof}
The proof of the theorem is based on the dynamical Yang-Baxter commutation
relations and on the fact that the left and right references states are $%
\mathsf{D}(\lambda |\tau )$-eigenstates:%
\begin{equation}
\langle \text{\textbf{0}}|\mathsf{A}(\lambda |\tau )=\text{\textsc{a}}%
(\lambda )\langle \text{\textbf{0}}|,\text{ \ \ \ }\langle \text{\textbf{0}}|%
\mathsf{D}(\lambda |\tau )=\text{\textsc{d}}(\lambda |t_{\text{\textbf{0}}%
})\langle \text{\textbf{0}}|,\text{ \ \ \ }\langle \text{\textbf{0}}|%
\mathcal{B}(\lambda |\tau )=\text{\b{0}},\text{ \ \ \ }\langle \text{\textbf{%
0}}|\mathcal{C}(\lambda |\tau )\neq \text{\b{0}},
\end{equation}%
and similarly:%
\begin{equation}
\mathsf{D}(\lambda |\tau )|\text{\textbf{1}}\rangle =|\text{\textbf{1}}%
\rangle \text{\textsc{a}}(\lambda ),\text{ \ \ \ }\mathsf{A}(\lambda |\tau )|%
\text{\textbf{1}}\rangle =|\text{\textbf{1}}\rangle \text{\textsc{d}}%
(\lambda |t_{\text{\textbf{0}}}),\text{ \ \ \ }\mathcal{B}(\lambda |\tau )|%
\text{\textbf{1}}\rangle =\text{\b{0}},\text{ \ \ \ }\mathcal{C}(\lambda
|\tau )|\text{\textbf{1}}\rangle \neq \text{\b{0}},
\end{equation}%
where:%
\begin{equation}
\text{\textsc{d}}(\lambda |t_{\text{\textbf{0}}})\equiv \text{\textsc{d}}%
(\lambda )\frac{\theta (\eta -t_{\text{\textbf{0}}})}{\theta (\eta +t_{\text{%
\textbf{0}}})}.
\end{equation}%
Indeed, to prove that $\left( \ref{FarXYZD-left-eigenstates}\right) $ and $%
\left( \ref{FarXYZD-right-eigenstates}\right) $ are left and right
pseudo-eigenstates of $\mathsf{D}(\lambda |\tau )$ as stated in $\left( \ref%
{FarXYZD-L-EigenV}\right) $ and $\left( \ref{FarXYZD-R-EigenV}\right) $, we
have just to repeat the standard computations in algebraic Bethe ansatz \cite%
{FarXYZF95} as done in \cite{FarXYZN12-0}. More in detail, considering the
action of $\mathcal{D}(\lambda |\tau )$ on the left states $\langle
h_{1},...,h_{\mathsf{N}}|$ and following the steps given in the proof of
Theorem 3.2 of \cite{FarXYZN12-0} by using here the dynamical 6-vertex
commutation relation:%
\begin{align}
\mathcal{C}(\mu |\tau )\mathcal{D}(\lambda |\tau )& =\left[ \mathcal{D}%
(\lambda |\tau )\mathcal{C}(\mu |\tau )\theta (\lambda -\mu +\eta )\theta
(\tau )-\mathcal{D}(\mu |\tau )\mathcal{C}(\lambda |\tau )\theta (\eta
)\theta (\tau +\lambda -\mu )\right]  \notag \\
& \times \frac{1}{\theta (\lambda -\mu )\theta (\tau +\eta )},
\label{FarXYZDC-YBC}
\end{align}

we get:%
\begin{equation}
\langle h_{1},...,h_{\mathsf{N}}|\text{ }\mathcal{D}(\lambda |\tau )=\mathsf{%
d}_{\text{\textbf{h}}}^{\mathcal{L}}(\lambda )\left[ \frac{1}{\text{\textsc{n%
}}}\left( \langle \text{\textbf{0}}|\mathsf{T}_{\tau }^{-}\right)
\prod_{n=1}^{\mathsf{N}}\left( \frac{\mathcal{C}(\xi _{n}|\tau )}{\text{%
\textsc{d}}(\xi _{n}-\eta )}\right) ^{h_{n}}\right] .
\label{FarXYZL-eigen-1}
\end{equation}%
Similarly, by using the dynamical 6-vertex commutation relation:%
\begin{align}
\mathcal{D}(\lambda |\tau )\mathcal{C}(\mu |\tau )& =\frac{1}{\theta (\mu
-\lambda )\theta (\tau +\eta )}[\theta (\mu -\lambda +\eta )\theta (\tau )%
\mathcal{C}(\mu |\tau )\mathcal{D}(\lambda |\tau )  \notag \\
& -\theta (\eta )\theta (\tau +\mu -\lambda )\mathcal{C}(\lambda |\tau )%
\mathcal{D}(\mu |\tau ),
\end{align}

we get:%
\begin{equation}
\mathcal{D}(\lambda |\tau )\text{ }|h_{1},...,h_{\mathsf{N}}\rangle =\text{ }%
\left( \frac{1}{\text{\textsc{n}}}\prod_{n=1}^{\mathsf{N}}\left( \frac{%
\mathcal{C}(\xi _{n}-\eta |\tau )}{\text{\textsc{d}}(\xi _{n}-\eta )}\right)
^{(1-h_{n})}\text{ }\left( \mathsf{T}_{\tau }^{-}\text{ }|\text{\textbf{1}}%
\rangle \right) \right) \mathsf{d}_{\text{\textbf{h}}}^{\mathcal{L}}(\lambda
).  \label{FarXYZR-eigen-1}
\end{equation}%
From the formulae $\left( \ref{FarXYZL-eigen-1}\right) $ and $\left( \ref%
{FarXYZR-eigen-1}\right) $ by using the commutation relations $\left( \ref%
{FarXYZDyn-op-comm}\right) $, $\mathsf{D}(\lambda |\tau )=\mathcal{D}%
(\lambda |\tau )\mathsf{T}_{\tau }^{+}$ and $\mathsf{D}(\lambda |\tau +\eta
)=\mathsf{T}_{\tau }^{+}\mathcal{D}(\lambda |\tau )$ then the formulae $%
\left( \ref{FarXYZD-L-EigenV}\right) $ and $\left( \ref{FarXYZD-R-EigenV}%
\right) $ simply follow.

Let us prove now that the states $\langle h_{1},...,h_{\mathsf{N}}|$ form a
set of 2$^{\mathsf{N}}$ independent states, i.e. a basis of $\mathbb{\bar{D}}%
_{\mathsf{(6VD)},\mathsf{N}}^{\mathcal{L}}$; similarly, we can prove that
the states $|h_{1},...,h_{\mathsf{N}}\rangle $ form a basis of $\mathbb{\bar{%
D}}_{\mathsf{(6VD)},\mathsf{N}}^{\mathcal{R}}$. By definition we have to
prove that their linear combination to zero:%
\begin{equation}
\sum_{h_{1},...,h_{\mathsf{N}}=0}^{1}c_{\text{\textbf{h}}}\langle
h_{1},...,h_{\mathsf{N}}|=\text{\b{0}}  \label{FarXYZLC-0}
\end{equation}%
holds only if all the coefficients are zeros. Let us denote with \textbf{\={h%
}}$=\{\bar{h}_{1},...,\bar{h}_{\mathsf{N}}\}$ the generic $\mathsf{N}$-tuple
in $\{0,1\}^{\mathsf{N}}$ then by applying to both side of $\left( \ref%
{FarXYZLC-0}\right) $ the operator product:%
\begin{equation}
\prod_{n=1}^{\mathsf{N}}\mathcal{D}(\xi _{n}^{\left( \bar{k}_{n}\right)
}|\tau )\text{ \ with }\bar{k}_{n}=\bar{h}_{n}+1\text{ mod2}\in \{0,1\}
\end{equation}
we get:
\begin{equation}
c_{\text{\textbf{\={h}}}}\text{ }\prod_{n=1}^{\mathsf{N}}\mathsf{d}_{\text{%
\textbf{\={h}}}}^{\mathcal{L}}(\xi _{n}^{\left( \bar{k}_{n}\right) })\left[
\left( \langle \text{\textbf{0}}|\mathsf{T}_{\tau }^{-\mathsf{N}}\right)
\prod_{n=1}^{\mathsf{N}}\left( \frac{\mathcal{C}(\xi _{n}|\tau )}{\text{%
\textsc{d}}(\xi _{n}-\eta )}\right) ^{\bar{h}_{n}}\right] =\text{\b{0}}
\end{equation}%
which implies $c_{\text{\textbf{\={h}}}}=0$ being:%
\begin{equation}
\prod_{n=1}^{\mathsf{N}}\mathsf{d}_{\text{\textbf{\={h}}}}^{\mathcal{L}}(\xi
_{n}^{\left( \bar{k}_{n}\right) })\neq 0,\text{ \ \ }\left[ \left( \langle 
\text{\textbf{0}}|\mathsf{T}_{\tau }^{-\mathsf{N}}\right) \prod_{n=1}^{%
\mathsf{N}}\left( \frac{\mathcal{C}(\xi _{n}|\tau )}{\text{\textsc{d}}(\xi
_{n}-\eta )}\right) ^{\bar{h}_{n}}\right] \neq \text{\b{0}}\,.
\end{equation}%
The action of $\mathcal{B}(\xi _{n}^{\left( h_{n}\right) }|\tau )$ and $%
\mathcal{C}(\xi _{n}^{\left( h_{n}\right) }|\tau )$ on the left and right
states $\left( \ref{FarXYZD-left-eigenstates}\right) $ and $\left( \ref%
{FarXYZD-right-eigenstates}\right) $ follows by imposing the dynamical
Yang-Baxter commutation relations and the quantum determinant relations:%
\begin{equation}
\langle h_{1},...,h_{\mathsf{N}}|\det{}_{q}\mathsf{M}(\lambda )=%
\text{\textsc{a}}(\lambda )\text{\textsc{d}}(\lambda -\eta )\langle
h_{1},...,h_{\mathsf{N}}|,\text{ \ }\det{}_{q}\mathsf{M}(\lambda
)|h_{1},...,h_{\mathsf{N}}\rangle =\text{\textsc{a}}(\lambda )\text{\textsc{d%
}}(\lambda -\eta )|h_{1},...,h_{\mathsf{N}}\rangle ,
\end{equation}%
where we have used that:%
\begin{equation}
\left. \frac{\theta (\tau +\eta \text{$\mathsf{S}$})}{\theta (\tau )}%
\right\vert _{\langle h_{1},...,h_{\mathsf{N}}|,|h_{1},...,h_{\mathsf{N}%
}\rangle }=\frac{\theta (-t_{\text{\textbf{h}}})}{\theta (t_{\text{\textbf{h}%
}})}=-1,
\end{equation}%
and in the quantum determinant $\det{}_{q}\mathsf{M}(\lambda )$ we
use the identities:%
\begin{align}
\mathcal{A}(\lambda |\tau )\mathcal{D}(\lambda -\eta |\tau )-\mathcal{B}%
(\lambda |\tau )\mathcal{C}(\lambda -\eta |\tau )& =\mathsf{A}(\lambda |\tau
)\mathsf{D}(\lambda -\eta |\tau +\eta )-\mathsf{B}(\lambda |\tau )\mathsf{C}%
(\lambda -\eta |\tau -\eta ), \\
\mathcal{D}(\lambda |\tau )\mathcal{A}(\lambda -\eta |\tau )-\mathcal{C}%
(\lambda |\tau )\mathcal{B}(\lambda -\eta |\tau )& =\mathsf{D}(\lambda |\tau
)\mathsf{A}(\lambda -\eta |\tau -\eta )-\mathsf{C}(\lambda |\tau )\mathsf{B}%
(\lambda -\eta |\tau +\eta ).
\end{align}%
Finally, the left $\left( \ref{FarXYZC-SOV_D-left}\right) $-$\left( \ref%
{FarXYZB-SOV_D-left}\right) $ and right $\left( \ref{FarXYZC-SOV_D-right}%
\right) $-$\left( \ref{FarXYZB-SOV_D-right}\right) $ representations of $%
\mathcal{B}(\lambda |\tau )$ and $\mathcal{C}(\lambda |\tau )$ are just
interpolation formulae in the special points $\{\xi _{1}^{\left(
h_{1}\right) },...,\xi _{\mathsf{N}}^{\left( h_{\mathsf{N}}\right) }\}$
which hold for elliptic polynomials as illustrated for example in Appendix
A of \cite{FarXYZPRL-08}.
\end{proof}

\subsubsection{SOV-decomposition of the identity}

The previous results allow to write the following spectral decomposition of
the identity $\mathbb{I}$:%
\begin{equation}
\mathbb{I}\equiv \sum_{h_{1},...,h_{\mathsf{N}}=0}^{1}\mu _{\text{\textbf{h}}%
}|h_{1},...,h_{\mathsf{N}}\rangle \langle h_{1},...,h_{\mathsf{N}}|,
\end{equation}%
in terms of the left and right SOV-basis. Here,%
\begin{equation}
\mu _{\text{\textbf{h}}}\equiv \frac{1}{\langle h_{1},...,h_{\mathsf{N}%
}|h_{1},...,h_{\mathsf{N}}\rangle },
\end{equation}%
is the analogous of the so-called Sklyanin's measure; which is discrete in
these representations and defined by the following proposition:

\begin{proposition}
Let $\langle h_{1},...,h_{\mathsf{N}}|$ be the generic covector (\ref%
{FarXYZD-left-eigenstates}) and $|k_{1},...,k_{\mathsf{N}}\rangle $ be the
generic vector (\ref{FarXYZD-right-eigenstates}), then it holds:%
\begin{equation}
\langle h_{1},...,h_{\mathsf{N}}|k_{1},...,k_{\mathsf{N}}\rangle =\frac{%
\prod_{c=1}^{\mathsf{N}}\delta _{h_{c},k_{c}}}{\det_{\mathsf{N}}\Theta
_{ij}^{\left( \text{\textbf{h}}\right) }},  \label{FarXYZM_jj}
\end{equation}%
where $\Theta _{ij}^{\left( \text{\textbf{h}}\right) }$ is the $\mathsf{N}%
\times \mathsf{N}$ matrix defined in $(\ref{FarXYZpseudo-Vander})$. Then,
the SOV-decomposition of the identity explicitly reads: 
\begin{equation}
\mathbb{I}\equiv \sum_{h_{1},...,h_{\mathsf{N}}=0}^{1}\det_{\mathsf{N}%
}\Theta _{ij}^{\left( \text{\textbf{h}}\right) }\text{ }|h_{1},...,h_{%
\mathsf{N}}\rangle \langle h_{1},...,h_{\mathsf{N}}|.
\label{FarXYZDecomp-Id}
\end{equation}
\end{proposition}

\begin{proof}
The fact that $\langle h_{1},...,h_{\mathsf{N}}|$\ and $|k_{1},...,k_{%
\mathsf{N}}\rangle $ are simultaneous eigenstates of $\tau $ with
eigenvalues $t_{\text{\textbf{h}}}$ and $t_{\text{\textbf{k}}}$,
respectively, implies that the l.h.s. of $(\ref{FarXYZM_jj})$ is zero unless:%
\begin{equation}
\sum_{c=1}^{\mathsf{N}}h_{c}=\sum_{c=1}^{\mathsf{N}}k_{c}.
\label{FarXYZCond-1}
\end{equation}%
Let us assume now that \textbf{h}$\neq $\textbf{k} but that they satisfy the
condition $(\ref{FarXYZCond-1})$. Under these conditions it is easy to
understand that there exists at least one $n\in \{1,...,\mathsf{N}\}$ such
that $h_{n}=1 $ and $k_{n}=0$ and then the l.h.s. of $(\ref{FarXYZM_jj})$
contains the product of operators $\mathcal{C}(\xi _{n}|\tau )\mathcal{C}%
(\xi _{n}-\eta |\tau )$ which is zero for the standard 6-vertex annihilation
identities. Then, as stated in $(\ref{FarXYZM_jj})$, $\langle h_{1},...,h_{%
\mathsf{N}}|k_{1},...,k_{\mathsf{N}}\rangle $ is zero for \textbf{h}$\neq $%
\textbf{k}; so we are left with the computations for \textbf{h}$=$\textbf{k}%
. In order to compute them, let us compute the matrix elements: 
\begin{equation}
x_{a}\equiv \langle h_{1},...,h_{a}=0,...,h_{\mathsf{N}}|\ \mathcal{C}(\xi
_{a}|\tau )|h_{1},...,h_{a}=1,...,h_{\mathsf{N}}\rangle ,
\end{equation}%
where $a\in \{1,...,\mathsf{N}\}$. Then using the left action of the
operator $\mathcal{C}(\eta _{a}|\tau )$ we get:%
\begin{equation}
x_{a}=\text{\textsc{d}}(\xi _{a}-\eta )\langle h_{1},...,h_{a}=1,...,h_{%
\mathsf{N}}|h_{1},...,h_{a}=1,...,h_{\mathsf{N}}\rangle ,
\end{equation}%
while using the right action of the operator $\mathcal{C}(\eta _{a}|\tau )$
and the orthogonality of right and left pseudo $\mathsf{D}$-eigenstates
corresponding to different eigenvalues we get:%
\begin{align}
x_{a}=& \prod_{b\neq a}\frac{\theta (\xi _{a}-\xi _{b}+\eta h_{b})}{\theta
(\xi _{a}-\eta -\xi _{b}+\eta h_{b})}\text{\textsc{d}}(\xi _{a}-\eta
)\langle h_{1},...,h_{a}=0,...,h_{\mathsf{N}}|h_{1},...,h_{a}=0,...,h_{%
\mathsf{N}}\rangle  \notag \\
& \times \frac{\theta (t_{\text{\textbf{h}}\mathbf{(}h_{a}=1)}-\eta )}{%
\theta (t_{\text{\textbf{h}}\mathbf{(}h_{a}=1)})},
\end{align}%
and so:%
\begin{equation*}
\frac{\langle h_{1},...,h_{a}=1,...,h_{\mathsf{N}}|h_{1},...,h_{a}=1,...,h_{%
\mathsf{N}}\rangle }{\langle h_{1},...,h_{a}=0,...,h_{\mathsf{N}%
}|h_{1},...,h_{a}=0,...,h_{\mathsf{N}}\rangle }=\frac{\theta (t_{\text{%
\textbf{h}}\mathbf{(}h_{a}=0)})}{\theta (t_{\text{\textbf{h}}\mathbf{(}%
h_{a}=1)})}\prod_{b\neq a,b=1}^{\mathsf{N}}\frac{\theta (\xi _{a}^{(0)}-\xi
_{b}^{(h_{b})})}{\theta (\xi _{a}^{(1)}-\xi _{b}^{(h_{b})})},
\end{equation*}%
from which the proposition simply follows when we use the identity\footnote{%
See for example Proposition 4 of \cite{FarXYZPRL-08} for a proof of it.}:%
\begin{equation}
\det_{\mathsf{N}}\Theta _{ij}^{\left( \text{\textbf{h}}\right) }=c_{\mathsf{N%
}}\text{ }\theta (\sum_{a=1}^{\mathsf{N}}\bar{\xi}_{a}^{(h_{a})}-\frac{%
\mathsf{N}-1}{2})\prod_{1\leq a<b\leq \mathsf{N}}\theta (\xi
_{a}^{(h_{a})}-\xi _{b}^{(h_{b})}),  \label{FarXYZR-Pseudo-Vendermonde}
\end{equation}%
and we recall that:%
\begin{equation}
t_{\text{\textbf{h}}}=-\sum_{a=1}^{\mathsf{N}}\bar{\xi}_{a}^{(h_{a})}+\frac{%
\mathsf{N}-1}{2},
\end{equation}%
and that the choice of the normalization \textsc{n} implies:%
\begin{equation}
\langle \text{\textbf{0}}|\text{\textbf{0}}\rangle =\frac{1}{\det_{\mathsf{N}%
}\Theta _{ij}^{\left( \text{\textbf{0}}\right) }}.
\end{equation}
\end{proof}

\subsection{\label{FarXYZT^6VD-spectrum}SOV characterization of $\overline{%
\mathcal{T}}^{\mathsf{(6VD)}}$-spectrum}

Let us denote with $\Sigma _{\overline{\mathcal{T}}^{\mathsf{(6VD)}}}$ the
set of the eigenvalue functions $\mathsf{t}_{\mathsf{6VD}}(\lambda )$ of the
transfer matrix $\overline{\mathcal{T}}^{\mathsf{(6VD)}}(\lambda |\tau )$,
then the following characterization of the $\overline{\mathcal{T}}^{\mathsf{%
(6VD)}}$-spectrum (eigenvalues \& eigenstates) in quantum separation of
variables holds:

\begin{theorem}
\label{FarXYZC:T-eigenstates}For any fixed $\mathsf{N}$-tuple of
inhomogeneities $\{\xi _{1},...,\xi _{\mathsf{N}}\}\in \mathbb{C}$ $^{%
\mathsf{N}}$ satisfying $\left( \ref{FarXYZE-SOV}\right) $ the spectrum of $%
\overline{\mathcal{T}}^{\mathsf{(6VD)}}(\lambda |\tau )$ in $\mathbb{\bar{D}}%
_{\mathsf{(6VD)},\mathsf{N}}^{\mathcal{L}/\mathcal{R}}$ is simple and $%
\Sigma _{{\overline{\mathcal{T}}^{\mathsf{(6VD)}}}}$ coincides with the set
of functions of the form:%
\begin{equation}
\mathsf{t}_{\mathsf{6VD}}(\lambda )=\sum_{a=1}^{\mathsf{N}}\frac{\theta (t_{%
\text{\textbf{0}}}-\lambda +\xi _{a})}{\theta (t_{\text{\textbf{0}}})}%
\prod_{b\neq a}\frac{\theta (\lambda -\xi _{b})}{\theta (\xi _{a}-\xi _{b})}%
\mathsf{t}_{\mathsf{6VD}}(\xi _{a})  \label{FarXYZset-t}
\end{equation}%
which are solutions of the discrete system of equations:%
\begin{equation}
\mathsf{t}_{\mathsf{6VD}}(\xi _{a}^{(0)})\mathsf{t}_{\mathsf{6VD}}(\xi
_{a}^{(1)})=\text{\textsc{a}}(\xi _{a}^{(0)})\text{\textsc{d}}(\xi
_{a}^{(1)}),\text{ \ \ }\forall a\in \{1,...,\mathsf{N}\}.
\label{FarXYZI-Functional-eq}
\end{equation}

\begin{itemize}
\item[\textsf{I)}] The right $\overline{\mathcal{T}}^{\mathsf{(6VD)}}$%
-eigenstate corresponding to $\mathsf{t}_{\mathsf{6VD}}(\lambda )\in \Sigma
_{\overline{\mathcal{T}}^{\mathsf{(6VD)}}}$ is characterized by:%
\begin{equation}
|\mathsf{t}_{\mathsf{6VD}}\rangle =\sum_{h_{1},...,h_{\mathsf{N}%
}=0}^{1}\prod_{a=1}^{\mathsf{N}}Q_{\mathsf{t}}(\xi _{a}^{(h_{a})})\det_{%
\mathsf{N}}\Theta _{ij}^{\left( \text{\textbf{h}}\right) }\text{ }%
|h_{1},...,h_{\mathsf{N}}\rangle ,  \label{FarXYZeigenT-r-D}
\end{equation}%
up to an overall normalization, where the coefficients are characterized by:%
\begin{equation}
Q_{\mathsf{t}}(\xi _{a}^{(1)})/Q_{\mathsf{t}}(\xi _{a}^{(0)})=\mathsf{t}_{%
\text{$6VD$}}(\xi _{a}^{(0)})/\text{\textsc{d}}(\xi _{a}^{(1)}).
\label{FarXYZt-Q-relation}
\end{equation}

\item[\textsf{II)}] The left $\overline{\mathcal{T}}^{\mathsf{(6VD)}}$%
-eigenstate corresponding to $\mathsf{t}_{\mathsf{6VD}}(\lambda )\in \Sigma
_{\overline{\mathcal{T}}^{\mathsf{(6VD)}}}$ is characterized by:%
\begin{equation}
\langle \mathsf{t}_{\mathsf{6VD}}|=\sum_{h_{1},...,h_{\mathsf{N}%
}=0}^{1}\prod_{a=1}^{\mathsf{N}}\bar{Q}_{\mathsf{t}}(\xi
_{a}^{(h_{a})})\det_{\mathsf{N}}\Theta _{ij}^{\left( \text{\textbf{h}}%
\right) }\text{ }\langle h_{1},...,h_{\mathsf{N}}|,  \label{FarXYZeigenT-l-D}
\end{equation}%
up to an overall normalization, where the coefficients are characterized by:%
\begin{equation}
\bar{Q}_{\mathsf{t}}(\xi _{a}^{(1)})/\bar{Q}_{\mathsf{t}}(\xi _{a}^{(0)})=%
\mathsf{t}_{\mathsf{6VD}}(\xi _{a}^{(0)})/\text{\textsc{a}}(\xi _{a}^{(0)}).
\label{FarXYZt-Qbar-relation}
\end{equation}
\end{itemize}
\end{theorem}

\begin{proof}
Let $\left\langle \mathsf{t}_{\mathsf{6VD}}\right\vert $ be a $\overline{%
\mathcal{T}}^{\mathsf{(6VD)}}$-eigenstate corresponding to the ${\overline{%
\mathcal{T}}^{\mathsf{(6VD)}}}$-eigenvalue $\mathsf{t}_{\mathsf{6VD}%
}(\lambda )$, then the coefficients (\textit{wave-functions}): 
\begin{equation}
\Psi _{t}(\text{\textbf{h}})\equiv \left\langle \mathsf{t}_{\mathsf{6VD}%
}\right\vert h_{1},...,h_{\mathsf{N}}\rangle
\end{equation}%
of $\left\langle \mathsf{t}_{\mathsf{6VD}}\right\vert $ in the right
SOV-basis satisfy the equations: 
\begin{equation}
\mathsf{t}_{\mathsf{6VD}}(\xi _{n}^{(h_{n})})\Psi _{t}(\text{\textbf{h}}%
)\,=\,\text{\textsc{a}}(\xi _{n}^{(h_{n})})\Psi _{t}(\mathsf{T}_{n}^{+}(%
\text{\textbf{h}}))+\text{\textsc{d}}(\xi _{n}^{(h_{n})})\Psi _{t}(\mathsf{T}%
_{n}^{-}(\text{\textbf{h}})),  \label{FarXYZSOVBax1}
\end{equation}%
for any$\,n\in \{1,...,\mathsf{N}\}$ and \textbf{h}$\in \{0,1\}^{\mathsf{N}}$%
, where we have denoted:%
\begin{equation}
\mathsf{T}_{n}^{\pm }(\text{\textbf{h}})\equiv (h_{1},\dots ,h_{n}\pm
1,\dots ,h_{\mathsf{N}}).
\end{equation}%
These equations are obtained by computing the matrix elements:%
\begin{equation}
\left\langle \mathsf{t}_{\mathsf{6VD}}\right\vert \overline{\mathcal{T}}^{%
\mathsf{(6VD)}}(\xi _{n}^{(h_{n})}|\tau )|h_{1},...,h_{\mathsf{N}}\rangle .
\end{equation}%
In particular, acting with $\overline{\mathcal{T}}^{\mathsf{(6VD)}}(\xi
_{n}^{(h_{n})}|\tau )$ on $\left\langle \mathsf{t}_{\mathsf{6VD}}\right\vert 
$ the l.h.s. of $\left( \ref{FarXYZSOVBax1}\right) $ is reproduced while
acting on $|h_{1},...,h_{\mathsf{N}}\rangle $ by using the SOV
representation of $\overline{\mathcal{T}}^{\mathsf{(6VD)}}(\xi
_{n}^{(h_{n})}|\tau )$ the r.h.s. of $\left( \ref{FarXYZSOVBax1}\right) $ is
reproduced. Moreover, the representation $\left( \ref{FarXYZset-t}\right) $
for the $\overline{\mathcal{T}}^{\mathsf{(6VD)}}$-eigenvalue functions $%
\mathsf{t}_{\mathsf{6VD}}(\lambda )$ follows by computing the matrix element:%
\begin{equation}
\left\langle \mathsf{t}_{\mathsf{6VD}}\right\vert \overline{\mathcal{T}}^{%
\mathsf{(6VD)}}(\lambda |\tau )|\text{\textbf{0}}\rangle
\end{equation}%
and by using $\left( \ref{FarXYZSOVBax1}\right) $ to rewrite the r.h.s. in
the desired form.

Then in the SOV representations the spectral problem for $\overline{\mathcal{%
T}}^{\mathsf{(6VD)}}(\lambda |\tau )$ is reduced to a discrete system of $2^{%
\mathsf{N}}$ Baxter-like equations $\left( \ref{FarXYZSOVBax1}\right) $ in
the class of function of the form $\left( \ref{FarXYZset-t}\right) $. Taking
into account the identities:%
\begin{equation}
\text{\textsc{a}}(\xi _{n}^{(1)})=\text{\textsc{d}}(\xi _{n}^{(0)})=0,
\end{equation}%
this system coincides with a system of homogeneous equations: 
\begin{equation}
\left( 
\begin{array}{cc}
\mathsf{t}_{\mathsf{6VD}}(\xi _{n}^{(0)}) & -\text{\textsc{a}}(\xi
_{n}^{(0)}) \\ 
-\text{\textsc{d}}(\xi _{n}^{(1)}) & \mathsf{t}_{\mathsf{6VD}}(\xi
_{n}^{(1)})%
\end{array}%
\right) \left( 
\begin{array}{c}
\Psi _{t}(h_{1},...,h_{n}=0,...,h_{1}) \\ 
\Psi _{t}(h_{1},...,h_{n}=1,...,h_{1})%
\end{array}%
\right) =\left( 
\begin{array}{c}
0 \\ 
0%
\end{array}%
\right) ,  \label{FarXYZhomo-system}
\end{equation}%
for any$\,n\in \{1,...,\mathsf{N}\}$ with $h_{m\neq n}\in \{0,1\}$. The
condition $\mathsf{t}_{\mathsf{6VD}}(\lambda )\in \Sigma _{{\overline{%
\mathcal{T}}^{\mathsf{(6VD)}}}}$ is then equivalent to the requirement that
the determinants of the $2\times 2$ matrices in $\left( \ref%
{FarXYZhomo-system}\right) $ must be zero for any$\,n\in \{1,...,\mathsf{N}%
\} $, i.e. the equation (\ref{FarXYZI-Functional-eq}). On the other hand
being: 
\begin{equation}
\text{\textsc{a}}(\xi _{n}^{(0)})\neq 0\text{\ \ and \ \textsc{d}}(\xi
_{n}^{(1)})\neq 0,  \label{FarXYZRank1}
\end{equation}%
the rank of the matrices in $\left( \ref{FarXYZhomo-system}\right) $ is 1
and then up to an overall normalization the solution is unique:%
\begin{equation}
\frac{\Psi _{t}(h_{1},...,h_{n}=1,...,h_{1})}{\Psi
_{t}(h_{1},...,h_{n}=0,...,h_{1})}=\frac{\mathsf{t}_{\mathsf{6VD}}(\xi
_{n}^{(0)})}{\text{\textsc{a}}(\xi _{n}^{(0)})},
\end{equation}%
for any$\,n\in \{1,...,\mathsf{N}\}$ with $h_{m\neq n}\in \{0,1\}$. This
implies that given a $\mathsf{t}_{\mathsf{6VD}}(\lambda )\in \Sigma _{{%
\overline{\mathcal{T}}^{\mathsf{(6VD)}}}}$ there exist (up to normalization)
one and only one corresponding $\overline{\mathcal{T}}^{\mathsf{(6VD)}}$%
-eigenstate $\left\langle \mathsf{t}_{\mathsf{6VD}}\right\vert $ with
coefficients which have the factorized form given in $\left( \ref%
{FarXYZeigenT-l-D}\right) $-$\left( \ref{FarXYZt-Qbar-relation}\right) $ and
then the $\overline{\mathcal{T}}^{\mathsf{(6VD)}}$-spectrum is simple. The
proof for the right eigenstates is given in a similar way.
\end{proof}

Let us remark that the previous theorem completely characterize the spectrum
of the transfer matrix $\overline{\mathcal{T}}^{\mathsf{(6VD)}}(\lambda
|\tau )$ in $\mathbb{\bar{D}}_{\mathsf{(6VD)},\mathsf{N}}^{\mathcal{L}/%
\mathcal{R}}$. However, a reformulation of the SOV characterization of the $%
\overline{\mathcal{T}}^{\mathsf{(6VD)}}$-spectrum by functional equations
can be important for practical aims. One standard way to accomplish this
result is by the construction of a Baxter Q-operator whose functional
equation reduces to the finite system of Baxter-like equations $\left(\ref%
{FarXYZSOVBax1}\right) $ when computed in the eigenvalues of the quantum separate variables, i.e. the operators zero of $\mathsf{D}(\la|\tau)$. This construction is currently under analysis and it can be
developed along the same lines presented in \cite{FarXYZBBOY95} for the
antiperiodic XXZ spin-1/2 chain for general values of the coupling $\eta $.
For the elliptic roots of unit case, we can construct the functional
equation directly by using cyclic representations for the operator $\tau $;
these interesting issues will be developed in a forthcoming paper. Finally, let us remark that once the Q-operator is constructed its eigenvalues can be used to completely characterize not only the eigenvalues but also the eigenstates of the antiperiodic dynamical 6-vertex transfer matrix in the SOV framework. Indeed, it is enough to introduce in formula (2.105) for the $Q_{\mathsf{t}}(\xi_{a}^{(h_{a})})$ the eigenvalue of the Q-operator computed in $\xi_{a}^{(h_{a})}$ to get the corresponding simultaneous eigenstate.

\subsection{\label{FarXYZseparate-states}Action of left separate states on
right separate states}

As for the others quantum integrable models analyzed by SOV in \cite%
{FarXYZGMN12-SG}-\cite{FarXYZN12-2}, a special role is played by the left
and right \textit{separate states} in the SOV representations. These are
states which have factorized coefficients in the SOV representations similar
to those of the transfer matrix eigenstates; more in detail we say that a
covector $\langle \alpha |\in \mathbb{\bar{D}}_{\mathsf{(6VD)},\mathsf{N}}^{%
\mathcal{L}}$ and a vector $|\beta \rangle \in \mathbb{\bar{D}}_{\mathsf{%
(6VD)},\mathsf{N}}^{\mathcal{R}}$ are \textit{separate states} if they admit
the following SOV-decompositions:%
\begin{align}
\langle \alpha |& =\sum_{h_{1},...,h_{\mathsf{N}}=0}^{1}\prod_{a=1}^{\mathsf{%
N}}\alpha _{a}(\xi _{a}^{(h_{a})})\det_{\mathsf{N}}\Theta _{ij}^{\left( 
\text{\textbf{h}}\right) }\langle h_{1},...,h_{\mathsf{N}}|,
\label{FarXYZFact-left-SOV} \\
|\beta \rangle & =\sum_{h_{1},...,h_{\mathsf{N}}=0}^{1}\prod_{a=1}^{\mathsf{N%
}}\beta _{a}(\xi _{a}^{(h_{a})})\det_{\mathsf{N}}\Theta _{ij}^{\left( \text{%
\textbf{h}}\right) }|h_{1},...,h_{\mathsf{N}}\rangle .
\label{FarXYZFact-right-SOV}
\end{align}%
The main interest toward these states is the simple determinant form for the
action of left separate states on the right ones characterized by the
following:

\begin{proposition}
\label{FarXYZOrthogonality-6VD}The left $\langle \alpha|$ and the right $%
|\beta\rangle$ separate states satisfy the identities: 
\begin{equation}
\langle \alpha |\beta \rangle =\det_{\mathsf{N}}||\mathcal{F}_{a,b}^{\left(
\alpha ,\beta \right) }||\text{ \ \ with \ }\mathcal{F}_{a,b}^{\left( \alpha
,\beta \right) }\equiv \sum_{h=0}^{1}\alpha _{a}(\xi _{a}^{(h)})\beta
_{a}(\xi _{a}^{(h)})\vartheta _{b-1}\left( \bar{\xi}_{a}^{(h)}\right) .
\end{equation}%
Then the matrix elements of the left and right $\overline{\mathcal{T}}^{%
\mathsf{(6VD)}}$-eigenstates corresponding to generic eigenvalues:%
\begin{equation}
\mathsf{t}_{\mathsf{6VD}}(\lambda )\in \Sigma _{\overline{\mathcal{T}}^{%
\mathsf{(6VD)}}}\text{, \ \ \ }\mathsf{t}_{\mathsf{6VD}}^{\prime }(\lambda
)\in \Sigma _{\overline{\mathcal{T}}^{\mathsf{(6VD)}}},
\end{equation}%
have the following simple form:%
\begin{equation}
\langle \mathsf{t}_{\mathsf{6VD}}|\mathsf{t}_{\mathsf{6VD}}^{\prime }\rangle
=\delta _{\mathsf{t},\mathsf{t}^{\prime }}\det_{\mathsf{N}}||\mathcal{F}%
_{a,b}^{(\mathsf{t},\mathsf{t})}||,\text{ \ }\mathcal{F}_{a,b}^{(\mathsf{t},%
\mathsf{t})}\equiv \sum_{l=0}^{1}\bar{Q}_{\mathsf{t}}(\xi _{a}^{(l)})Q_{%
\mathsf{t}}(\xi _{a}^{(l)})\vartheta _{b-1}\left( \bar{\xi}_{a}^{(l)}\right)
,  \label{FarXYZ6vd-scalar-p.}
\end{equation}%
where we have used the notation $\delta _{\mathsf{t},\mathsf{t}^{\prime
}}=\{0$\, for $\mathsf{t}_{\mathsf{6VD}}(\lambda )\neq \mathsf{t}_{\mathsf{%
6VD}}^{\prime }(\lambda )\in \Sigma _{\overline{\mathcal{T}}^{\mathsf{(6VD)}%
}},$ $1$\, for $\mathsf{t}_{\mathsf{6VD}}(\lambda )=\mathsf{t}_{\mathsf{6VD}%
}^{\prime }(\lambda )\in \Sigma _{\overline{\mathcal{T}}^{\mathsf{(6VD)}}}\}$%
.
\end{proposition}

\begin{proof}
From the SOV-decomposition, we have:%
\begin{equation}
\langle \alpha |\beta \rangle =\sum_{h_{1},...,h_{\mathsf{N}}=0}^{1}\det_{%
\mathsf{N}}\Theta _{ij}^{\left( \text{\textbf{h}}\right) }\prod_{a=1}^{%
\mathsf{N}}\alpha _{a}(\xi _{a}^{(h_{a})})\beta _{a}(\xi _{a}^{(h_{a})}),
\end{equation}%
from this formula by using the multilinearity of the determinant w.r.t. the
rows we prove the first identity in the proposition. The presence of the
delta in (\ref{FarXYZ6vd-scalar-p.}) simply follows from the identities:%
\begin{equation}
\mathsf{t}_{\mathsf{6VD}}(\lambda )\langle \mathsf{t}_{\mathsf{6VD}}|\mathsf{%
t}_{\mathsf{6VD}}^{\prime }\rangle =\langle \mathsf{t}_{\mathsf{6VD}}|%
\overline{\mathcal{T}}^{\mathsf{(6VD)}}(\lambda |\tau )|\mathsf{t}_{\mathsf{%
6VD}}^{\prime }\rangle =\mathsf{t}_{\mathsf{6VD}}^{\prime }(\lambda )\langle 
\mathsf{t}_{\mathsf{6VD}}|\mathsf{t}_{\mathsf{6VD}}^{\prime }\rangle
\end{equation}%
where the l.h.s. is obtained acting on the left state with $\overline{%
\mathcal{T}}^{\mathsf{(6VD)}}(\lambda |\tau )$ and the r.h.s. is obtained
acting on the right state with $\overline{\mathcal{T}}^{\mathsf{(6VD)}%
}(\lambda |\tau )$. Indeed for $\mathsf{t}_{\mathsf{6VD}}(\lambda )\neq 
\mathsf{t}_{\mathsf{6VD}}^{\prime }(\lambda )$ the identity implies:%
\begin{equation}
\langle \mathsf{t}_{\mathsf{6VD}}|\mathsf{t}_{\mathsf{6VD}}^{\prime }\rangle
=0.
\end{equation}%
Finally, for $\mathsf{t}_{\mathsf{6VD}}(\lambda )=\mathsf{t}_{\mathsf{6VD}%
}^{\prime }(\lambda ) $ the form of the matrix elements just follow
recalling that the eigenstates of the transfer matrix are indeed separate
states.
\end{proof}

\subsection{Decomposition of the identity in left and right separate basis}

The results of the previous section allow us to write the decomposition of
the identity in left and right basis of separate states. In order to define
such basis let us introduce the following natural isomorphism between the
sets $\{0,1\}^{\mathsf{N}}$ and $\{1,...,2^{\mathsf{N}}\}$:%
\begin{equation}
\varkappa :\text{\textbf{h}}\equiv \{h_{1},...,h_{\mathsf{N}}\}\in \{0,1\}^{%
\mathsf{N}}\rightarrow \varkappa \left( \text{\textbf{h}}\right) \equiv
1+\sum_{a=1}^{\mathsf{N}}2^{(a-1)}h_{a}\in \{1,...,2^{\mathsf{N}}\},
\label{FarXYZADMFKcorrisp}
\end{equation}%
and let us denote with $\varkappa _{a}^{-1}\left( i\right) \in \{0,1\}$ the
entry $a\in \{1,...,\mathsf{N}\}$ in the $\mathsf{N}$-tuple $\varkappa
^{-1}\left( i\right) $ associated by $\varkappa ^{-1}$ to any integer $i\in
\{1,...,2^{\mathsf{N}}\}$. Then, under the conditions:%
\begin{eqnarray}
\det_{2^{\mathsf{N}}}||M_{i,j}^{\left( \alpha \right) }|| &\neq &0,\text{ \
\ \ }M_{i,j}^{\left( \alpha \right) }\equiv \prod_{a=1}^{\mathsf{N}}\alpha
_{a}^{(j)}(\xi _{a}^{(\varkappa _{a}^{-1}\left( i\right) )})\text{ }\forall
i,j\in \{1,...,2^{\mathsf{N}}\}, \\
\det_{2^{\mathsf{N}}}||M_{i,j}^{\left( \beta \right) }|| &\neq &0,\text{ \ \
\ }M_{i,j}^{\left( \beta \right) }\equiv \prod_{a=1}^{\mathsf{N}}\beta
_{a}^{(j)}(\xi _{a}^{(\varkappa _{a}^{-1}\left( i\right) )})\text{ }\forall
i,j\in \{1,...,2^{\mathsf{N}}\},
\end{eqnarray}%
the sets of covectors $\langle \alpha _{j}|$ and vector $|\beta _{j}\rangle $
defined by:%
\begin{align}
\langle \alpha _{j}|& =\sum_{h_{1},...,h_{\mathsf{N}}=0}^{1}\prod_{a=1}^{%
\mathsf{N}}\alpha _{a}^{\left( j\right) }(\xi _{a}^{(h_{a})})\det_{\mathsf{N}%
}\Theta _{ij}^{\left( \text{\textbf{h}}\right) }\langle h_{1},...,h_{\mathsf{%
N}}|\text{ \ \ }\forall j\in \{1,...,2^{\mathsf{N}}\},
\label{FarXYZFact-left-SOV} \\
|\beta _{j}\rangle & =\sum_{h_{1},...,h_{\mathsf{N}}=0}^{1}\prod_{a=1}^{%
\mathsf{N}}\beta _{a}^{\left( j\right) }(\xi _{a}^{(h_{a})})\det_{\mathsf{N}%
}\Theta _{ij}^{\left( \text{\textbf{h}}\right) }|h_{1},...,h_{\mathsf{N}%
}\rangle \text{ \ \ }\forall j\in \{1,...,2^{\mathsf{N}}\},
\label{FarXYZFact-right-SOV}
\end{align}%
generate separate basis of $\mathbb{\bar{D}}_{\mathsf{(6VD)},\mathsf{N}}^{%
\mathcal{L}} $ and $\mathbb{\bar{D}}_{\mathsf{(6VD)},\mathsf{N}}^{\mathcal{R}%
}$, respectively. Moreover, defined:%
\begin{equation}
\text{\textbf{N}}_{\text{\textbf{h}}}\equiv \{\text{\textbf{k}}\in \{0,1\}^{%
\mathsf{N}}:\det_{\mathsf{N}}||\mathcal{F}_{a,b}^{\left( \alpha _{\varkappa
\left( \text{\textbf{h}}\right) },\beta _{\varkappa \left( \text{\textbf{k}}%
\right) }\right) }||\neq 0\},
\end{equation}%
then the following decomposition of the identity is implied on these
separate basis:%
\begin{equation}
\mathbb{I=}\sum_{\text{\textbf{h}}\in \{0,1\}^{\mathsf{N}}}\sum_{\text{%
\textbf{k}}\in \text{\textbf{N}}_{\text{\textbf{h}}}}\left( \det_{\mathsf{N}%
}||\mathcal{F}_{a,b}^{\left( \alpha _{\varkappa \left( \text{\textbf{h}}%
\right) },\beta _{\varkappa \left( \text{\textbf{k}}\right) }\right)
}||\right) ^{-1}|\beta _{\varkappa \left( \text{\textbf{k}}\right) }\rangle
\langle \alpha _{\varkappa \left( \text{\textbf{h}}\right) }|\text{,}
\label{FarXYZId-decomp-separate}
\end{equation}%
which reads:%
\begin{equation}  \label{FarXYZId-decomp-T-eigenstates}
\mathbb{I=}\sum_{\mathsf{t}(\lambda )\in \Sigma _{\overline{\mathcal{T}}^{%
\mathsf{(6VD)}}}}\left( \det_{\mathsf{N}}||\mathcal{F}_{a,b}^{(\mathsf{t},%
\mathsf{t})}||\right) ^{-1}|\mathsf{t}_{\mathsf{6VD}}\rangle \langle \mathsf{%
t}_{\mathsf{6VD}}|,
\end{equation}
for the representations for which the antiperiodic transfer matrix $%
\overline{\mathcal{T}}^{\mathsf{(6VD)}}(\lambda |\tau )$ is proven to be
diagonalizable.

\section{On the periodic 8-vertex spectrum and connection with SOV}

In this section we will analyze the connection between the spectral problem
of the periodic 8-vertex transfer matrix on chains with an odd number of
quantum sites and the one of the antiperiodic dynamical 6-vertex transfer
matrix. The Baxter's gauge transformations are used together with the
functional characterization of the 8-vertex transfer matrix to get central
information on the spectrum (eigenvalues and eigenstates) of this model by
our SOV results.

\subsection{The 8-vertex model}

Let us recall the characterization in terms of QISM of the XYZ spin-1/2
quantum chain. The 8-vertex R-matrix reads:%
\begin{equation}
R_{0a}^{\mathsf{(8V)}}(\lambda )=\left( 
\begin{array}{cccc}
\text{a}(\lambda ) & 0 & 0 & \text{d}(\lambda ) \\ 
0 & \text{b}(\lambda ) & \text{c}(\lambda ) & 0 \\ 
0 & \text{c}(\lambda ) & \text{b}(\lambda ) & 0 \\ 
\text{d}(\lambda ) & 0 & 0 & \text{a}(\lambda )%
\end{array}%
\right) ,
\end{equation}%
where: 
\begin{align}
\text{a}(\lambda )& =\frac{2\theta _{4}(\eta |2\omega )\theta _{1}(\lambda
+\eta |2\omega )\theta _{4}(\lambda |2\omega )}{\theta _{2}(0|\omega )\theta
_{4}(0|2\omega )},\quad \text{b}(\lambda )=\frac{2\theta _{4}(\eta |2\omega
)\theta _{1}(\lambda |2\omega )\theta _{4}(\lambda +\eta |2\omega )}{\theta
_{2}(0|\omega )\theta _{4}(0|2\omega )}, \\
\text{c}(\lambda )& =\frac{2\theta _{1}(\eta |2\omega )\theta _{4}(\lambda
|2\omega )\theta _{4}(\lambda +\eta |2\omega )}{\theta _{2}(0|\omega )\theta
_{4}(0|2\omega )},\quad \text{d}(\lambda )=\frac{2\theta _{1}(\eta |2\omega
)\theta _{1}(\lambda +\eta |2\omega )\theta _{1}(\lambda |2\omega )}{\theta
_{2}(0|\omega )\theta _{4}(0|2\omega )},
\end{align}%
is solution of the Yang-Baxter equation:%
\begin{equation}
R_{12}^{\mathsf{(8V)}}(\lambda _{12})R_{1a}^{\mathsf{(8V)}}(\lambda
_{1})R_{2a}^{\mathsf{(8V)}}(\lambda _{2})=R_{2a}^{\mathsf{(8V)}}(\lambda
_{2})R_{1a}^{\mathsf{(8V)}}(\lambda _{1})R_{12}^{\mathsf{(8V)}}(\lambda
_{12}).
\end{equation}%
Then the monodromy matrix of the spin-1/2 representations is defined by:%
\begin{equation}
\mathsf{M}_{0}^{\mathsf{(8V)}}(\lambda )\equiv R_{0\mathsf{N}}^{\mathsf{(8V)}%
}(\lambda -\xi _{\mathsf{N}})\cdots R_{01}^{\mathsf{(8V)}}(\lambda -\xi
_{1})\equiv \left( 
\begin{array}{cc}
\mathsf{A}^{\mathsf{(8V)}}(\lambda ) & \mathsf{B}^{\mathsf{(8V)}}(\lambda )
\\ 
\mathsf{C}^{\mathsf{(8V)}}(\lambda ) & \mathsf{D}^{\mathsf{(8V)}}(\lambda )%
\end{array}%
\right) ,
\end{equation}%
with the parameters $\xi _{a}$ which are the inhomogeneities. The monodromy
matrix $\mathsf{M}_{0}^{\mathsf{(8V)}}(\lambda )$ is itself solution of the
Yang-Baxter equation:%
\begin{equation}
R_{12}^{\mathsf{(8V)}}(\lambda _{12})\mathsf{M}_{1}^{\mathsf{(8V)}}(\lambda
_{1})\mathsf{M}_{2}^{\mathsf{(8V)}}(\lambda _{2})=\mathsf{M}_{2}^{\mathsf{%
(8V)}}(\lambda _{2})\mathsf{M}_{1}^{\mathsf{(8V)}}(\lambda _{1})R_{12}^{%
\mathsf{(8V)}}(\lambda _{12}),  \label{8v-YB-algebra}
\end{equation}%
and then the corresponding transfer matrix:%
\begin{equation}
\mathsf{T}^{\mathsf{(8V)}}(\lambda )=tr_{0}\mathsf{M}_{0}^{\mathsf{(8V)}%
}(\lambda )
\end{equation}%
defines a one parameter family of commuting operators; the Hamiltonian of
the XYZ spin-1/2 quantum chain is obtained in the homogeneous limit by:%
\begin{equation}
H_{XYZ}=2\sinh \eta \,\left. \frac{\partial \ln \mathsf{T}^{\mathsf{(8V)}%
}(\lambda )}{\partial \lambda }\right\vert _{\lambda =0,\xi _{n}=0}-\mathsf{N%
}\cosh \eta .  \label{FarXYZlogderiv}
\end{equation}

\subsection{Elementary properties of the periodic 8-vertex transfer matrix}

Let us describe some elementary properties of the periodic 8-vertex transfer
matrix which allow a first characterization of the spectrum evidencing its
connection to the antiperiodic dynamical 6-vertex spectrum in the case of
odd chains.

\begin{lemma}
In the 8-vertex Yang-Baxter algebra, we can introduce the following central
quantum determinant:%
\begin{align}
\det{}_{q}\mathsf{M}^{\mathsf{(8V)}}(\lambda )& \equiv \left( \mathsf{A}^{%
\mathsf{(8V)}}(\lambda )\mathsf{D}^{\mathsf{(8V)}}(\lambda -\eta )-\mathsf{B}%
^{\mathsf{(8V)}}(\lambda )\mathsf{C}^{\mathsf{(8V)}}(\lambda -\eta )\right)
\\
& =\left( \mathsf{D}^{\mathsf{(8V)}}(\lambda )\mathsf{A}^{\mathsf{(8V)}%
}(\lambda -\eta )-\mathsf{C}^{\mathsf{(8V)}}(\lambda )\mathsf{B}^{\mathsf{%
(8V)}}(\lambda -\eta )\right) \\
& =\text{\textsc{a}}(\lambda )\text{\textsc{d}}(\lambda -\eta ),
\end{align}%
and the following inversion formula holds:%
\begin{equation}
\left[ \mathsf{M}_{0}^{\mathsf{(8V)}}(\lambda )\right] ^{-1}=\frac{\sigma
_{0}^{y}\left[ \mathsf{M}_{0}^{\mathsf{(8V)}}(\lambda -\eta )\right]
^{t_{0}}\sigma _{0}^{y}}{\text{\textsc{a}}(\lambda )\text{\textsc{d}}%
(\lambda -\eta )}.  \label{Inv-8v-M}
\end{equation}
\end{lemma}

\begin{proof}
The centrality of the quantum determinant is a well known property in the
6-vertex Yang-Baxter algebra and it is possible to extend it to the 8-vertex
case. The proof is given by proving the statement for the generic quantum
site $n$ and then showing that the product of the local quantum determinants
reproduce the complete one. Let us introduce the notation:%
\begin{align}
\det{}_{q}R_{0n}^{\mathsf{(8V)}}(\lambda )& =\left( R_{0n}^{\mathsf{(8V)}}\right)
_{11}(\lambda )\left( R_{0n}^{\mathsf{(6VD)}}\right) _{22}(\lambda -\eta ) 
\notag \\
& -\left( R_{0n}^{\mathsf{(6VD)}}\right) _{12}(\lambda )\left( R_{0n}^{%
\mathsf{(6VD)}}\right) _{21}(\lambda -\eta ),
\end{align}
its explicit form reads:%
\begin{equation}
\det{}_{q}R_{0,n}^{\mathsf{(8V)}}(\lambda -\xi _{n})=\left( 
\begin{array}{cc}
\text{a}(\lambda )\text{b}(\lambda -\eta )-\text{d}(\lambda )\text{d}%
(\lambda -\eta ) & 0 \\ 
0 & \text{b}(\lambda )\text{a}(\lambda -\eta )-\text{c}(\lambda )\text{c}%
(\lambda -\eta )%
\end{array}%
\right) ,
\end{equation}%
then all we need to prove are the following identities:%
\begin{equation}
\text{a}(\lambda )\text{b}(\lambda -\eta )-\text{d}(\lambda )\text{d}%
(\lambda -\eta )=\text{b}(\lambda )\text{a}(\lambda -\eta )-\text{c}(\lambda
)\text{c}(\lambda -\eta )=a(\lambda -\xi _{n})a(\lambda -\xi _{n}-2\eta )
\end{equation}%
which trivially follow once we use the formulae\footnote{%
Respectively, equations 7, 10 and 2 at page 881 of \cite{FarXYZTables of
integrals}.}:%
\begin{eqnarray}
&&\theta _{1}(x+y|2\omega )\theta _{1}(x-y|2\omega )\theta _{4}^{2}(0)\left.
=\right. \theta _{3}^{2}(x|2\omega )\theta _{2}^{2}(y|2\omega )-\theta
_{2}^{2}(x|2\omega )\theta _{3}^{2}(y|2\omega ), \\
&&\theta _{4}(x+y|2\omega )\theta _{4}(x-y|2\omega )\theta _{4}^{2}(0)\left.
=\right. \theta _{4}^{2}(x|2\omega )\theta _{4}^{2}(y|2\omega )-\theta
_{1}^{2}(x|2\omega )\theta _{1}^{2}(y|2\omega ), \\
&&\theta _{1}(x|\omega )\theta _{2}(x|\omega )\left. =\right. \theta
_{1}(x+y|2\omega )\theta _{4}(x-y|2\omega )+\theta _{4}(x+y|2\omega )\theta
_{1}(x-y|2\omega ).
\end{eqnarray}%
Finally, the inversion formula $\left( \ref{Inv-8v-M}\right) $\ follows from
the quantum determinant formulae and from the identities:%
\begin{equation}
\mathsf{A}^{\mathsf{(8V)}}(\lambda )\mathsf{B}^{\mathsf{(8V)}}(\lambda -\eta
)-\mathsf{B}^{\mathsf{(8V)}}(\lambda )\mathsf{A}^{\mathsf{(8V)}}(\lambda
-\eta )=0,\text{ \ }\mathsf{D}^{\mathsf{(8V)}}(\lambda )\mathsf{C}^{\mathsf{%
(8V)}}(\lambda -\eta )-\mathsf{C}^{\mathsf{(8V)}}(\lambda )\mathsf{D}^{%
\mathsf{(8V)}}(\lambda -\eta )=0,
\end{equation}%
which directly follows from the 8-vertex Yang-Baxter equations $\left( \ref%
{8v-YB-algebra}\right) $.
\end{proof}

Moreover, it holds:

\begin{lemma}
\label{8v-AI}The following products $(\mathsf{M}^{\mathsf{(8V)}}(\xi
_{n}^{(0)}))_{h,j}\left( \mathsf{M}^{\mathsf{(8V)}}(\xi _{n}^{(1)})\right)
_{k,j}$ and $(\mathsf{M}^{\mathsf{(8V)}}(\xi _{n}^{(1)}))_{j,h}(\mathsf{M}^{%
\mathsf{(8V)}}(\xi _{n}^{(0)}))_{j,k}$ of the elements of the 8-vertex
monodromy matrix vanish for any $n\in \{1,...,\mathsf{N}\}$\ if $h=k$ and the following identities hold:%
\begin{align}
\mathsf{A}^{\mathsf{(8V)}}(\xi _{n}^{(0)})\mathsf{D}^{\mathsf{(8V)}}(\xi
_{n}^{(1)})& =-\mathsf{C}^{\mathsf{(8V)}}(\xi _{n}^{(0)})\mathsf{B}^{\mathsf{%
(8V)}}(\xi _{n}^{(1)}),  \label{8v-recomb1} \\
\mathsf{D}^{\mathsf{(8V)}}(\xi _{n}^{(0)})\mathsf{A}^{\mathsf{(8V)}}(\xi
_{n}^{(1)})& =-\mathsf{B}^{\mathsf{(8V)}}(\xi _{n}^{(0)})\mathsf{C}^{\mathsf{%
(8V)}}(\xi _{n}^{(1)}).  \label{8v-recomb2}
\end{align}
\end{lemma}

\begin{proof}
This lemma is a trivial generalization to the 8-vertex Yang-Baxter algebra
of the results known in the 6-vertex case, see for example \cite%
{FarXYZKKMNST07}. Both in the 8-vertex and 6-vertex case these results are
simple consequences of the reconstruction formulae of local operators in
terms of matrix elements of the monodromy matrix first proven for the
6-vertex case in \cite{FarXYZKitMT99}\ and then extended also to the
8-vertex case in \cite{FarXYZMaiT00}. In the 8-vertex case the
reconstructions read:%
\begin{eqnarray}
X_{n} &=&\prod_{b=1}^{n-1}\mathsf{T}^{\mathsf{(8V)}}(\xi _{b}^{(0)})\text{tr}%
_{0}(\mathsf{M}^{\mathsf{(8V)}}(\xi _{n}^{(0)})X_{0})\prod_{b=1}^{n}\frac{%
\mathsf{T}^{\mathsf{(8V)}}(\xi _{b}^{(1)})}{\det{}_{q}\mathsf{M}^{\mathsf{%
(8V)}}(\xi _{b}^{(0)})}  \label{8v-re1} \\
&=&\prod_{b=1}^{n}\mathsf{T}^{\mathsf{(8V)}}(\xi _{b}^{(0)})\frac{\text{tr}%
_{0}(\mathsf{M}^{\mathsf{(8V)}}(\xi _{n}^{(1)})\sigma
_{0}^{(y)}X_{0}^{t_{0}}\sigma _{0}^{(y)})}{\det{}_{q}\mathsf{M}^{\mathsf{%
(8V)}}(\xi _{n}^{(0)})}\prod_{b=1}^{n-1}\frac{\mathsf{T}^{\mathsf{(8V)}}(\xi
_{b}^{(1)})}{\det{}_{q}\mathsf{M}^{\mathsf{(8V)}}(\xi _{b}^{(0)})},
\label{8v-re2}
\end{eqnarray}%
where $X_{n}$ is a local operator on the quantum space $n$, i.e. it acts as
the identity on any quantum space associate to a site $m\neq n$ and as the $%
2\times 2$ matrix $X$ on the quantum space in the site $n$, while $X_{0}$ is
the $2\times 2$ matrix $X$ on the auxiliary space. Let us consider for
example the following identity:%
\begin{equation}
\left( 
\begin{array}{cc}
0 & 0 \\ 
0 & 1%
\end{array}%
\right) _{n}=X_{n}X_{n}=Y_{n}Z_{n},
\end{equation}%
where:%
\begin{equation}
X_{n}=\left( 
\begin{array}{cc}
0 & 0 \\ 
0 & 1%
\end{array}%
\right) _{n},\text{ }Y_{n}=\left( 
\begin{array}{cc}
0 & 0 \\ 
1 & 0%
\end{array}%
\right) _{n}\text{ \ and \ }Z_{n}=\left( 
\begin{array}{cc}
0 & 1 \\ 
0 & 0%
\end{array}%
\right) _{n},
\end{equation}%
then the identity $\left( \ref{8v-recomb2}\right) $ simply follows by using
for the first $X_{n}$ and the $Y_{n}$ the reconstruction $\left( \ref{8v-re1}%
\right) $ while for the second $X_{n}$ and the $Z_{n}$ the reconstruction $%
\left( \ref{8v-re2}\right) $. All the other identities in this lemma are
proven similarly by taking the product of a couple of local operators and
using for them the two reconstructions.
\end{proof}

\subsection{\label{1ch-8v-eigenvalues}On the periodic 8-vertex transfer
matrix eigenvalues}

The previous two lemmas allow to prove some preliminary characterization of
the periodic 8-vertex eigenvalues as presented in the
following proposition:

\begin{proposition}
The periodic 8-vertex transfer matrix of a chain with $\mathsf{N}$ quantum
sites satisfies the following properties:%
\begin{equation}
\mathsf{T}^{\mathsf{(8V)}}(\xi _{n}^{(0)})\mathsf{T}^{\mathsf{(8V)}}(\xi
_{n}^{(1)})=\det{}_{q}\mathsf{M}^{\mathsf{(8V)}}(\xi _{n}^{(0)})\text{ \ \ }%
\forall n\in \{1,...,\mathsf{N}\},  \label{8v-T-R1}
\end{equation}%
and%
\begin{equation}
\mathsf{T}^{\mathsf{(8V)}}(\lambda +\pi )=\left( -1\right) ^{\mathsf{N}}%
\mathsf{T}^{\mathsf{(8V)}}(\lambda ),\text{ }\mathsf{T}^{\mathsf{(8V)}%
}(\lambda +\pi \omega )=\left( -e^{-i\left( 2\lambda +\pi w\right) }\right)
^{\mathsf{N}}e^{-2i\left(t_{\text{\textbf{0}}}-\sum_{a=1}^{\mathsf{N}}\xi _{a}\right) }\mathsf{T}^{\mathsf{(8V)}}(\lambda ).  \label{8v-T-R2}
\end{equation}%
Then the 8-vertex eigenvalues are elliptic polynomials (or theta functions)
of degree $\mathsf{N}$ and character $e^{-2i\left( t_{\text{\textbf{0}}%
}-\sum_{a=1}^{\mathsf{N}}\xi _{a}\right) }$ and they admit the following
interpolation formula:%
\begin{equation}
\mathsf{t}_{\mathsf{8V}}(\lambda )=\sum_{a=1}^{\mathsf{N}}\frac{\theta (t_{%
\text{\textbf{0}}}-\lambda +\xi _{a})}{\theta (t_{\text{\textbf{0}}})}%
\prod_{b\neq a}\frac{\theta (\lambda -\xi _{b})}{\theta (\xi _{a}-\xi _{b})}%
\mathsf{t}_{\mathsf{8V}}(\xi _{a}),  \label{Eigenv8v-ch1}
\end{equation}%
where the $\mathsf{t}_{\mathsf{8V}}(\xi _{a})$ are solutions of the discrete
system of equations:%
\begin{equation}
\mathsf{t}_{\mathsf{8V}}(\xi _{a}^{(0)})\mathsf{t}_{\mathsf{8V}}(\xi
_{a}^{(1)})=\text{\textsc{a}}(\xi _{a}^{(0)})\text{\textsc{d}}(\xi
_{a}^{(1)}),\text{ \ \ }\forall a\in \{1,...,\mathsf{N}\}.
\label{Eigenv8v-ch2}
\end{equation}
\end{proposition}

\begin{proof}
The Lemma \ref{8v-AI} in particular implies the annihilation identities:%
\begin{equation}
\mathsf{A}^{\mathsf{(8V)}}(\xi _{n}^{(0)})\mathsf{A}^{\mathsf{(8V)}}(\xi
_{n}^{(1)})=\mathsf{D}^{\mathsf{(8V)}}(\xi _{n}^{(0)})\mathsf{D}^{\mathsf{%
(8V)}}(\xi _{n}^{(1)})=0,
\end{equation}%
from which we can write:%
\begin{equation}
\mathsf{T}^{\mathsf{(8V)}}(\xi _{n}^{(0)})\mathsf{T}^{\mathsf{(8V)}}(\xi
_{n}^{(1)})=\mathsf{A}^{\mathsf{(8V)}}(\xi _{n}^{(0)})\mathsf{D}^{\mathsf{%
(8V)}}(\xi _{n}^{(1)})+\mathsf{D}^{\mathsf{(8V)}}(\xi _{n}^{(0)})\mathsf{A}^{%
\mathsf{(8V)}}(\xi _{n}^{(1)})
\end{equation}
and then eliminating in the above equation $\mathsf{D}^{\mathsf{(8V)}}(\xi
_{n}^{(0)})\mathsf{A}^{\mathsf{(8V)}}(\xi _{n}^{(1)})$ by using $\left( \ref%
{8v-recomb1}\right) $ or $\mathsf{A}^{\mathsf{(8V)}}(\xi _{n}^{(0)})\mathsf{D%
}^{\mathsf{(8V)}}(\xi _{n}^{(1)})$ by using $\left( \ref{8v-recomb2}\right) $
we get the identity $\left( \ref{8v-T-R1}\right) $.

Let us observe now that by using the identities\footnote{%
See the equations 8.182-1, 8.182-3 and 8.183-5, 8.183-6 at page 878 of \cite%
{FarXYZTables of integrals}.}:%
\begin{eqnarray}
&&\theta _{1}(x+\pi |2\omega )=-\theta _{1}(x|2\omega ),\text{ \ }\theta
_{1}(x+\pi \omega |2\omega )=ie^{-i\left( \lambda +\pi w/2\right) }\theta
_{4}(x|2\omega ) \\
&&\theta _{4}(x+\pi |2\omega )=\theta _{4}(x|2\omega ),\text{ \ \ \ }\theta
_{4}(x+\pi \omega |2\omega )=ie^{-i\left( \lambda +\pi w/2\right) }\theta
_{1}(x|2\omega )
\end{eqnarray}%
it is simple to show that the coefficients of the 8-vertex $R$-matrix
satisfy the following transformation properties:%
\begin{align}
\text{a}(\lambda +\pi \omega )& =-e^{-i\left( 2\lambda +\pi w\right)
}e^{-i\eta }\text{b}(\lambda ),\text{ \ a}(\lambda +\pi )=-\text{a}(\lambda
),\text{ \ d}(\lambda +\pi )=\text{d}(\lambda ), \\
\text{c}(\lambda +\pi \omega )& =-e^{-i\left( 2\lambda +\pi w\right)
}e^{-i\eta }\text{d}(\lambda ),\text{ \ b}(\lambda +\pi )=-\text{b}(\lambda
),\text{ \ c}(\lambda +\pi )=\text{c}(\lambda ),
\end{align}%
which are equivalent to the following identities on the 8-vertex $R$-matrix:%
\begin{equation}
R_{0a}^{\mathsf{(8V)}}(\lambda +\pi \omega )=-e^{-i\left( 2\lambda +\pi
w\right) }e^{-i\eta }\sigma _{0}^{(x)}R_{0a}^{\mathsf{(8V)}}(\lambda )\sigma
_{0}^{(x)},\text{ \ }R_{0a}^{\mathsf{(8V)}}(\lambda +\pi )=-\sigma
_{0}^{(z)}R_{0a}^{\mathsf{(8V)}}(\lambda )\sigma _{0}^{(z)}.
\end{equation}%
Then, the monodromy matrix satisfy the identities:%
\begin{eqnarray}
\mathsf{M}_{0}^{\mathsf{(8V)}}(\lambda +\pi w) &=&\left( -e^{-i\left(
2\lambda +\pi w\right) }\right) ^{\mathsf{N}}e^{-2i\left( t_{\text{\textbf{0}%
}}-\sum_{a=1}^{\mathsf{N}}\xi _{a}\right) }\sigma _{0}^{(x)}\mathsf{M}_{0}^{%
\mathsf{(8V)}}(\lambda )\sigma _{0}^{(x)}, \\
\mathsf{M}_{0}^{\mathsf{(8V)}}(\lambda +\pi ) &=&\left( -1\right) ^{\mathsf{N%
}}\sigma _{0}^{(z)}\mathsf{M}_{0}^{\mathsf{(8V)}}(\lambda )\sigma _{0}^{(z)},
\end{eqnarray}%
from which $\left( \ref{8v-T-R1}\right) $ follows by the cyclicity of the
trace. The formula $\left( \ref{Eigenv8v-ch1}\right) $ is the interpolation formula 
\cite{FarXYZPRL-08}\ for elliptic polynomials of degree $\mathsf{N}$ and
character $e^{-2i\left( t_{\text{\textbf{0}}}-\sum_{a=1}^{\mathsf{N}}\xi
_{a}\right) }$ while $\left( \ref{Eigenv8v-ch2}\right) $ is the rewiting of $%
\left( \ref{8v-T-R1}\right) $ for the periodic 8-vertex eigenvalues.
\end{proof}

\textbf{Remark 1.} Let us notice that the above results hold for both the
even and the odd quantum chains. In the odd case these results implies that
the set of the periodic 8-vertex transfer matrix eigenvalues is contained in
the set of the antiperiodic dynamical 6-vertex transfer matrix eigenvalues.
As we have proven that the antiperiodic dynamical 6-vertex transfer matrix
has simple spectrum for general values of the inhomogeneities, then
differences in these sets of eigenvalues can be only produced from a
degeneracy of the periodic 8-vertex transfer matrix eigenvalues. A
preliminary analysis based on direct diagonalization shows that the periodic 8-vertex transfer matrix spectrum is
double degenerate for $\mathsf{N}=1,$ $3$.\ These cases are considered in
the appendix where it is also given a direct verification of the statement
proven in Theorem \ref{FarXYZC:T-eigenstates} that the system of equations
defined by $\left( \ref{FarXYZset-t}\right) $ and $\left( \ref%
{FarXYZI-Functional-eq}\right) $ characterize the complete set of the
antiperiodic dynamical 6-vertex transfer matrix eigenvalues.

\subsection{Gauge transformation from 8-vertex to dynamical 6-vertex models}

In the case of an even chain the spectral problem of the 8-vertex transfer
matrix $\mathsf{T}^{\mathsf{(8V)}}(\lambda )$ has been reduced to the one of
the periodic dynamical 6-vertex transfer matrix by the gauge transformations
introduced by Baxter in \cite{FarXYZBa72-2}. In detail the following gauge
transformation exists: 
\begin{equation}\label{8V-6VD-GT0}
R_{0a}^{\mathsf{(8V)}}(\lambda _{12})S_{0}(\lambda _{1}|\tau )S_{a}(\lambda
_{2}|\tau +\eta \sigma _{0}^{z})=S_{a}(\lambda _{2}|\tau )S_{0}(\lambda
_{1}|\tau +\eta \sigma _{a}^{z})R_{0a}^{\mathsf{(6VD)}}(\lambda _{12}|\tau ),
\end{equation}%
which for the monodromy matrices reads:%
\begin{equation}
\mathsf{M}_{0}^{\mathsf{(8V)}}(\lambda )S_{0}(\lambda |\tau )S_{q}(\tau
+\eta \sigma _{0}^{z})=S_{q}(\tau )S_{0}(\lambda |\tau +\eta \text{$\mathsf{S%
}$})\mathsf{M}_{0}^{\mathsf{(6VD)}}(\lambda |\tau ),  \label{FarXYZp-gauge}
\end{equation}%
where:%
\begin{equation}
S_{0}(\lambda |\tau )\equiv \left( 
\begin{array}{cc}
\theta _{2}(-\lambda +\tau |2w) & \theta _{2}(\lambda +\tau |2w) \\ 
\theta _{3}(-\lambda +\tau |2w) & \theta _{3}(\lambda +\tau |2w)%
\end{array}%
\right) _{0},
\end{equation}%
and:%
\begin{equation}
S_{q}(\tau )\equiv S_{1}(\xi _{1}|\tau )\cdots S_{\mathsf{N}}(\xi _{\mathsf{N%
}}|\tau +\eta \sum_{a=1}^{\mathsf{N}-1}\sigma _{a}^{z}).
\end{equation}
Then we can prove:

\begin{lemma}
\label{FarXYZsimilarity-L}In a chain with an odd number of quantum sites the
periodic 8-vertex transfer matrix has the following right action on the
states of $\mathbb{\bar{D}}_{\mathsf{(6VD)},\mathsf{N}}^{\mathcal{R}}$:%
\begin{equation}
\mathsf{T}^{\mathsf{(8V)}}(\lambda )S_{q}(\tau )=S_{q}(\tau -\eta )\mathsf{C}%
(\lambda |\tau -\eta )+S_{q}(\tau +\eta )\mathsf{B}(\lambda |\tau +\eta ).
\label{FarXYZP-ris-r}
\end{equation}
\end{lemma}

\begin{proof}
To prove (\ref{FarXYZP-ris-r}) let us first rewrite the gauge transformation
(\ref{FarXYZp-gauge}) as it follows:%
\begin{align}
& S_{0}(\lambda +\eta |\tau )S_{q}(\tau +\eta \sigma _{0}^{z})\left( 
\begin{array}{cc}
\mathsf{D}(\lambda |\tau +\eta ) & -\mathsf{B}(\lambda |\tau +\eta ) \\ 
-\mathsf{C}(\lambda |\tau -\eta ) & \mathsf{A}(\lambda |\tau -\eta )%
\end{array}%
\right) \frac{\theta (\tau +\eta \text{$\mathsf{S}$})}{\theta (\tau )}\left.
=\right.  \notag \\
& \text{ \ \ \ \ \ \ \ \ \ \ \ \ \ \ \ \ \ \ \ \ \ \ \ \ \ \ \ }\left.
=\right. \left( 
\begin{array}{cc}
\mathsf{D}^{\mathsf{(8V)}}(\lambda ) & -\mathsf{B}^{\mathsf{(8V)}}(\lambda )
\\ 
-\mathsf{C}^{\mathsf{(8V)}}(\lambda ) & \mathsf{A}^{\mathsf{(8V)}}(\lambda )%
\end{array}%
\right) S_{q}(\tau )S_{0}(\lambda +\eta |\tau +\eta \text{$\mathsf{S}$}),
\label{FarXYZStep-1-g2}
\end{align}

obtained by multiply both sides of (\ref{FarXYZp-gauge}) from the right by
the inverse of $\mathsf{M}_{0}^{\mathsf{(6VD)}}(\lambda _{1}|\tau )$, as
defined in $\left( \ref{FarXYZRight-1-dyn-Mon}\right) $, and from the left by the inverse of $\mathsf{M}_{0}^{\mathsf{(8V)}}(\lambda )$, as defined in 
$\left( \ref{Inv-8v-M}\right) $, and finally doing the change of variable $%
\lambda \rightarrow \lambda +\eta $. The gauge transformation $\left( \ref%
{FarXYZStep-1-g2}\right) $ can be further rewritten as it follows:%
\begin{align}
& S_{0}(\lambda +\eta |-\tau )S_{q}(\tau -\eta \sigma _{0}^{z})\left( 
\begin{array}{cc}
-\mathsf{C}(\lambda |\tau -\eta ) & \mathsf{A}(\lambda |\tau -\eta ) \\ 
\mathsf{D}(\lambda |\tau +\eta ) & -\mathsf{B}(\lambda |\tau +\eta )%
\end{array}%
\right) \frac{\theta (\tau +\eta \text{$\mathsf{S}$})}{\theta (\tau )}\left.
=\right.   \notag \\
& \text{ \ \ \ \ \ \ \ \ \ \ \ \ \ \ \ \ \ \ \ \ \ \ \ \ \ \ \ }\left.
=\right. \left( 
\begin{array}{cc}
\mathsf{D}^{\mathsf{(8V)}}(\lambda ) & -\mathsf{B}^{\mathsf{(8V)}}(\lambda )
\\ 
-\mathsf{C}^{\mathsf{(8V)}}(\lambda ) & \mathsf{A}^{\mathsf{(8V)}}(\lambda )%
\end{array}%
\right) S_{q}(\tau )S_{0}(\lambda +\eta |\tau +\eta \text{$\mathsf{S}$}),
\end{align}%
by using the identity:%
\begin{equation}
S_{0}(\lambda |-\tau )=S_{0}(-\lambda |\tau )=S_{0}(\lambda |\tau )\sigma
_{0}^{x}.
\end{equation}%
We can take now the trace w.r.t. the auxiliary space $0$ and we get:%
\begin{eqnarray}
\mathsf{T}^{\mathsf{(8V)}}(\lambda )S_{q}(\tau ) &=&tr_{0}\left\{ S_{q}(\tau
-\eta \sigma _{0}^{z})\left( 
\begin{array}{cc}
-\mathsf{C}(\lambda |\tau -\eta ) & \mathsf{A}(\lambda |\tau -\eta ) \\ 
\mathsf{D}(\lambda |\tau +\eta ) & -\mathsf{B}(\lambda |\tau +\eta )%
\end{array}%
\right) \right.   \notag \\
&&\times \left. \frac{\theta (\tau +\eta \text{$\mathsf{S}$})}{\theta (\tau )%
}\left[ S_{0}(\lambda +\eta |\tau +\eta \text{$\mathsf{S}$})\right]
^{-1}S_{0}(\lambda +\eta |-\tau )\right\} ,  \label{G-tr-r}
\end{eqnarray}%
where we have used the commutativity:%
\begin{equation}
\lbrack \mathsf{\bar{M}}_{0}^{\mathsf{(6VD)}}(\lambda |\tau ),\tau ]=0
\end{equation}%
and the cyclicity of the trace to move $S_{0}(\lambda +\eta |-\tau )$. It is
central to remark that $S_{0}(\lambda +\eta |\tau +\eta \mathsf{S}$$)$ is an
invertible matrix in the auxiliary space on any state of $\mathbb{\bar{D}}_{%
\mathsf{(6VD)},\mathsf{N}}^{\mathcal{R}}$ and so the identity (\ref{G-tr-r})
is well defined on $\mathbb{\bar{D}}_{\mathsf{(6VD)},\mathsf{N}}^{\mathcal{R}%
}$ and using it we get our result (\ref{FarXYZP-ris-r}) being the
eigenvalues of $\tau $ and $-\tau -\eta \mathsf{S}$ coinciding on any left
state of $\mathbb{\bar{D}}_{\mathsf{(6VD)},\mathsf{N}}^{\mathcal{R}}$.
\end{proof}

\subsection{On the periodic 8-vertex transfer matrix spectrum by SOV}

\subsubsection{Connections between periodic 8-vertex and antiperiodic
dynamical 6-vertex spectrum}

Let us define $\mathbb{\bar{P}}_{\mathsf{N}}^{\mathcal{R}}\mathbb{\equiv }%
\sum_{\mathsf{s}=-\mathsf{N}}^{\mathsf{N}}\langle t(\mathsf{s})|$, by the
definition of scalar product given in Section \ref{FarXYZRep-space}, the
action of $\mathbb{\bar{P}_{\mathsf{N}}^{\mathcal{R}}}$ reduces the
dynamical-spin vector space $\mathbb{D}_{\mathsf{(6VD)},\mathsf{N}}^{%
\mathcal{R}}$ to the 2$^{\mathsf{N}}$-dimensional spin vector space$\ 
\mathbb{S}_{\mathsf{N}}^{\mathcal{R}}=\mathbb{\bar{P}}_{\mathsf{N}}^{%
\mathcal{R}}\mathbb{\bar{D}}_{\mathsf{(6VD)},\mathsf{N}}^{\mathcal{R}}$ and
relates the generic vectors of their basis by $\mathbb{\bar{P}}_{\mathsf{N}%
}^{\mathcal{R}}\otimes _{n=1}^{\mathsf{N}}|n,h_{n}\rangle \otimes |t_{\bold{\text{h}}}\rangle =\otimes _{n=1}^{\mathsf{N}}|n,h_{n}\rangle $. Moreover, let us
introduce the pure spin operator \textsc{S}$_{q}^{\mathcal{R}}\in $ End($%
\mathbb{S}_{\mathsf{N}}^{\mathcal{R}}$) by the following actions:%
\begin{eqnarray}
\text{\textsc{S}}_{q}^{\mathcal{R}}\otimes _{n=1}^{\mathsf{N}%
}|n,h_{n}\rangle  &\equiv &S_{1}(\xi _{1}|-\frac{\eta }{2}\sum_{a=1}^{%
\mathsf{N}}\sigma _{a}^{z})\cdots S_{n}(\xi _{n}|\frac{\eta }{2}%
\sum_{a=1}^{n-1}\sigma _{a}^{z}-\frac{\eta }{2}\sum_{a=n}^{\mathsf{N}}\sigma
_{a}^{z})  \notag \\
&&\cdots S_{\mathsf{N}}(\xi _{\mathsf{N}}|\frac{\eta }{2}\sum_{a=1}^{\mathsf{%
N}-1}\sigma _{a}^{z}-\frac{\eta }{2}\sigma _{\mathsf{N}}^{z})\otimes _{n=1}^{%
\mathsf{N}}|n,h_{n}\rangle ,
\end{eqnarray}%
on the spin basis of the pure spin quantum space $\mathbb{S}_{\mathsf{N}}^{%
\mathcal{R}}$, then we have:

\begin{proposition}
\label{FarXYZProjector action}\textsf{\newline
I)} On any vector of $\mathbb{\bar{D}}_{\mathsf{(6VD)},\mathsf{N}}^{\mathcal{%
R}}$ it holds:%
\begin{equation}
\mathbb{\bar{P}}_{\mathsf{N}}^{\mathcal{R}}S_{q}(\tau )=\text{\textsc{S}}%
_{q}^{\mathcal{R}}\mathbb{\bar{P}}_{\mathsf{N}}^{\mathcal{R}}.
\label{Id-proj-S}
\end{equation}%
\textsf{II)} The following identities hold:%
\begin{equation}
\mathbb{\bar{P}_{\mathsf{N}}^{\mathcal{R}}}S_{q}(\tau -\eta )\mathsf{C}%
(\lambda |\tau -\eta )=\text{\textsc{S}}_{q}^{\mathcal{R}}\mathbb{\bar{P}_{%
\mathsf{N}}^{\mathcal{R}}}\mathcal{C}(\lambda |\tau ),\text{ \ \ }\mathbb{%
\bar{P}_{\mathsf{N}}^{\mathcal{R}}}S_{q}(\tau +\eta )\mathsf{B}(\lambda
|\tau +\eta )=\text{\textsc{S}}_{q}^{\mathcal{R}}\mathbb{\bar{P}_{\mathsf{N}%
}^{\mathcal{R}}}\mathcal{B}(\lambda |\tau ),  \label{FarXYZRight-projection}
\end{equation}%
on any vector of $\mathbb{\bar{D}}_{\mathsf{(6VD)},\mathsf{N}}^{\mathcal{R}}$
and so it also follows:%
\begin{equation}
\mathsf{T}^{\mathsf{(8V)}}(\lambda )\text{\textsc{S}}_{q}^{\mathcal{R}}%
\mathbb{\bar{P}_{\mathsf{N}}^{\mathcal{R}}}=\text{\textsc{S}}_{q}^{\mathcal{R%
}}\mathbb{\bar{P}_{\mathsf{N}}^{\mathcal{R}}}\overline{\mathcal{T}}^{\mathsf{%
(6VD)}}(\lambda |\tau ).  \label{FarXYZRis-R}
\end{equation}
\end{proposition}

\begin{proof}
The identity $\left( \ref{Id-proj-S}\right) $ follows by computing the
action on the generic elements of the dynamical-spin basis of $\mathbb{\bar{D%
}}_{\mathsf{(6VD)},\mathsf{N}}^{\mathcal{R}}$ and using the orthonormality
of these states.

Let us prove now $\left( \ref{FarXYZRight-projection}\right) $, we use first
the identities:%
\begin{eqnarray}
\mathcal{C}(\lambda |\tau ) &=&\mathsf{C}(\lambda |\tau )\mathsf{T}_{\tau
}^{+}=\mathsf{T}_{\tau }^{+}\mathsf{C}(\lambda |\tau -\eta ), \\
\mathcal{B}(\lambda |\tau ) &=&\mathsf{B}(\lambda |\tau )\mathsf{T}_{\tau
}^{-}=\mathsf{T}_{\tau }^{-}\mathsf{B}(\lambda |\tau +\eta ),
\end{eqnarray}%
to write:%
\begin{equation}
S_{q}(\tau -\eta )\mathsf{C}(\lambda |\tau -\eta )=\mathsf{T}_{\tau
}^{-}S_{q}(\tau )\mathcal{C}(\lambda |\tau ),\text{ \ \ }S_{q}(\tau +\eta )%
\mathsf{B}(\lambda |\tau +\eta )=\mathsf{T}_{\tau }^{+}\mathcal{B}(\lambda
|\tau )S_{q}(\tau ).
\end{equation}%
Then the following identities hold:%
\begin{align}
\mathbb{\bar{P}_{\mathsf{N}}^{\mathcal{R}}}\mathsf{T}_{\tau }^{-}S_{q}(\tau )%
\mathcal{C}(\lambda |\tau )|h_{1},...,h_{\mathsf{N}}\rangle & =\sum_{\mathsf{%
s}=-\mathsf{N}}^{\mathsf{N}}\langle t(\mathsf{s})-\eta |S_{q}(t_{\text{%
\textbf{h}}}-\eta )\mathcal{C}(\lambda |t_{\text{\textbf{h}}})|h_{1},...,h_{%
\mathsf{N}}\rangle  \label{FarXYZprojectC-1} \\
& =\sum_{\mathsf{s}=-\mathsf{N}}^{\mathsf{N}}\langle t(\mathsf{s})|S_{q}(t_{%
\text{\textbf{h}}}-\eta )\mathcal{C}(\lambda |t_{\text{\textbf{h}}%
})|h_{1},...,h_{\mathsf{N}}\rangle  \label{FarXYZprojectC-2} \\
& =\mathbb{\bar{P}_{\mathsf{N}}^{\mathcal{R}}}S_{q}(\tau )\mathcal{C}%
(\lambda |\tau )|h_{1},...,h_{\mathsf{N}}\rangle \\
& =\langle h_{1},...,h_{\mathsf{N}}|\mathcal{C}(\lambda |\tau )\mathbb{\bar{P%
}_{\mathsf{N}}^{\mathcal{L}}}\left[ \text{\textsc{S}}_{q}^{\mathcal{L}}%
\right] ^{-1},
\end{align}%
where $\left( \ref{FarXYZprojectC-2}\right) $ is equal to $\left( \ref%
{FarXYZprojectC-1}\right) $ being $\mathcal{C}(\lambda |\tau
)|h_{1}=0,...,h_{\mathsf{N}}=0\rangle $ zero while the last equality follows
directly by $\left( \ref{Id-proj-S}\right) $ being $\mathcal{C}(\lambda
|\tau )|h_{1},...,h_{\mathsf{N}}\rangle \in \mathbb{\bar{D}}_{\mathsf{(6VD)},%
\mathsf{N}}^{\mathcal{R}}$; similarly one can prove the second identity in $%
\left( \ref{FarXYZRight-projection}\right) $.

The identity $\left( \ref{FarXYZRis-R}\right) $ then is proven by the
following set of equalities which hold on $\mathbb{\bar{D}}_{\mathsf{(6VD)},%
\mathsf{N}}^{\mathcal{R}}$:%
\begin{eqnarray}
&&\mathsf{T}^{\mathsf{(8V)}}(\lambda )\text{\textsc{S}}_{q}^{\mathcal{R}}%
\mathbb{\bar{P}_{\mathsf{N}}^{\mathcal{R}}}\underset{\left( \ref{Id-proj-S}%
\right) }{=}\mathbb{\bar{P}_{\mathsf{N}}^{\mathcal{R}}}\mathsf{T}^{\mathsf{%
(8V)}}(\lambda )S_{q}(\tau ) \\
&&\underset{\left( \ref{Id-proj-S}\right) }{=}\mathbb{\bar{P}_{\mathsf{N}}^{%
\mathcal{R}}}S_{q}(\tau -\eta )\mathsf{C}(\lambda |\tau -\eta )+\mathbb{\bar{%
P}_{\mathsf{N}}^{\mathcal{R}}}S_{q}(\tau +\eta )\mathsf{B}(\lambda |\tau
+\eta ) \\
&&\underset{\left( \ref{FarXYZRight-projection}\right) }{=}\text{\textsc{S}}%
_{q}^{\mathcal{R}}\mathbb{\bar{P}_{\mathsf{N}}^{\mathcal{R}}}\overline{%
\mathcal{T}}^{\mathsf{(6VD)}}(\lambda |\tau ).
\end{eqnarray}
\end{proof}

\subsubsection{On the periodic 8-vertex transfer matrix spectrum by SOV in
odd chains}

It is central to remark that the operator \textsc{S}$_{q}^{\mathcal{R}}\in $
End($\mathbb{S}_{\mathsf{N}}^{\mathcal{R}}$) is not invertible in $\mathbb{S}%
_{\mathsf{N}}^{\mathcal{R}}$. This statement is simply verified observing
that the subspace of $\mathbb{S}_{\mathsf{N}}^{\mathcal{R}}$ generated by
the vectors:%
\begin{equation}
\otimes _{n=1}^{\mathsf{N}-1}|n,h_{n}\rangle \otimes (|\mathsf{N},h_{\mathsf{%
N}}=1\rangle -|\mathsf{N},h_{\mathsf{N}}=0\rangle )\text{ \ for which }%
\sum_{a=1}^{\mathsf{N}-1}h_{a}=\frac{\mathsf{N}-1}{2},
\end{equation}%
belongs to the kernel of \textsc{S}$_{q}^{\mathcal{R}}$ as all these states
are clearly annihilated by the action of the operator $S_{\mathsf{N}}(\xi _{%
\mathsf{N}}|\frac{\eta }{2}\sum_{a=1}^{\mathsf{N}-1}\sigma _{a}^{z}-\frac{%
\eta }{2}\sigma _{\mathsf{N}}^{z})$. Then, it is clear that we cannot use
the identity $\left( \ref{FarXYZRis-R}\right) $ to completely reconstruct
the spectrum of the periodic 8-vertex transfer matrix $\mathsf{T}^{\mathsf{%
(8V)}}(\lambda )$ by using the SOV characterization of the spectrum of the
antiperiodic dynamical 6-vertex transfer matrix $\overline{\mathcal{T}}^{%
\mathsf{(6VD)}}(\lambda |\tau )$. Nevertheless, we can use $\left( \ref%
{FarXYZRis-R}\right) $ to define a first criterion to select eigenvalues of $%
\overline{\mathcal{T}}^{\mathsf{(6VD)}}(\lambda |\tau )$ which are also
eigenvalues of $\mathsf{T}^{\mathsf{(8V)}}(\lambda )$ and to associate to
anyone of these eigenvalues just one eigenvector of $\mathsf{T}^{\mathsf{(8V)%
}}(\lambda )$. In particular, we can prove the following lemma:

\begin{lemma}
\label{Lemma-Ch-8v-SOV}Let us consider a $\mathsf{t}_{\mathsf{6VD}}(\lambda
)\in \Sigma _{\overline{\mathcal{T}}^{\mathsf{(6VD)}}}$ and let $|\mathsf{t}%
_{\mathsf{6VD}}\rangle $\ be the corresponding eigenvector of $\overline{%
\mathcal{T}}^{\mathsf{(6VD)}}(\lambda |\tau )$, then if the vector%
\begin{equation}
\mathbb{\bar{P}_{\mathsf{N}}^{\mathcal{R}}}|\mathsf{t}_{\mathsf{6VD}}\rangle
=\sum_{h_{1},...,h_{\mathsf{N}}=0}^{1}\prod_{a=1}^{\mathsf{N}}Q_{\mathsf{t}%
}(\xi _{a}^{(h_{a})})\det_{\mathsf{N}}\Theta _{ij}^{\left( \text{\textbf{h}}%
\right) }\underline{|h_{1},...,h_{\mathsf{N}}\rangle },\text{ \ with }%
\underline{|h_{1},...,h_{\mathsf{N}}\rangle }=\langle t_{\text{\textbf{h}}%
}|h_{1},...,h_{\mathsf{N}}\rangle
\end{equation}%
does not belong to the kernel of \textsc{S}$_{q}^{\mathcal{R}}$, $\mathsf{t}%
_{\mathsf{6VD}}(\lambda )$ is eigenvalue of $\mathsf{T}^{\mathsf{(8V)}%
}(\lambda )$ and \textsc{S}$_{q}^{\mathcal{R}}\mathbb{\bar{P}_{\mathsf{N}}^{%
\mathcal{R}}}|\mathsf{t}_{\mathsf{6VD}}\rangle $ is one corresponding
eigenvector.
\end{lemma}

\begin{proof}
Under the condition:%
\begin{equation}
\text{\textsc{S}}_{q}^{\mathcal{R}}\mathbb{\bar{P}_{\mathsf{N}}^{\mathcal{R}}%
}|\mathsf{t}_{\mathsf{6VD}}\rangle \neq \text{\b{0}}\in \mathbb{S}_{\mathsf{N%
}}^{\mathcal{R}},
\end{equation}%
the proposition is a simple consequence of the identity $\left( \ref%
{FarXYZRis-R}\right) $; indeed, it hold:%
\begin{equation}
\mathsf{T}^{\mathsf{(8V)}}(\lambda )\text{\textsc{S}}_{q}^{\mathcal{R}}%
\mathbb{\bar{P}_{\mathsf{N}}^{\mathcal{R}}}|\mathsf{t}_{6\text{$VD$}}\rangle
=\text{\textsc{S}}_{q}^{\mathcal{R}}\mathbb{\bar{P}_{\mathsf{N}}^{\mathcal{R}%
}}\overline{\mathcal{T}}^{\mathsf{(6VD)}}(\lambda |\tau )|\mathsf{t}_{6\text{%
$VD$}}\rangle =\text{\textsc{S}}_{q}^{\mathcal{R}}\mathbb{\bar{P}_{\mathsf{N}}^{\mathcal{R}}%
}|\mathsf{t}_{\mathsf{6VD}}\rangle \mathsf{t}_{\mathsf{6VD}%
}(\lambda ).
\end{equation}
\end{proof}

\textbf{Remark 2.} It is worth remarking that we need this criterion as, from
the results in subsection \ref{1ch-8v-eigenvalues}, we only know that
the set $\Sigma _{\mathsf{T}^{\mathsf{(8V)}}}$ of the eigenvalues of $%
\mathsf{T}^{\mathsf{(8V)}}(\lambda )$ is contained in the set $\Sigma _{%
\overline{\mathcal{T}}^{\mathsf{(6VD)}}}$ of the eigenvalues of $\overline{%
\mathcal{T}}^{\mathsf{(6VD)}}(\lambda |\tau )$ in the case of a chain with
an odd number of quantum sites. Moreover, it is important to clarify that
currently we have only proven that the above lemma defines a criterion, i.e.
a sufficient condition for an element of $\Sigma _{\overline{\mathcal{T}}^{%
\mathsf{(6VD)}}}$ to be also an element of $\Sigma _{\mathsf{T}^{\mathsf{(8V)%
}}}$. It will be fundamental to understand if this is also a necessary
condition as in this last case we will get a complete characterization of $%
\Sigma _{\mathsf{T}^{\mathsf{(8V)}}}$ and one eigenstate of $\mathsf{T}^{%
\mathsf{(8V)}}(\lambda )$ for any element of $\Sigma _{\mathsf{T}^{\mathsf{%
(8V)}}}$ just using the SOV characterization of the spectrum of $\overline{%
\mathcal{T}}^{\mathsf{(6VD)}}(\lambda |\tau )$. Finally, let us point out
that the fact that \textsc{S}$_{q}^{\mathcal{R}}\in $ End($\mathbb{S}_{%
\mathsf{N}}^{\mathcal{R}}$) is not invertible in $\mathbb{S}_{\mathsf{N}}^{%
\mathcal{R}}$ is just required to make the identity $\left( \ref{FarXYZRis-R}%
\right) $ compatible with the observed degeneracy of the spectrum of $%
\mathsf{T}^{\mathsf{(8V)}}(\lambda )$  for the cases $\mathsf{N}=1$ and $3$
explicitly analyzed in appendix. As in \ presence of this degeneracy the
proven simplicity of the spectrum of $\overline{\mathcal{T}}^{\mathsf{(6VD)}%
}(\lambda |\tau )$ implies that $\Sigma _{\mathsf{T}^{\mathsf{(8V)}}}$ must
be properly contained in $\Sigma _{\overline{\mathcal{T}}^{\mathsf{(6VD)}}}$.

\section{Conclusion}

In this paper we have focused our attention on the highest weight
representations of the dynamical 6-vertex Yang-Baxter algebra on a generic
spin-1/2 quantum chain with an odd number $\mathsf{N}$ of sites. We have
studied the integrable quantum model associated to the antiperiodic boundary
conditions in the framework of the SOV method. For this integrable quantum
models, we have derived:
\begin{itemize}
\item The complete SOV description of transfer matrix eigenvalues and
eigenstates and the simplicity of the spectrum.
\item Matrix elements of the identity on separate states expressed by one
determinant formulae of $\mathsf{N}\times\mathsf{N}$ matrices with elements
given by sums over the eigenvalues of the quantum separate variables of the
product of the coefficients of the left/right separate states, which holds
in particular for the eigenstates of antiperiodic dynamical 6-vertex
transfer matrix.
\end{itemize}
The results derived in this paper provide the required setup to compute
matrix elements on transfer matrix eigenstates of local operators. The
analysis of the following steps:
\begin{itemize}
\item local operator reconstructions in terms of Sklyanin's quantum separate
variables,
\item form factors of the local operators on the transfer matrix eigenstates
in determinant form,
\end{itemize}
will be presented in a paper which is currently under completion in
collaboration with Levy-Bencheton and Terras \cite{FarXYZ?NT12}. The study
of correlation functions done in \cite{FarXYZ?Terras12} for the periodic
dynamical 6-vertex chain even if developed in the framework of the algebraic
Bethe ansatz is relevant also for the current analysis. Indeed, the
reconstruction of local operators of \cite{FarXYZ?Terras12} can be adapted
to the antiperiodic dynamical 6-vertex chain to get reconstruction of local
operators in terms of the quantum separate variables. Then the computation
of the form factors for the antiperiodic dynamical 6-vertex proceed in a
similar way to that of the standard 6-vertex quantum chain with antiperiodic
boundary conditions as derived in \cite{FarXYZN12-0} in the SOV framework.
It is also worth mentioning that the knowledge of the form factors of local
operators represents also an efficient tools for controlled numerical
analysis of correlation functions. Indeed, the decomposition of the identity 
$\left( \ref{FarXYZId-decomp-T-eigenstates}\right) $ allows to rewrite the
correlation functions in terms of form factors and then it is a priori
possible to apply the same kind of approach developed in \cite{FarXYZCM05}
in the ABA framework\footnote{%
See the series of papers \cite{FarXYZCM05}-\cite{FarXYZCCS07} where the
dynamical structure factors, observable by neutron scattering experiments 
\cite{FarXYZBloch36}-\cite{FarXYZBalescu75}, were numerically evaluated.}
also in our SOV framework to get numerical evaluations of correlation
functions.

We have moreover shown that the existence of gauge transformations allows
us to use the antiperiodic dynamical 6-vertex transfer matrix as a tool to
further analyze the spectral problem of the periodic 8-vertex transfer
matrix in the case of a chain with $\mathsf{N}$ odd and for general values
of the coupling constant $\eta $ (non restricted to the \textit{elliptic
roots of unit}). The potential relevance of this analysis is made clear
observing that the standard Bethe ansatz analysis developed in \cite%
{FarXYZBa72-1, FarXYZBa72-2, FarXYZBa72-3,FarXYZFT79} does not apply to this
case. More in detail, we have shown that the gauge transformations allow to
define a criterion to select eigenvalues of the antiperiodic
dynamical 6-vertex transfer matrix which are also eigenvalues of the
periodic 8-vertex transfer matrix, moreover associating to any one of them
one nonzero periodic 8-vertex eigenstate. Finally, let us stress the importance to understand if this criterion also define a necessary condition for a antiperiodic dynamical 6-vertex eigenvalue to be also a periodic 8-vertex eigenvalue. Indeed, in this case the SOV characterization of the antiperiodic
dynamical 6-vertex spectrum will also allow the complete characterization of
the periodic 8-vertex eigenvalues and the construction of one of its
eigenstates for anyone of its eigenvalues. The answer to this fundamental
question requires a systematic and simultaneous analysis of the degeneracy
of the periodic 8-vertex spectrum and of the dimension of the kernel of the
operator \textsc{S}$_{q}^{\mathcal{R}}$ which we will try to address
elsewhere. \bigskip 

\textbf{Acknowledgments}\thinspace\ The author gratefully acknowledge B.
McCoy for the many stimulating discussions on the 8-vertex model and the
interesting questions on quantum separation of variables which have strongly
inspired and motivated the author to develop the present paper. The author
would also like to thank N. Kitanine, K. K. Kozlowski and J. M. Maillet for
their interest and D. Levy-Bencheton and V. Terras for their interest,
attentive reading and remarks on a first draft of this paper. The author is
supported by National Science Foundation grants PHY-0969739 and gratefully
acknowledges the YITP Institute of Stony Brook for the opportunity to
develop his research programs. The author would also like to thank for their
hospitality the Theoretical Physics Group of the Laboratory of Physics at
ENS-Lyon and the Mathematical Physics Group at IMB of the Dijon University
(under support ANR-10-BLAN-0120-04-DIADEMS).
\appendix
\section{Appendix}

Here we analyze explicitly the spectral problem of the transfer matrices $%
\mathsf{T}^{\mathsf{(8V)}}(\lambda )$ and $\overline{\mathcal{T}}^{\mathsf{%
(6VD)}}(\lambda |\tau )$ for the trivial cases of chains with $\mathsf{N}=1$
and $\mathsf{N}=3$. The aim is to make clear some of the above
statements and to have some basic analysis about the degeneracy of the $%
\mathsf{T}^{\mathsf{(8V)}}(\lambda )$ spectrum.

\subsection{Spectrum of $\mathsf{T}^{\mathsf{(8V)}}(\protect\lambda )$ and $%
\overline{\mathcal{T}}^{\mathsf{(6VD)}}(\protect\lambda |\protect\tau )$ for 
$\mathsf{N}=1$}

In the one site case it holds:%
\begin{equation}
\mathsf{T}^{\mathsf{(8V)}}(\lambda )=(\text{a}(\lambda )+\text{b}(\lambda
))\left( 
\begin{array}{cc}
1 & 0 \\ 
0 & 1%
\end{array}%
\right) \text{ \ on }\mathbb{S}_{\mathsf{N}=1}^{\mathcal{R}},
\end{equation}%
and from:%
\begin{equation}
\overline{\mathcal{T}}^{\mathsf{(6VD)}}(\lambda |\tau )=\mathsf{T}_{\tau
}^{+}\mathsf{C}(\lambda |\tau -\eta )+\mathsf{T}_{\tau }^{-}\mathsf{B}%
(\lambda |\tau +\eta )
\end{equation}%
it holds%
\begin{equation}
\overline{\mathcal{T}}^{\mathsf{(6VD)}}(\lambda |\tau )=c(\lambda |\eta
/2)\left( 
\begin{array}{cc}
0 & \mathsf{T}_{\tau }^{+} \\ 
\mathsf{T}_{\tau }^{-} & 0%
\end{array}%
\right) \text{ \ on }\mathbb{D}_{\mathsf{(6VD)},\mathsf{N}=1}^{\mathcal{R}}.
\end{equation}%
Then, we have that the spectrum of $\overline{\mathcal{T}}^{\mathsf{(6VD)}%
}(\lambda |\tau )$ is simple and characterized by:%
\begin{equation}
\overline{\mathcal{T}}^{\mathsf{(6VD)}}(\lambda |\tau )|\mathsf{t}_{\mathsf{%
6VD}}^{\left( \pm \right) }\rangle =|\mathsf{t}_{\mathsf{6VD}}^{\left( \pm
\right) }\rangle \mathsf{t}_{\mathsf{6VD}}^{\left( \pm \right) }(\lambda )
\end{equation}%
where we have defined: 
\begin{equation}
\mathsf{t}_{\mathsf{6VD}}^{\left( \pm \right) }(\lambda )\equiv \pm
c(\lambda |\eta /2),\text{ }|\mathsf{t}_{\mathsf{6VD}}^{\left( \pm \right)
}\rangle \equiv \left( 
\begin{array}{l}
\pm |t(1)\rangle  \\ 
|t(-1)\rangle 
\end{array}%
\right) \in \mathbb{D}_{\mathsf{(6VD)},\mathsf{N}=1}^{\mathcal{R}}\text{,}
\end{equation}%
and\ $|t(a)\rangle $ is the $\tau $ eigenstate with eigenvalue $t(a)\equiv
-\eta a/2$. Instead, the spectrum of $\mathsf{T}^{\mathsf{(8V)}}(\lambda )$ is double
degenerate with eigenvalue $\mathsf{t}_{\mathsf{8V}}(\lambda )=$a$(\lambda )+
$b$(\lambda )$, then the identity:%
\begin{equation}
\text{a}(\lambda )+\text{b}(\lambda )=c(\lambda |\eta /2)
\end{equation}%
implies the proper set inclusion $\Sigma _{\mathsf{T}^{\mathsf{(8V)}%
}}\subset \Sigma _{\overline{\mathcal{T}}^{\mathsf{(6VD)}}}$ for the $%
\mathsf{N}=1$ case. Let us now observe that:%
\begin{equation}
\text{\textsc{S}}_{q}^{\mathcal{R}}=S_{1}(\lambda |-\eta /2\sigma
_{1}^{z})\equiv \left( 
\begin{array}{cc}
\theta _{2}(\lambda +\eta /2|2w) & \theta _{2}(\lambda +\eta /2|2w) \\ 
\theta _{3}(\lambda +\eta /2|2w) & \theta _{3}(\lambda +\eta /2|2w)%
\end{array}%
\right) _{0}
\end{equation}%
and:
\begin{equation}
\text{\textsc{S}}_{q}^{\mathcal{R}}\mathbb{\bar{P}}_{\mathsf{N}=1}^{\mathcal{%
R}}|\mathsf{t}_{\mathsf{6VD}}^{(\pm )}\rangle =\left( 
\begin{array}{c}
\theta _{2}(\lambda +\eta /2|2w)(1\pm 1) \\ 
\theta _{3}(\lambda +\eta /2|2w)(1\pm 1)%
\end{array}%
\right) \in \mathbb{S}_{\mathsf{N}}^{\mathcal{R}}.
\end{equation}%
Then, in agreement with the Lemma \ref{Lemma-Ch-8v-SOV}, we have that $%
\mathsf{t}_{\mathsf{6VD}}^{\left( +\right) }(\lambda )$ is $\mathsf{T}^{%
\mathsf{(8V)}}(\lambda )$ eigenvalue and:%
\begin{equation}
\text{\textsc{S}}_{q}^{\mathcal{R}}\mathbb{\bar{P}}_{\mathsf{N}=1}^{\mathcal{%
R}}|\mathsf{t}_{\mathsf{6VD}}^{(+)}\rangle =2\left( 
\begin{array}{c}
\theta _{2}(\lambda +\eta /2|2w) \\ 
\theta _{3}(\lambda +\eta /2|2w)%
\end{array}%
\right) \in \mathbb{S}_{\mathsf{N}}^{\mathcal{R}},
\end{equation}%
is one corresponding eigenstate. It is also interesting to remark that the
eigenvalue $\mathsf{t}_{\mathsf{6VD}}^{\left( -\right) }(\lambda )$ of $%
\overline{\mathcal{T}}^{\mathsf{(6VD)}}(\lambda |\tau )$ for which it holds 
\textsc{S}$_{q}^{\mathcal{R}}\mathbb{\bar{P}}_{\mathsf{N}=1}^{\mathcal{R}}|%
\mathsf{t}_{\mathsf{6VD}}^{(-)}\rangle =$\b{0} is not an eigenvalue of $%
\mathsf{T}^{\mathsf{(8V)}}(\lambda )$.

\subsection{Spectrum of $\mathsf{T}^{\mathsf{(8V)}}(\protect\lambda )$ and $%
\overline{\mathcal{T}}^{\mathsf{(6VD)}}(\protect\lambda |\protect\tau )$ for 
$\mathsf{N}=3$}

\subsubsection{General statements on the spectrum for $\mathsf{N}=3$}

The system of equations which completely characterize the spectrum
(eigenvalues and eigenvectors) of $\overline{\mathcal{T}}^{\mathsf{(6VD)}%
}(\lambda |\tau )$ reads:%
\begin{equation}
x_{n}\left( \sum_{a=1}^{\mathsf{N}}J_{n,a}x_{a}\right) -q_{n}=0\text{ \ \ \
\ }\forall n\in \{1,...,\mathsf{N}\}.  \label{SOV-Ch-Spectrum_6VD}
\end{equation}%
It is an inhomogeneous system of $\mathsf{N}$ quadratic equations in the $%
\mathsf{N}$ unknown $x_{n}=\mathsf{t}_{\mathsf{6VD}}(\xi _{n})$ with
coefficients characterized by:%
\begin{equation}
J_{n,a}=\frac{\theta (t_{\text{\textbf{0}}}-\xi _{n}+\xi _{a}+\eta )}{\theta
(t_{\text{\textbf{0}}})}\prod_{b\neq a}\frac{\theta (\xi _{n}-\xi _{b}-\eta )%
}{\theta (\xi _{a}-\xi _{b})},\text{ \ \ }q_{n}=\text{\textsc{a}}(\xi
_{n}^{(0)})\text{\textsc{d}}(\xi _{n}^{(1)}),\text{ \ \ }\forall n\in
\{1,...,\mathsf{N}\},  \label{Coeff-SOV-Ch-Spectrum_6VD}
\end{equation}%
as it is simple to derive substituting in $\left( \ref{FarXYZset-t}\right) $
the interpolation formula $\left( \ref{FarXYZI-Functional-eq}\right) $ for $%
\mathsf{t}_{\mathsf{6VD}}(\xi _{n}^{(1)})$. It is trivial to observe that
the set of the solutions to this system has a $Z_{2}$ symmetry; i.e. if z$%
_{_{\mathsf{N}}}^{(+)}\equiv \{x_{1},...,x_{n},...,x_{\mathsf{N}}\}$ is a
solution of it then also z$_{_{\mathsf{N}}}^{\left( -\right) }\equiv
\{-x_{1},...,-x_{n},...,-x_{\mathsf{N}}\}$ is a solution. So from the
interpolation formula $\left( \ref{FarXYZI-Functional-eq}\right) $, it
follows that if $\mathsf{t}_{\mathsf{6VD}}^{(+)}(\lambda )\in \Sigma _{%
\overline{\mathcal{T}}^{\mathsf{(6VD)}}}$ then also $\mathsf{t}_{\mathsf{6VD}%
}^{\left( -\right) }(\lambda )=\left( -\mathsf{t}_{\mathsf{6VD}%
}^{(+)}(\lambda )\right) \in \Sigma _{\overline{\mathcal{T}}^{\mathsf{(6VD)}%
}}$.

In the case $\mathsf{N}=3$, we have verified that the system (\ref{SOV-Ch-Spectrum_6VD}) has $2^{3}$ distinct solutions. Then $\Sigma
_{\overline{\mathcal{T}}^{\mathsf{(6VD)}}}$ is composed by $2^{3}$ distinct
elliptic polynomials $\mathsf{t}_{\mathsf{6VD}}^{\left( \pm ,a\right)
}(\lambda )$ with $a\in \{1,...,4\}$ and so $\overline{\mathcal{T}}^{\mathsf{%
(6VD)}}(\lambda |\tau )$ has simple spectrum as proven in this paper.
Moreover, in the case of $\mathsf{N}=3$, we have studied the spectrum of $%
\mathsf{T}^{\mathsf{(8V)}}(\lambda )$ and we have observed that it has 4
distinct eigenvalues each one being double degenerate. We have observed that
the set $\Sigma _{\mathsf{T}^{\mathsf{(8V)}}}\equiv \{\mathsf{t}_{\mathsf{8V}%
}^{\left( 1\right) }(\lambda ),\mathsf{t}_{\mathsf{8V}}^{\left( 2\right)
}(\lambda ),\mathsf{t}_{\mathsf{8V}}^{\left( 3\right) }(\lambda ),\mathsf{t}%
_{\mathsf{8V}}^{\left( 4\right) }(\lambda )\}$ coincide with the set of
elliptic polynomials which are generated by using 4 distinct solutions $z_{_{%
\mathsf{N}}}^{(a)}$ of the system (\ref{SOV-Ch-Spectrum_6VD})
and\ the interpolation formula $\left( \ref{FarXYZI-Functional-eq}\right) $.
Moreover, the solutions $z_{_{\mathsf{N}}}^{(a)}$ with $a\in \{1,...,4\}$
used to construct $\Sigma _{\mathsf{T}^{\mathsf{(8V)}}}$ appear not to be
related by the $Z_{2}$ symmetry of the system (\ref{SOV-Ch-Spectrum_6VD}); in fact, it holds:%
\begin{equation}
z_{_{\mathsf{N}}}^{(a)}\neq -z_{_{\mathsf{N}}}^{(b)}\text{ \ for any }a,b\in
\{1,...,4\}.
\end{equation}%
Let us fix the notation:%
\begin{equation}
\mathsf{t}_{\mathsf{6VD}}^{\left( \pm ,a\right) }(\lambda )\equiv \pm 
\mathsf{t}_{\mathsf{8V}}^{\left( a\right) }(\lambda )\text{ for any }a\in
\{1,...,4\},
\end{equation}%
then we can summarize the above observations on the spectrum of $\overline{%
\mathcal{T}}^{\mathsf{(6VD)}}(\lambda |\tau )$ and $\mathsf{T}^{\mathsf{(8V)}%
}(\lambda )$ by saying that to the two nondegenerate eigenvalues $\mathsf{t}%
_{\mathsf{6VD}}^{\left( +,a\right) }(\lambda )$ and $\mathsf{t}_{\mathsf{6VD}%
}^{\left( -,a\right) }(\lambda )$ of $\overline{\mathcal{T}}^{\mathsf{(6VD)}%
}(\lambda |\tau )$ it is associates just the double degenerate eigenvalue $%
\mathsf{t}_{\mathsf{8V}}^{\left( a\right) }(\lambda )$ of $\mathsf{T}^{%
\mathsf{(8V)}}(\lambda )$. Our analysis shows that this statement holds for $%
\mathsf{N}=1$ and $\mathsf{N}=3$ to be able to verify if it persists for a
generic odd $\mathsf{N}$ can be one central step toward the characterization
of the spectrum of $\mathsf{T}^{\mathsf{(8V)}}(\lambda )$ in terms of the
SOV characterization of the $\overline{\mathcal{T}}^{\mathsf{(6VD)}}(\lambda
|\tau )$ spectrum.

\subsubsection{Some numerical data for $\mathsf{N}=3$}

For convenience of the reader, we write explicitly the periodic 8-vertex
transfer matrix $\mathsf{T}^{\mathsf{(8V)}}(\lambda )$ in the $\mathsf{N}=3$
case: 
\begin{equation*}
\text{ \ }{\tiny \left( 
\begin{array}{cccccccc}
\text{a}_{1}\text{a}_{2}\text{a}_{3}+\text{b}_{1}\text{b}_{2}\text{b}_{3} & 0
& 0 & \text{b}_{1}\text{d}_{2}\text{c}_{3}+\text{a}_{1}\text{c}_{2}\text{d}%
_{3} & 0 & \text{d}_{1}\text{a}_{2}\text{c}_{3}+\text{c}_{1}\text{b}_{2}%
\text{d}_{3} & \text{d}_{1}\text{c}_{2}\text{b}_{3}+\text{c}_{1}\text{d}_{2}%
\text{a}_{3} & 0 \\ 
0 & \text{b}_{1}\text{b}_{2}\text{a}_{3}+\text{a}_{1}\text{a}_{2}\text{b}_{3}
& \text{a}_{1}\text{c}_{2}\text{c}_{3}+\text{b}_{1}\text{d}_{2}\text{d}_{3}
& 0 & \text{c}_{1}\text{b}_{2}\text{c}_{3}+\text{d}_{1}\text{a}_{2}\text{d}%
_{3} & 0 & 0 & \text{d}_{1}\text{c}_{2}\text{a}_{3}+\text{c}_{1}\text{d}_{2}%
\text{b}_{3} \\ 
0 & \text{b}_{1}\text{c}_{2}\text{c}_{3}+\text{a}_{1}\text{d}_{2}\text{d}_{3}
& \text{a}_{1}\text{b}_{2}\text{a}_{3}+\text{b}_{1}\text{a}_{2}\text{b}_{3}
& 0 & \text{c}_{1}\text{c}_{2}\text{a}_{3}+\text{d}_{1}\text{d}_{2}\text{b}%
_{3} & 0 & 0 & \text{d}_{1}\text{b}_{2}\text{c}_{3}+\text{c}_{1}\text{a}_{2}%
\text{d}_{3} \\ 
\text{a}_{1}\text{d}_{2}\text{c}_{3}+\text{b}_{1}\text{c}_{2}\text{d}_{3} & 0
& 0 & \text{b}_{1}\text{a}_{2}\text{a}_{3}+\text{a}_{1}\text{b}_{2}\text{b}%
_{3} & 0 & \text{c}_{1}\text{c}_{2}\text{b}_{3}+\text{d}_{1}\text{d}_{2}%
\text{a}_{3} & \text{c}_{1}\text{a}_{2}\text{c}_{3}+\text{d}_{1}\text{b}_{2}%
\text{d}_{3} & \text{0} \\ 
0 & \text{c}_{1}\text{a}_{2}\text{c}_{3}+\text{d}_{1}\text{b}_{2}\text{d}_{3}
& \text{c}_{1}\text{c}_{2}\text{b}_{3}+\text{d}_{1}\text{d}_{2}\text{a}_{3}
& 0 & \text{b}_{1}\text{a}_{2}\text{a}_{3}+\text{a}_{1}\text{b}_{2}\text{b}%
_{3} & 0 & 0 & \text{a}_{1}\text{d}_{2}\text{c}_{3}+\text{b}_{1}\text{c}_{2}%
\text{d}_{3} \\ 
\text{d}_{1}\text{b}_{2}\text{c}_{3}+\text{c}_{1}\text{a}_{2}\text{d}_{3} & 0
& 0 & \text{c}_{1}\text{c}_{2}\text{a}_{3}+\text{d}_{1}\text{d}_{2}\text{b}%
_{3} & 0 & \text{a}_{1}\text{b}_{2}\text{a}_{3}+\text{b}_{1}\text{a}_{2}%
\text{b}_{3} & \text{b}_{1}\text{c}_{2}\text{c}_{3}+\text{a}_{1}\text{d}_{2}%
\text{d}_{3} & 0 \\ 
\text{d}_{1}\text{c}_{2}\text{a}_{3}+\text{c}_{1}\text{d}_{2}\text{b}_{3} & 0
& 0 & \text{c}_{1}\text{b}_{2}\text{c}_{3}+\text{d}_{1}\text{a}_{2}\text{d}%
_{3} & 0 & \text{a}_{1}\text{c}_{2}\text{c}_{3}+\text{b}_{1}\text{d}_{2}%
\text{d}_{3} & \text{b}_{1}\text{b}_{2}\text{a}_{3}+\text{a}_{1}\text{a}_{2}%
\text{b}_{3} & 0 \\ 
0 & \text{d}_{1}\text{c}_{2}\text{b}_{3}+\text{c}_{1}\text{d}_{2}\text{a}_{3}
& \text{d}_{1}\text{a}_{2}\text{c}_{3}+\text{c}_{1}\text{b}_{2}\text{d}_{3}
& 0 & \text{b}_{1}\text{d}_{2}\text{c}_{3}+\text{a}_{1}\text{c}_{2}\text{d}%
_{3} & 0 & 0 & \text{a}_{1}\text{a}_{2}\text{a}_{3}+\text{b}_{1}\text{b}_{2}%
\text{b}_{3}%
\end{array}%
\bigskip \right) ,}
\end{equation*}%
where we have used the notations:%
\begin{equation*}
\text{a}_{n}=\text{a}(\lambda -\xi _{n}),\text{ b}_{n}=\text{b}(\lambda -\xi
_{n}),\text{ c}_{n}=\text{c}(\lambda -\xi _{n}),\text{ d}_{n}=\text{d}%
(\lambda -\xi _{n}).
\end{equation*}%
To verify the statements done in the previous subsection it is only need to
write in Mathematica the previous $8\times 8$ matrix and the system of
equations (\ref{SOV-Ch-Spectrum_6VD}) for $\mathsf{N}=3$ and
solve the eigenvalue problem and the system for generic values of the 5
parameters $(\xi _{1},\xi _{2},\xi _{3},\eta ,t=e^{i\pi w})$. Here we report
just few numerical data as a confirmation that we have really implemented
this numerical exercise. Defined w$_{3}^{(a)}=\{\mathsf{t}_{\mathsf{8V}%
}^{\left( a\right) }(\xi _{1}),\mathsf{t}_{\mathsf{8V}}^{\left( a\right)
}(\xi _{2}),\mathsf{t}_{\mathsf{8V}}^{\left( a\right) }(\xi _{3})\}$, it
holds:

1) For $\xi _{1}=5.7,\xi _{2}=1.5,\xi _{3}=0.22,\eta =0.7,t=0.26$:
\begin{equation*}
\begin{array}{l}
\text{z}_{3}^{(\pm ,1)}=\pm
\{2.4648971133384494,0.5263660613291964,-0.0461646762536026\} \\ 
\text{z}_{3}^{(\pm ,2)}=\pm
\{0.16746377944367666,0.09438584696000717,-3.7893847598813264\} \\ 
\text{z}_{3}^{(\pm ,3)}=\pm
\{0.15697838428546823,0.5124574129431847,-0.7445585159876167\} \\ 
\text{z}_{3}^{(\pm ,4)}=\pm
\{0.02568158650662899,3.433163601035112,-0.679328947667353\} \\ 
\text{w}_{3}^{(1)}=\{2.46489711333845,0.5263660613291976,-0.0461646762536022%
\} \\ 
\text{w}_{3}^{(2)}=\{0.167463779423851,0.0943858469664461,-3.789384759881333%
\} \\ 
\text{w}_{3}^{(3)}=%
\{0.15697838428547273,0.5124574129431814,-0.7445585159876165\} \\ 
\text{w}_{3}^{(4)}=%
\{0.025681586506630664,3.4331636010351154,-0.6793289476673527\}.%
\end{array}%
\end{equation*}

2) For $\xi _{1}=2.5,\xi _{2}=3.1,\xi _{3}=1.33,\eta =0.3,t=0.45$:
\begin{equation*}
\begin{array}{l}
\text{z}_{3}^{(\pm ,1)}=\pm
\{-2.3672052885387806,-0.03421683553328285,0.560404707906603\} \\ 
\text{z}_{3}^{(\pm ,2)}=\pm
\{0.1607220217069632,7.959749585813279,0.03548156343430941\} \\ 
\text{z}_{3}^{(\pm ,3)}=\pm
\{0.14344459641406113,0.5655603642746968,0.5595184106850913\} \\ 
\text{z}_{3}^{(\pm ,4)}=\pm
\{0.009963704747040916,0.5039536632240319,9.03990912589408\} \\ 
\text{w}_{3}^{(1)}=%
\{-2.367205288523499,-0.034216835529656396,0.5604047079065965\} \\ 
\text{w}_{3}^{(2)}=\{0.1607220217069637,7.95974958581329,0.03548156343431045%
\} \\ 
\text{w}_{3}^{(3)}=%
\{0.14344459639912585,0.5655603642711194,0.5595184106850958\} \\ 
\text{w}_{3}^{(4)}=%
\{0.009963750993033916,0.5039536669291063,0.5595184106850958\}.%
\end{array}%
\end{equation*}

3) For $\xi _{1}=1.7,\xi _{2}=3.5,\xi _{3}=5.22,\eta =4.7,t=0.05$:
\begin{equation*}
\begin{array}{l}
\text{z}_{3}^{(\pm ,1)}=\pm
\{0.9071447507669119,0.0010355130798548361,-0.6163903868766624\} \\ 
\text{z}_{3}^{(\pm ,2)}=\pm
\{-0.18602724783757033,-0.02888852650572982,-0.10774226124070294\} \\ 
\text{z}_{3}^{(\pm ,3)}=\pm
\{0.13725423857934435,-0.024752594653532196,0.1704282336621456\} \\ 
\text{z}_{3}^{(\pm ,4)}=\pm
\{-0.04740255397294748,0.8919753005921505,0.013694099141681645\} \\ 
\text{w}_{3}^{(1)}=%
\{0.907144750766913,0.001035513079898853,-0.6163903868766655\} \\ 
\text{w}_{3}^{(2)}=%
\{-0.18602724783757013,-0.028888526505732478,-0.10774226124070306\} \\ 
\text{w}_{3}^{(3)}=%
\{0.13725423857934346,-0.02475259465352673,0.17042823366214616\} \\ 
\text{w}_{3}^{(4)}=%
\{-0.04740255397294748,0.8919753005921487,0.013694099141681883\}.%
\end{array}%
\end{equation*}

4) For $\xi _{1}=49.7,\xi _{2}=10.5,\xi _{3}=12.22,\eta =5.87,t=0.726$:
\begin{equation*}
\begin{array}{l}
\text{z}_{3}^{(\pm ,1)}=\pm
\{0.158866785906656,-0.002317414600871322,0.004665001427754174\} \\ 
\text{z}_{3}^{(\pm ,2)}=\pm
\{0.004163560745980359,-0.13352504997041553,0.0030893063326063934\} \\ 
\text{z}_{3}^{(\pm ,3)}=\pm
\{0.0027572370077268236,-7.693461066977195,0.00008096415168424851\} \\ 
\text{z}_{3}^{(\pm ,4)}=\pm
\{-0.001396539108516703,-0.13352504998006434,-0.009210278823835091\} \\ 
\text{w}_{3}^{(1)}=%
\{0.15886678590666517,-0.0023174146009546297,0.0046650014277542385\} \\ 
\text{w}_{3}^{(2)}=%
\{0.004163560745980381,-0.13352504997042003,0.003089306332606317\} \\ 
\text{w}_{3}^{(3)}=%
\{0.002757237007726877,-7.693461066977227,0.00008096415168424613\} \\ 
\text{w}_{3}^{(4)}=%
\{-0.001396539108516455,-0.133525049979987,-0.009210278823835037\}.%
\end{array}%
\end{equation*}

5) For $\xi _{1}=11.2,\xi _{2}=1.1,\xi _{3}=0.82,\eta =3.3,t=0.096$:
\begin{equation*}
\begin{array}{l}
\text{z}_{3}^{(\pm ,1)}=\pm
\{-0.13845098667904934,-0.04279356398629822,0.017867992946492404\} \\ 
\text{z}_{3}^{(\pm ,2)}=\pm
\{0.12350539448737866,0.022662651149136445,0.03782279719611843\} \\ 
\text{z}_{3}^{(\pm ,3)}=\pm
\{0.11482851797211138,-0.02822854036213841,-0.032659693368688764\} \\ 
\text{z}_{3}^{(\pm ,4)}=\pm
\{-0.10167300872962227,0.052191832632450655,-0.019949961538809933\} \\ 
\text{w}_{3}^{(1)}=%
\{-0.13845098667905043,-0.04279356398629837,0.01786799294649241\} \\ 
\text{w}_{3}^{(2)}=%
\{0.1235053944873589,0.022662651149137868,0.03782279719611853\} \\ 
\text{w}_{3}^{(3)}=%
\{0.11482851797211588,-0.02822854036213898,-0.032659693368688944\} \\ 
\text{w}_{3}^{(4)}=%
\{-0.10167300872962239,0.05219183263245088,-0.019949961538809936\}.%
\end{array}%
\end{equation*}%
In all these cases (up to small numerical errors) it is possible to observe
that any eigenvalue of $\mathsf{T}^{\mathsf{(8V)}}(\lambda )$ is double
degenerate and that the following identities hold:%
\begin{equation}
\text{z}_{3}^{(+,a)}=\text{w}_{3}^{(a)}\text{ for any }a\in \{1,...,4\}.
\end{equation}

\begin{small}

\end{small}


\begin{thebibliography}{999}
\bibitem{FarXYZGMN12-SG} N. Grosjean, J. M. Maillet, G. Niccoli, {\it On form factors of local operators in the lattice sine-Gordon model}, J. Stat. Mech. (2012) P10006.

\vspace{-0.20cm}

\bibitem{FarXYZGMN12-T2} N. Grosjean, J. M. Maillet, G. Niccoli, {\it On form factors of local operators in the $\tau _{2}$-model and the chiral Potts model}, to appear.
\vspace{-0.20cm}

\bibitem{FarXYZN12-0} G. Niccoli, {\it Antiperiodic spin-1/2 XXZ quantum chains by
separation of variables: Form factors and complete spectrum}, arXiv:1205.4537.

\vspace{-0.20cm}

\bibitem{FarXYZN12-1} G. Niccoli, {\it Form factors and complete spectrum of XXX antiperiodic higher spin chains by
quantum separation of variables}, arXiv:1206.2418.

\vspace{-0.20cm}

\bibitem{FarXYZN12-2} G. Niccoli, {\it Non-diagonal open spin-1/2 XXZ quantum chain by separation of variables: Complete spectrum and matrix elements of some
quasi-local operators}, J. Stat. Mech. (2012) P11005.
\vspace{-0.20cm}
\bibitem{FarXYZSF78} E. K. Sklyanin and L. D. Faddeev, Sov. Phys. Dokl. \textbf{23} (1978) 902.

\vspace{-0.20cm}

\bibitem{FarXYZFST80} E. K. Sklyanin and L. A. Takhtajan, L. D. Faddeev, Theor. Math. Phys. \textbf{40} (1980) 688.

\vspace{-0.20cm}

\bibitem{FarXYZFT79}  L. A. Takhtajan, L. D. Faddeev, Russ. Math. Surv. \textbf{34} : 5 (1979) 11.

\vspace{-0.20cm}

\bibitem{FarXYZS79} E. K. Sklyanin, Dokl. Akad. Nauk SSSR \textbf{244} (1979) 1337; Sov. Phys. Dokl. \textbf{24} (1979) 107.

\vspace{-0.20cm}

\bibitem{FarXYZKS79} P. P. Kulish and E. K. Sklyanin, Phys. Lett. A \textbf{70} (1979) 461.

\vspace{-0.20cm}

\bibitem{FarXYZF80} L. D. Faddeev, Sov. Sci. Rev. Math. Cl (1980) 107.

\vspace{-0.20cm}

\bibitem{FarXYZS82} E. K. Sklyanin, J. Sov. Math. \textbf{19} (1982) 1546.

\vspace{-0.20cm}

\bibitem{FarXYZF82} L. D. Faddeev, Les Houches lectures of 1982, Elsevier Sci. Publ. \textbf{563} (1984).

\vspace{-0.20cm}

\bibitem{FarXYZF95} L. D. Faddeev, {\it How Algebraic Bethe Ansatz works for integrable
model}, hep-th/9605187v1.

\vspace{-0.20cm}

\bibitem{FarXYZJ90} M. Jimbo, Adv. Series in Math. Phys. \textbf{10}, Singapore, World Scientific, (1990).

\vspace{-0.20cm}

\bibitem{FarXYZKS82} P. P. Kulish and E. K. Sklyanin, Lect. Notes in Phys. \textbf{151} (1982) 61.


\vspace{-0.20cm}

\bibitem{FarXYZSh85} B. S. Shastry, Lect. Notes in Phys. \textbf{242} (1985).

\vspace{-0.20cm}

\bibitem{FarXYZTh81} H. B. Thacker, Rev. Mod. Phys. \textbf{53} (1982) 253.

\vspace{-0.20cm}

\bibitem{FarXYZIK82} A. G. Izergin and V. E. Korepin, Nucl. Phys. B \textbf{205} (1982) 401.
\vspace{-0.20cm}

\bibitem{FarXYZFelder94} G. Felder. {\it Elliptic Quantum Groups}. In Xlth International Congress of Mathematical Physics, Paris 1994 (D- Ialgonitzer, ed.), International Press (1995) 211.

\vspace{-0.20cm}

\bibitem{FarXYZFelderV96} G. Felder and A. Varchenko, Comm. Math. Phys. \textbf{181} (1996) 741.
\vspace{-0.20cm}

\bibitem{FarXYZBa72-2} R. J. Baxter, Ann. Phys. \textbf{76} (1973) 25.
\vspace{-0.20cm}


\bibitem{FarXYZBa72-1} R. J. Baxter, Ann. Phys. \textbf{76} (1973) 1.

\vspace{-0.20cm}

\bibitem{FarXYZBa72-3} R. J. Baxter, Ann. Phys. \textbf{76} (1973) 48.


\vspace{-0.20cm}

\bibitem{FarXYZFelderVa96} G. Felder and A. Varchenko, Nucl. Phys. B \textbf{480} (1996) 485.

\vspace{-0.20cm}

\bibitem{FarXYZJMO} M. Jimbo, T. Miwa, M. Okado, Comm. Math. Phys. \textbf{116} (1988) 507.
\vspace{-0.20cm}

\bibitem{FarXYZH28} W. Heisenberg, Z. Phys. \textbf{49} (1928) 619.

\vspace{-0.20cm}

\bibitem{FarXYZBe31} H. Bethe, 1931 Z. Phys. \textbf{71} 205.

\vspace{-0.20cm}

\bibitem{FarXYZHul38} L. Hulthen, Ark. Mat. Astron. Fys. \textbf{26} (1938) 1.

\vspace{-0.20cm}

\bibitem{FarXYZOr58} R. Orbach, Phys. Rev. \textbf{112} (1958) 309.

\vspace{-0.20cm}

\bibitem{FarXYZW59} L. R. Walker, Phys. Rev. \textbf{116} (1959) 1089.

\vspace{-0.20cm}

\bibitem{FarXYZYY661} C. N. Yang and C. P. Yang, Phys. Rev. \textbf{150} (1966) 321.

\vspace{-0.20cm}

\bibitem{FarXYZYY662} C. N. Yang and C. P. Yang, Phys. Rev. \textbf{150} (1966) 327.

\vspace{-0.20cm}

\bibitem{FarXYZGa83} M. Gaudin, \textit{La Fonction d'onde de Bethe}, Paris: Masson (1983).

\vspace{-0.20cm}

\bibitem{FarXYZLM66} E. H. Lieb and D. C. Mattis, \textit{Mathematical Physics in One Dimension}, New-York: Academic (1966).



\vspace{-0.20cm}
\bibitem{FarXYZBa72} R. J. Baxter, Ann. Phys. \textbf{70} (1972) 193.

\vspace{-0.20cm}

\bibitem{FarXYZFMcCoy03} K. Fabricius and B. M. McCoy, J. Stat. Phys. \textbf{111} (2003) 323.

\vspace{-0.20cm}

\bibitem{FarXYZ} K. Fabricius and B. M. McCoy, \textit{Functional equations and fusion
matrices for the eight vertex model}, Publ. RIMS, Kyoto Univ. \textbf{40} (2004) 905.

\vspace{-0.20cm}

\bibitem{FarXYZ} K. Fabricius and B. M. McCoy, J. Stat. Phys. \textbf{120} (2005) 37.

\vspace{-0.20cm}

\bibitem{FarXYZ} K. Fabricius, J. Phys. A \textbf{40} (2007) 4075.

\vspace{-0.20cm}

\bibitem{FarXYZ} K. Fabricius, and B. M. McCoy, J. Phys. A \textbf{40} (2007) 14893.

\vspace{-0.20cm}

\bibitem{FarXYZFMcCoy08} K. Fabricius, B. M. McCoy, J. Stat. Phys. \textbf{134} (2009) 643.


\vspace{-0.20cm}

\bibitem{FarXYZBa02} R. J. Baxter, J. Statist. Phys. \textbf{108} (2002) 1.

\vspace{-0.20cm}

\bibitem{FarXYZSlav89} N. A. Slavnov, Theor. Math. Phys. \textbf{79} (1989) 502.

\vspace{-0.20cm}

\bibitem{FarXYZSlav97} N. A. Slavnov, Zap. Nauchn. Semin. POMI \textbf{245} (1997) 270.

\vspace{-0.20cm}

\bibitem{FarXYZKitMT99} N.~Kitanine, J.~M. Maillet, and V.~Terras, Nucl. Phys. B 
\textbf{554} (1999) 647.



\vspace{-0.20cm}

\bibitem{FarXYZSk1} E. K.~Sklyanin, Lect. Notes Phys. \textbf{226} (1985)
196; E. K.~Sklyanin J. Sov. Math. \textbf{31} (1985) 3417.

\vspace{-0.20cm}

\bibitem{FarXYZSk2} E. K.~Sklyanin, \emph{Quantum inverse scattering method.
Selected topics}. In: Quantum groups and quantum integrable systems
(World Scientific, 1992) 63, arXiv:hep-th/9211111v1.

\vspace{-0.20cm}

\bibitem{FarXYZSk3} E. K.~Sklyanin, Prog. Theor. Phys. Suppl. \textbf{118} (1995) 35.
\vspace{-0.20cm}

\bibitem{FarXYZDJMO} E. Date, M. Jimbo, A. Kuniba, T. Miwa, M. Okado. \textit{Exactly
Solvable SOS-Models II: Proof of the Star-Triangle-Relation and
Combinatorial Identities}, in Adv. Stud. in Pure Math. \textbf{16} M.
Jimbo, T. Miwa, A. Tsuchiya ed. Academic Press (1988) 17.

\vspace{-0.20cm}
\bibitem{FarXYZFelder-Schorr-99} G. Felder and A. Schorr, J. Phys. A: Math. Gen. {\bf 32} (1999) 8001.
\vspace{-0.20cm}
\bibitem{FarXYZSchorr} A. Schorr, {\it Separation of variables for the eight-vertex SOS model with
antiperiodic boundary conditions}, http://e-collection.library.ethz.ch/eserv/eth:23538/eth-23538-02.pdf.


\vspace{-0.20cm}

\bibitem{FarXYZMaiT00} J.~M. Maillet and V.~Terras, Nucl. Phys. B \textbf{575} (2000) 627.


\vspace{-0.20cm}

\bibitem{FarXYZIKitMT99} A. G. Izergin, N.~Kitanine, J.~M. Maillet, and V.~Terras, Nucl. Phys. B \textbf{554} (1999) 679.

\vspace{-0.20cm}

\bibitem{FarXYZKitMT00} N.~Kitanine, J.~M. Maillet, and V.~Terras, Nucl. Phys. B \textbf{567} (2000) 544.

\vspace{-0.20cm}

\bibitem{FarXYZKitMST02a} N.~Kitanine, J.~M. Maillet, N.~A. Slavnov, and V.~Terras, Nucl. Phys. B \textbf{641} (2002) 487.

\vspace{-0.20cm}

\bibitem{FarXYZKitMST02b} N.~Kitanine, J.~M. Maillet, N.~A. Slavnov, and V.~Terras, Nucl. Phys. B \textbf{642} (2002) 433.

\vspace{-0.20cm}

\bibitem{FarXYZKitMST02c} N.~Kitanine, J.~M. Maillet, N.~A. Slavnov, and V.~Terras, J. Phys. A \textbf{35} (2002) L385 .

\vspace{-0.20cm}

\bibitem{FarXYZKitMST02d} N.~Kitanine, J.~M. Maillet, N.~A. Slavnov, and V.~Terras, J. Phys. A \textbf{35} (2002) L753.

\vspace{-0.20cm}

\bibitem{FarXYZKitMST04a} N.~Kitanine, J.~M. Maillet, N.~A. Slavnov, and V.~Terras, Nucl. Phys. B \textbf{712} (2005) 600.

\vspace{-0.20cm}

\bibitem{FarXYZKitMST04b} N.~Kitanine, J.~M. Maillet, N.~A. Slavnov, and V.~Terras, Nucl. Phys. B \textbf{729} (2005) 558.

\vspace{-0.20cm}

\bibitem{FarXYZKitMST04c} N.~Kitanine, J.~M. Maillet, N.~A. Slavnov, and V.~Terras, J. Phys. A \textbf{38} (2005) 7441.

\vspace{-0.20cm}

\bibitem{FarXYZKitMST05a} N.~Kitanine, J.~M. Maillet, N.~A. Slavnov, and V.~Terras, J. Stat. Mech. L09002 (2005).

\vspace{-0.20cm}

\bibitem{FarXYZKitMST05b} N.~Kitanine, J.~M. Maillet, N.~A. Slavnov, and V.~Terras, \textit{On the algebraic Bethe Ansatz approach to the correlation functions of the $XXZ$ spin-1/2 Heisenberg chain}, In Recent Progress in Solvable lattice Models, RIMS Sciences Project Research 2004 on \textit{Method of Algebraic Analysis in Integrable Systems}, RIMS, Kyoto, Kokyuroku, \textbf{1480} (2006) 14; hep-th/0505006.

\vspace{-0.20cm}

\bibitem{FarXYZKKMST07} N. Kitanine, K. Kozlowski, J. M. Maillet, N. A. Slavnov, V. Terras, J. Stat. Mech. P01022 (2007).
\vspace{-0.20cm}

\bibitem{FarXYZKKMNST07} N. Kitanine, K. K. Kozlowski, J. M. Maillet, G. Niccoli, N. A. Slavnov, V. Terras , J. Stat. Mech. (2007) P10009.

\vspace{-0.20cm}

\bibitem{FarXYZK08} K. K. Kozlowski, J. Stat. Mech. (2008) P02006.

\vspace{-0.20cm}

\bibitem{FarXYZKKMNST08} N. Kitanine, K. Kozlowski, J. M. Maillet, G. Niccoli, N. A. Slavnov, V. Terras, J. Stat. Mech. (2008) P07010.
\vspace{-0.20cm}

\bibitem{FarXYZBaxBook} R. J. Baxter, {\it Exactly Solved Models in Statistical
Mechanics}, Academic Press, New York (1982).

\vspace{-0.20cm}

\bibitem{FarXYZABBBQ87} F. C. Alcaraz, M. N. Barber, M. T. Batchelor, R. J. Baxter and
G. R. W. Quispel, J. Phys. A {\bf20} (1987) 6397.

\vspace{-0.20cm}

\bibitem{FarXYZRe83-1} N. Yu Reshetikhin, Lett. Math. Phys. \textbf{7} (1983) 205.

\vspace{-0.20cm}

\bibitem{FarXYZRe83-2} N. Yu Reshetikhin, Sov. Phys. JETP \textbf{57} (1983) 691.


\vspace{-0.20cm}

\bibitem{FarXYZBBOY95} M. T. Batchelor, R. J. Baxter, M. J. O'Rourke, C. M. Yung, J. Phys. A \textbf{28} (1995) 2759.

\vspace{-0.20cm}

\bibitem{FarXYZNWF09} S. Niekamp, T. Wirth, H. Frahm, J. Phys. A \textbf{42} (2009) 195008.


\vspace{-0.20cm}

\bibitem{FarXYZG08}  W. Galleas, Nucl. Phys. B {\bf790} (2008) 524.
\vspace{-0.20cm}

\bibitem{FarXYZK01} N. Kitanine, J. Phys. A: Math. Gen. \textbf{34} (2001) 8151.

\vspace{-0.20cm}

\bibitem{FarXYZCM07} O. A. Castro-Alvaredo, J. M. Maillet, J. Phys. A \textbf{40} (2007) 7451.

\vspace{-0.20cm}

\bibitem{FarXYZBS90} V. V. Bazhanov, Yu G. Stroganov, J. Stat. Phys. \textbf{59}
(1990) 799.

\vspace{-0.20cm}

\bibitem{FarXYZBBP90} R. J. Baxter, V. V. Bazhanov and J. H. H. Perk, Int. J. Mod.
Phys. \textbf{B4} (1990) 803.

\vspace{-0.20cm}

\bibitem{FarXYZBa89} R. J. Baxter, J. Stat. Phys. \textbf{57} (1989) 1.

\vspace{-0.20cm}

\bibitem{FarXYZAMcP0} G. Albertini, B. M. McCoy and J. H. H. Perk, Adv. Study in Pure
Math. \textbf{19} (1989) 1.

\vspace{-0.20cm}

\bibitem{FarXYZAMcP1} G. Albertini, B. M. McCoy and J. H. H. Perk, Phys. Lett. \textbf{A 135} (1989) 159.

\vspace{-0.20cm}

\bibitem{FarXYZAMcP2} G. Albertini, B. M. McCoy and J. H. H. Perk, Phys. Lett. A \textbf{139} (1989 ) 204.

\vspace{-0.20cm}

\bibitem{FarXYZBaPauY} R. J. Baxter, J. H. H. Perk and H. Au-Yang, Phys. Lett. A \textbf{128} (1988) 138; H. Au-Yang and J. H. H. Perk, {\it Onsager's star triangle equation: Master key to the integrability}, in Adv. Stud. in Pure Math. {\bf 19} (1989) Kinokuniya-Academic.

\vspace{-0.20cm}

\bibitem{FarXYZBa89-1} R. J. Baxter, Phys. Lett. A \textbf{133} (1989) 185.



\vspace{-0.20cm}

\bibitem{FarXYZauYMcPTY} H. Au-Yang, B. M. McCoy, J. H. H. Perk, S. Tang, and M.
Yan, Phys. Lett. A \textbf{123} (1987) 219.

\vspace{-0.20cm}

\bibitem{FarXYZMcPTS} B. M. McCoy, J. H. H. Perk, S. Tang, and C. H. Sah, Phys. Lett. A \textbf{125} (1987) 9.

\vspace{-0.20cm}

\bibitem{FarXYZauYMcPT} H. Au-Yang, B. M. McCoy, J. H. H. Perk, and S. Tang, {\it Algebraic Analysis} {\bf1} M. Kashiwara and T. Kawai eds Academic Press, New York (1988).

\vspace{-0.20cm}

\bibitem{FarXYZTarasovSChP} V. O. Tarasov Phys. Lett. A \textbf{147} (1990) 487.
\vspace{-0.20cm}

\bibitem{FarXYZBa04} R. J. Baxter, J. Stat. Phys. \textbf{117} (2004) 1.

\vspace{-0.20cm}

\bibitem{FarXYZNT-10} G. Niccoli and J. Teschner, J. Stat. Mech. (2010) P09014.

\vspace{-0.20cm}

\bibitem{FarXYZN-10} G. Niccoli,\ Nucl. Phys. B \textbf{835} (2010) 263.

\vspace{-0.20cm}

\bibitem{FarXYZN-11} G. Niccoli, JHEP (2011) 1103:123.

\vspace{-0.20cm}

\bibitem{FarXYZGIPS06} G. von Gehlen, N. Iorgov, S. Pakuliak and V. Shadura, J.
Phys. A: Math. Gen. \textbf{39} (2006) 7257.

\vspace{-0.20cm}

\bibitem{FarXYZGIPST07} G. von Gehlen, N. Iorgov, S. Pakuliak, V. Shadura and Yu
Tykhyy, J. Phys. A: Math. Theor. \textbf{40} (2007) 14117.

\vspace{-0.20cm}

\bibitem{FarXYZGIPST08} G. von Gehlen, N. Iorgov, S. Pakuliak, V. Shadura and Yu
Tykhyy, J. Phys. A: Math. Theor. \textbf{41} (2008) 095003.

\vspace{-0.20cm}

\bibitem{FarXYZGIPS09} G von Gehlen, N Iorgov, S Pakuliak, V Shadura, J. Phys. A:
Math. Theor. \textbf{42} (2009) 304026.

\vspace{-0.20cm}

\bibitem{FarXYZGN12} N. Grosjean and G. Niccoli, {\it The $\tau _{2}$-model and the
chiral Potts model revisited: Completeness of Bethe equations originated
from Sklyanin SOV}, J. Stat. Mech. (2012) P10025.

\vspace{-0.20cm}

\bibitem{FarXYZGau71} M. Gaudin , Phys. Rev. A \textbf{4} (1971) 386.

\vspace{-0.20cm}

\bibitem{FarXYZSkly88} E. K. Sklyanin, J. Phys. A: Math. Gen. \textbf{21} (1988) 2375.

\vspace{-0.20cm}

\bibitem{FarXYZ} I. V. Cherednik, Theor. Math. Phys. \textbf{61} (1984) 977.

\vspace{-0.20cm}

\bibitem{FarXYZ} P. P. Kulish and E. K. Sklyanin, J. Phys. A: Math. Gen. \textbf{24}
(1991) L435.

\vspace{-0.20cm}

\bibitem{FarXYZ} L. Mezincescu and R. Nepomechie, Int. J. Mod. Phys. A \textbf{6}
(1991) 5231.

\vspace{-0.20cm}

\bibitem{FarXYZ} P. P. Kulish and E. K. Sklyanin, J. Phys. A: Math. Gen. \textbf{25}
(1992) 5963.

\vspace{-0.20cm}

\bibitem{FarXYZ} S. Ghoshal and A. B. Zamolodchikov, Int. J. Mod. Phys. A \textbf{9}
(1994) 3841.

\vspace{-0.20cm}

\bibitem{FarXYZGZ94} S. Ghoshal and A. B. Zamolodchikov, Int. J. Mod. Phys. A \textbf{9}
(1994) 4353.

\vspace{-0.20cm}

\bibitem{FarXYZSm98} F. Smirnov, J. Phys. A: Math. Gen. {\bf31} (1998) 8953.

\vspace{-0.20cm}

\bibitem{FarXYZOB-04} O. Babelon, J. Phys. A \textbf{37} (2004) 303.

\vspace{-0.20cm}

\bibitem{FarXYZJMS11-03} M. Jimbo, T. Miwa, F. Smirnov, Lett. Math. Phys. \textbf{96} (2011) 325.

\vspace{-0.20cm}
\bibitem{FarXYZJMS09-02} M. Jimbo, T. Miwa, F. Smirnov, J. Phys. A {\bf42} (2009) 304018.

\vspace{-0.20cm}

\bibitem{FarXYZBJMS09} H. Boos, M. Jimbo, T. Miwa, F. Smirnov, {\it Hidden Grassmann Structure in the XXZ Model IV: CFT limit}, arXiv:0911.3731.

\vspace{-0.20cm}

\bibitem{FarXYZJMS11-01} M. Jimbo, T. Miwa, F. Smirnov, Nucl. Phys. B \textbf{852} (2011) 390.

\vspace{-0.20cm}

\bibitem{FarXYZJMS11-02} M. Jimbo, T. Miwa, F. Smirnov, {\it Fermionic screening operators in the sine-Gordon model}, arXiv:1103.1534.
\vspace{-0.20cm}

\bibitem{FarXYZA.Zam77} A. B. Zamolodchikov, Pis. Zh. Eksp. Teor. Fiz. {\bf 25} (1977) 499; A. B. Zamolodchikov, Comm. Math. Phys. {\it 55} (1977) 183.

\vspace{-0.20cm}

\bibitem{FarXYZZAZA79} A. B. Zamolodchikov, Al. B. Zamolodchikov, Ann. Phys. {\bf 120}
(1979) 253.

\vspace{-0.20cm}

\bibitem{FarXYZKT77} M. Karowski and H. J. Thun, Nucl. Phys. B \textbf{130} (1977) 295.

\vspace{-0.20cm}

\bibitem{FarXYZK80} V. E. Korepin, Comm. Math. Phys. \textbf{76} (1980) 165.

\vspace{-0.20cm}

\bibitem{FarXYZM92} G. Mussardo, Phys. Rep. \textbf{218} (1992) 215.

\vspace{-0.20cm}

\bibitem{FarXYZBBS96} O. Babelon, D. Bernard, F. Smirnov, Comm. Math. Phys. \textbf{182} (1996) 319.

\vspace{-0.20cm}

\bibitem{FarXYZBBS97} O. Babelon, D. Bernard, F. Smirnov, Comm. Math. Phys. \textbf{186} (1997) 601.


\vspace{-0.20cm}

\bibitem{FarXYZZam88} A. B. Zamolodchikov, Int. J. Mod. Phys. A \textbf{3} (1988)
743.

\vspace{-0.20cm}

\bibitem{FarXYZZam89} A. B. Zamolodchikov, Adv. Stud. Pure Math. \textbf{19} (1989) 641.

\vspace{-0.20cm}

\bibitem{FarXYZAlZam91} Al. B. Zamolodchikov, Nucl. Phys. B \textbf{348} (1991) 619.

\vspace{-0.20cm}

\bibitem{FarXYZGM96} R. Guida, N. Magnoli, Nucl. Phys. B \textbf{471} (1996) 361.
\vspace{-0.20cm}

\bibitem{FarXYZVi70} M. A. Virasoro, Phys. Rev. D \textbf{1} (1970) 2933.

\vspace{-0.20cm}

\bibitem{FarXYZBPZ84} A. A. Belavin, A. M. Polyakov, A. B. Zamolodchikov, Nucl. Phys.
B \textbf{241} (1984) 333.

\vspace{-0.20cm}

\bibitem{FarXYZGin89} P. Ginsparg, {\it Applied Conformal Field
Theory} , in: Fields, Strings and Critical Phenomena, Les
Houches Lecture Notes 1988, eds. E. Br\'{e}zin and J. Zinn-Justin, Elsevier,
New York (1989).

\vspace{-0.20cm}

\bibitem{FarXYZCa88} J. L. Cardy, {\it Conformal Invariance and Statistical Mechanics}, in
Fields, Strings and Critical Phenomena, ed. E. Br\'ezin and J. Zinn-Justin, Les Houches (1988), Session XLIX North-Holland,
Amsterdam (1990) 169.

\vspace{-0.20cm}

\bibitem{FarXYZDFMS97} P. Di Francesco, P. Mathieu and D. S\'{e}n\'{e}chal,
{\it Conformal Field Theory}, Springer, New York (1997).

\vspace{-0.20cm}

\bibitem{FarXYZKW78} M. Karowski, P. Weisz, Nucl. Phys. B \textbf{139} (1978) 455.

\vspace{-0.20cm}

\bibitem{FarXYZSm88} F. A. Smirnov, Physica A \textbf{3} (1988) 743.

\vspace{-0.20cm}

\bibitem{FarXYZSm92} F. A. Smirnov, {\it Form Factors in Completely Integrable Models of
Quantum Field Theory}, World Scientific, (1992).

\vspace{-0.20cm}

\bibitem{FarXYZSm90} F. A. Smirnov, Nucl. Phys. B \textbf{337} (1990) 156.

\vspace{-0.20cm}

\bibitem{FarXYZFMS93-1} A. Fring, G. Mussardo and P. Simonetti, Nucl. Phys. B \textbf{393} (1993) 413.

\vspace{-0.20cm}

\bibitem{FarXYZFMS93-2} A. Fring, G. Mussardo and P. Simonetti, Phys. Lett. B \textbf{307} (1993) 83.

\vspace{-0.20cm}

\bibitem{FarXYZKM93} A. Koubek and G. Mussardo, Phys. Lett. B \textbf{311} (1993) 193.

\vspace{-0.20cm}

\bibitem{FarXYZMS94} G. Mussardo and P. Simonetti, Int. J. Mod. Phys. A \textbf{9} (1994) 3307.

\vspace{-0.20cm}

\bibitem{FarXYZKub94} A. Koubek, Nucl. Phys. B \textbf{428} (1994) 655.

\vspace{-0.20cm}

\bibitem{FarXYZDM95} G. Delfino and G. Mussardo, Nucl. Phys. B \textbf{455} (1995) 724.

\vspace{-0.20cm}

\bibitem{FarXYZDSC96} G. Delfino, P. Simonetti and J. L. Cardy, Phys. Lett. B \textbf{387} (1996) 327.

\vspace{-0.20cm}

\bibitem{FarXYZBFKZ99} H. Babujian, A. Fring, M. Karowski  and A. Zapletal, Nucl. Phys. B \textbf{538} (1999) 535.

\vspace{-0.20cm}

\bibitem{FarXYZLZ01} S. Lukyanov and A. Zamolodchikov, Nucl. Phys. B \textbf{607} (2001) 437.

\vspace{-0.20cm}

\bibitem{FarXYZBK02} H. Babujian and M. Karowski, Nucl. Phys. B \textbf{620} (2002) 407.

\vspace{-0.20cm}

\bibitem{FarXYZBK02+} H. Babujian and M. Karowski, J. Phys. A: Math. Gen. \textbf{35} (2002) 9081.

\vspace{-0.20cm}

\bibitem{FarXYZD04} G. Delfino, J. Phys. A: Math. Gen. \textbf{37} (2004) R45.


\vspace{-0.20cm}

\bibitem{FarXYZCM90} J. L. Cardy and G. Mussardo, Nucl. Phys. B \textbf{340} (1990) 387.

\vspace{-0.20cm}

\bibitem{FarXYZKub95} A. Koubek, Nucl. Phys. B \textbf{435} (1995) 703.

\vspace{-0.20cm}

\bibitem{FarXYZSm96} F. Smirnov, Nucl. Phys. B \textbf{453} (1995) 807.

\vspace{-0.20cm}

\bibitem{FarXYZJMT03} M. Jimbo, T. Miwa, Y. Takeyama, {\it Counting minimal form
factors of the restricted sine-Gordon model}, arXiv:math-ph/0303059v6.
\vspace{-0.20cm}

\bibitem{FarXYZDN05-1} G. Delfino and G. Niccoli, Nucl. Phys. B \textbf{707} (2005) 381.

\vspace{-0.20cm}

\bibitem{FarXYZDN05-2} G. Delfino and G. Niccoli{, }J. Stat. Mech. (2005) P04004.
\vspace{-0.20cm}

\bibitem{FarXYZDN05-3} G. Niccoli, {\it Descendant Operators in Massive Integrable Quantum Field Theories}, Ph.D. Thesis, SISSA (2005).
\vspace{-0.20cm}

\bibitem{FarXYZDN06} G. Delfino and G. Niccoli, JHEP 05 (2006) 035.

\vspace{-0.20cm}

\bibitem{FarXYZDN08} G. Delfino and G. Niccoli, Nucl.\ Phys.{\ }B \textbf{799} (2008) 364.

\vspace{-0.20cm}

\bibitem{FarXYZD09} G. Delfino, Nucl.\ Phys.{\ }B \textbf{807} (2009) 455.
\vspace{-0.20cm}

\bibitem{FarXYZIK81} A. G. Izergin and V. E. Korepin, Dokl. Akad. Nauk SSSR {\bf259} (1981) 76.

\vspace{-0.20cm}

\bibitem{FarXYZIK09} A. G. Izergin and V. E. Korepin, \textit{A lattice model related to the nonlinear Schroedinger equation}, arXiv:0910.0295.
\vspace{-0.20cm}

\bibitem{FarXYZTables of integrals} I.S. Gradshteyn and I.M. Ryzhik; Alan Jeffrey,
Daniel Zwillinger, editors. {\it Table of Integrals, Series, and Products},
seventh edition. Academic Press (2007).
\vspace{-0.20cm}
\bibitem{FarXYZPRL-08} S. Pakuliak, V. Rubtsov and A. Silantyev, J. Phys. A: Math. Theor. {\bf41} (2008) 295204.
\vspace{-0.20cm}
\bibitem{FarXYZ?NT12} D. Levy-Bencheton, G. Niccoli and V. Terras, {\it Antiperiodic dynamical 6-vertex model II: Form factors by separation of variables}, to appear.
\vspace{-0.20cm}
\bibitem{FarXYZ?Terras12} D. Levy-Bencheton, V. Terras, private communications.
\vspace{-0.20cm}

\bibitem{FarXYZCM05} J.-S. Caux, J. M. Maillet, Phys. Rev. Lett. \textbf{95} (2005) 077201.

\vspace{-0.20cm}

\bibitem{FarXYZCHM05} J.-S. Caux, R. Hagemans, J. M. Maillet, J. Stat. Mech. (2005) P09003.

\vspace{-0.20cm}

\bibitem{FarXYZPSCHMWA06} R. G. Pereira, J. Sirker, J.-S. Caux, R. Hagemans, J. M. Maillet, S. R. White, I. Affleck, Phys. Rev. Lett. \textbf{96} (2006) 257202.

\vspace{-0.20cm}

\bibitem{FarXYZHCM06} R. Hagemans, J.-S. Caux, J. M. Maillet, Proceedings of the "Tenth Training Course in the Physics of Correlated Electron Systems and High-Tc Superconductors", Salerno, Oct 2005, AIP Conference Proceedings {\bf 846} (2006) 245.

\vspace{-0.20cm}

\bibitem{FarXYZPSCHMWA07} R. G. Pereira, J. Sirker, J.-S. Caux, R. Hagemans, J. M. Maillet, S. R. White, I. Affleck, J. Stat. Mech. (2007) P08022.

\vspace{-0.20cm}

\bibitem{FarXYZSPCHMWA07} J. Sirker, R. G. Pereira, J.-S. Caux, R. Hagemans, J. M. Maillet, S. R. White, I. Affleck, Physica B \textbf{403} (2008) 1520.

\vspace{-0.20cm}

\bibitem{FarXYZCCS07} J. S. Caux, P. Calabrese and N. A. Slavnov, J. Stat. Mech. (2007) P01008.


\vspace{-0.20cm}

\bibitem{FarXYZBloch36} F. Bloch, Phys. Rev. \textbf{50} (1936) 259.

\vspace{-0.20cm}

\bibitem{FarXYZ} J. S. Schwinger, Phys. Rev. \textbf{51} (1937) 544.

\vspace{-0.20cm}

\bibitem{FarXYZ} O. Halpern and M. H. Johnson, Phys. Rev. \textbf{55} (1938) 898.

\vspace{-0.20cm}

\bibitem{FarXYZ} L. Van Hove, Phys. Rev. $\mathbf{95}$ (1954) 249.

\vspace{-0.20cm}

\bibitem{FarXYZ} L. Van Hove, Phys. Rev. \textbf{95} (1954) 1374.

\vspace{-0.20cm}

\bibitem{FarXYZ} W. Marshall and S. W. Lovesey, {\it Theory of Thermal Neutron Scattering}, Oxford: Academic (1971).

\vspace{-0.20cm}

\bibitem{FarXYZBalescu75} R. Balescu, {\it Equilibrium and Nonequilibrium Statistical Mechanics}, New York: Wiley (1975). 


\end{thebibliography}
\end{document}